\documentclass[prodmode,acmjacm]{acmsmall}

\usepackage{hyperref}
\usepackage{appendix}

\usepackage{amsmath,latexsym,amssymb,revsymb}

\newcommand{\extra}[1]{}

\newcommand{\comment}[1]{}

\newdef{myremark}{Remark}

\newdef{claim}{Claim}

\newenvironment{proof+}[1]{\noindent \textbf{{Proof #1~} }}{\qed\medskip}

\newcommand{\norm}[1]{\|{ #1 }\|_2}

\newcommand{\nc}{\newcommand}
\nc{\rnc}{\renewcommand}
\nc{\beq}{\begin{equation}}
\nc{\eeq}{{\end{equation}}}
\nc{\beqa}{\begin{eqnarray}}
\nc{\eeqa}{\end{eqnarray}}
\nc{\lbar}[1]{\overline{#1}}
\nc{\bra}[1]{\langle#1|}
\nc{\ket}[1]{|#1\rangle}
\nc{\ketbra}[2]{|#1\rangle\!\langle#2|}
\nc{\braket}[2]{\langle#1|#2\rangle}
\nc{\proj}[1]{| #1\rangle\!\langle #1 |}
\nc{\avg}[1]{\langle#1\rangle}
\nc{\smfrac}[2]{\mbox{$\frac{#1}{#2}$}}
\nc{\tr}{\operatorname{tr}}
\nc{\tracedist}[1]{\Delta_{}\!\left( #1 \right)}
\nc{\fid}[1]{F\!\left( #1 \right)}

\nc{\ox}{\otimes}
\nc{\dg}{\dagger}
\nc{\dn}{\downarrow}
\nc{\cA}{{\cal A}}
\nc{\cB}{{\cal B}}
\nc{\cC}{{\cal C}}
\nc{\cD}{{\cal D}}
\nc{\cE}{{\mathcal E}}
\nc{\cF}{{\cal F}}
\nc{\cG}{{\cal G}}
\nc{\cH}{{\cal H}}
\nc{\cI}{{\cal I}}
\nc{\cJ}{{\cal J}}
\nc{\cK}{{\cal K}}
\nc{\cL}{{\cal L}}
\nc{\cM}{{\cal M}}
\nc{\cN}{{\cal N}}
\nc{\cO}{{\cal O}}
\nc{\cP}{{\cal P}}
\nc{\cR}{{\cal R}}
\nc{\cS}{{\cal S}}
\nc{\cT}{{\cal T}}
\nc{\cU}{{\cal U}}
\nc{\cX}{{\cal X}}
\nc{\cZ}{{\cal Z}}

\nc{\entI}{{\bf I}}
\nc{\entIarrow}{{\bf I}^{\leftarrow}}
\nc{\entH}{{\bf H}}
\nc{\entS}{{\bf S}}
\nc{\entHmin}{\mathbf{H}_{\min}}

\nc{\entF}{{\bf E}_f}

\nc{\isom}{\simeq}

\nc{\rank}{\operatorname{rank}}
\nc{\rar}{\rightarrow}
\nc{\lrar}{\longrightarrow}
\nc{\polylog}{\operatorname{polylog}}
\nc{\poly}{\operatorname{poly}}
\nc{\1}{{\openone}}

\nc{\weight}{\textbf{w}}
\nc{\hamdist}{d_{H}}

\def\e{\epsilon}

\def\ph{\varphi}

\nc{\Sp}{{{\mathbb S}}}
\nc{\RR}{{{\mathbb R}}}
\nc{\CC}{{{\mathbb C}}}
\nc{\FF}{{{\mathbb F}}}
\nc{\NN}{{{\mathbb N}}}
\nc{\ZZ}{{{\mathbb Z}}}
\nc{\PP}{{{\mathbb P}}}
\nc{\QQ}{{{\mathbb Q}}}
\nc{\UU}{{{\mathbb U}}}
\nc{\OO}{{{\mathbb O}}}
\nc{\EE}{{{\mathbb E}}}
\nc{\id}{{\operatorname{id}}}

\nc{\qubitchannel}{\id_2}
\nc{\bitchannel}{\overline{\id}_2}

\nc{\re}{\mathbf{Re}}
\nc{\im}{\mathbf{Im}}

\nc{\be}{\begin{equation}}
\nc{\ee}{{\end{equation}}}
\nc{\bea}{\begin{eqnarray}}
\nc{\eea}{\end{eqnarray}}
\nc{\<}{\langle}
\rnc{\>}{\rangle}
\nc{\Hom}[2]{\mbox{Hom}(\CC^{#1},\CC^{#2})}
\nc{\rU}{\mbox{U}}

\nc{\ob}[1]{#1}

\newcommand{\eqdef}	{:=}

\newcommand{\ex}[1]	{\mathbf{E}\left\{ #1 \right\}}
\newcommand{\pr}[1]	{\mathbf{P}\left\{ #1 \right\}}

\newcommand{\event}[1]	{\left[ #1 \right]}

\renewcommand{\exp}[1]	{\operatorname{exp}\left( #1 \right)}

\newcommand{\floor}[1]	{\left\lfloor #1 \right\rfloor}
\newcommand{\ceil}[1]	{\left\lceil #1 \right\rceil}

\newcommand{\comp}{\circ}
\newcommand{\concat}{\cdot}
\nc{\unif}{\textrm{unif}}

\nc{\binent}{h_2}

\nc{\red}[1]{{ #1}}

\begin{document}

\title{From Low-Distortion Norm Embeddings to Explicit Uncertainty Relations and Efficient Information Locking
}

\author{OMAR FAWZI and PATRICK HAYDEN \affil{McGill University}
PRANAB SEN \affil{Tata Institute of Fundamental Research}
}

\date{\today}




\begin{abstract}
The existence of quantum uncertainty relations is the essential reason that some classically unrealizable cryptographic primitives become realizable when quantum communication is allowed. One operational manifestation of these uncertainty relations is a purely quantum effect referred to as \emph{information locking} \cite{DHLST04}. A locking scheme can be viewed as a cryptographic protocol in which a uniformly random $n$-bit message is encoded in a quantum system using a classical key of size much smaller than $n$. Without the key, no measurement of this quantum state can extract more than a negligible amount of information about the message, in which case the message is said to be  ``locked''. Furthermore, knowing the key, it is possible to recover, that is ``unlock'', the message.

In this paper, we make the following contributions by exploiting a connection between uncertainty relations and low-distortion embeddings of Euclidean spaces into slightly larger spaces endowed with the $\ell_1$ norm.
We introduce the notion of a \emph{metric uncertainty relation} and connect it to low-distortion embeddings of $\ell_2$ into $\ell_1$. A metric uncertainty relation also implies an entropic uncertainty relation.
We prove that random bases satisfy uncertainty relations with a stronger definition and better parameters than previously known. Our proof is also considerably simpler than earlier proofs. We then apply this result to show the existence of locking schemes with key size independent of the message length. 
Moreover, we give \emph{efficient} constructions of bases satisfying metric uncertainty relations. The bases defining these metric uncertainty relations are computable by quantum circuits of almost linear size. This leads to the first explicit construction of a strong information locking scheme. 
These constructions are obtained by adapting an explicit norm embedding due to \citeN{Ind07} and an extractor construction of \citeN{GUV09}.
We apply our metric uncertainty relations to exhibit communication protocols that perform equality testing of $n$-qubit states. We prove that this task can be performed by a single message protocol using $O(\log^2 n)$ qubits and $n$ bits of communication, where the computation of the sender is efficient. 
\end{abstract}

\begin{bottomstuff}
A preliminary version of this paper appeared in the Proceedings of the 43rd ACM Symposium on Theory of Computing (STOC 2011), pp. 773--782 and was presented at the Workshop on Quantum Information Processing (QIP 2011).

This research was supported by the Canada Research Chairs program, the Perimeter Institute, CIFAR, FQRNT's INTRIQ, MITACS, NSERC, ONR through grant N000140811249 and QuantumWorks. 

Author's addresses: O. Fawzi, (Current address) Institute for Theoretical Physics, ETH Z\"{u}rich, Switzerland, e-mail: ofawzi@phys.ethz.ch; P. Hayden, School of Computer Science, McGill University, Montr\'eal, Qu\'ebec, Canada, e-mail: patrick@cs.mcgill.ca; P. Sen, School of Technology and Computer Science, Tata Institute of Fundamental Research, Mumbai, India, e-mail: pgdsen@tcs.tifr.res.in\end{bottomstuff}

\maketitle



\section{Introduction}

Uncertainty relations express the fundamental incompatibility of certain measurements in quantum mechanics \cite{Hei27,Rob29}. \red{They quantify the fact that noncommuting quantum mechanical observables cannot simultaneously have definite values.} Far from just being puzzling constraints on our ability to know the state of a quantum system, uncertainty relations are at the heart of why some classically unrealizable cryptographic primitives become realizable when quantum communication is allowed. For example, so-called \emph{entropic} uncertainty relations introduced in \cite{BM75,Deu83} are the main ingredients for modern security proofs for quantum key distribution \cite{TR11,TLGR12} and for secure computation in the bounded and noisy quantum storage models \cite{DFSS05,DFRSS07,KWW09}. A simple example of an entropic uncertainty relation was given by \citeN{MU88}. Let $\cB_{+}$ denote a ``rectilinear'' or computational basis of $\CC^2$ and $\cB_{\times}$ be a ``diagonal'' or Hadamard basis and let $\cB_{+^n}$ and $\cB_{\times^n}$ be the corresponding bases obtained on the tensor product space $(\CC^2)^{\otimes n}$. All vectors in the rectilinear basis $\cB_{+^n}$ have an inner product with all vectors in the diagonal basis $\cB_{\times^n}$ upper bounded by $2^{-n/2}$ in absolute value. The uncertainty relation of \citeN{MU88} states  that for \emph{any} quantum state on $n$ qubits described by a unit vector $\ket{\psi} \in (\CC^2)^{\otimes n}$, the average measurement entropy satisfies
\begin{equation}
\label{eq:entropic-ur-n}
\frac{1}{2} \left( \entH(p_{\cB_{+^n} , \ket{\psi}}) + \entH(p_{\cB_{\times^n}, \ket{\psi}}) \right) \geq \frac{n}{2},
\end{equation}
where $p_{\cB , \ket{\psi}}$ denotes the outcome probability distribution when $\ket{\psi}$ is measured in basis $\cB$ and $\entH$ denotes the Shannon entropy. Equation \eqref{eq:entropic-ur-n} expresses the fact that measuring in a random basis $\cB_K$, where $K \in_u \{+^n,\times^n\}$ is uniformly chosen from the set $\{+^n, \times^n\}$, produces an outcome that has some uncertainty irrespective of the state being measured.

A surprising application of entropic uncertainty relations is the effect known as \emph{information locking} \cite{DHLST04} (see also \cite{Leu09}). Suppose Alice holds a uniformly distributed random $n$-bit string $X$. She chooses a random basis $K \in_u \{+^n,\times^n\}$ and encodes $X$ in the basis $\cB_K$. This random quantum state $\cE(X,K)$ is then given to Bob. How much information about $X$ can Bob, who does not know $K$, extract from this quantum system via a measurement? To better appreciate the quantum case, observe that if $X$ were encoded in a classical state $\cE_c(X,K)$, then $\cE_c(X,K)$ would ``hide'' at most one bit about $X$; more precisely, the mutual information between $X$ and $\cE_c(X,K)$ is at least $n-1$. For the quantum encoding $\cE$, one can show that for \emph{any measurement} that Bob applies on $\cE(X,K)$ whose outcome is denoted $I$, the mutual information between $X$ and $I$ is at most $n/2$ \cite{DHLST04}. 
The $n/2$ missing bits of information about $X$ are said to be \emph{locked} in the quantum state $\cE(X,K)$. If Bob had access to $K$, then $X$ can be easily obtained from $\cE(X,K)$: The one-bit key $K$ can be used to \emph{unlock} $n/2$ bits about $X$.

A natural question is whether it is possible to lock more than $n/2$ bits in this way. In order to achieve this, the key $K$ has to be chosen from a larger set. In terms of uncertainty relations, this means that we need to consider $t > 2$ bases to achieve an average measurement entropy larger than $n/2$ (equation \eqref{eq:entropic-ur-n}). In this case, the natural candidate is a set of $t$ \emph{mutually unbiased bases}, the defining property of which is a small inner product between any pair of vectors in different bases. Surprisingly, it was shown by \citeN{BW07} and \citeN{Amb09} that there are up to $t = 2^{n/2}$ mutually unbiased bases $\{ \cB_0, \cB_1, \dots, \cB_{t-1}\}$ that only satisfy an average measurement entropy of $n/2$, which is only as good as what can be achieved with two measurements \eqref{eq:entropic-ur-n}. In other words, looking at the pairwise inner product between vectors in different bases is not enough to obtain uncertainty relations stronger than \eqref{eq:entropic-ur-n}. \red{It is for this reason that so little is understood about uncertainty relations for $t > 2$ measurements. (See \cite{WW10}.)}

 To achieve an average measurement entropy of $(1-\e) n$ for small $\e$ while keeping the number of bases subexponential in $n$, the only known constructions are probabilistic and computationally inefficient. 
 \citeN{HLSW04} prove that random bases satisfy entropic uncertainty relations of the form \eqref{eq:entropic-ur-n} with $n^4$ measurements with an average measurement entropy of $n-3$. This leads to an encoding that locks $n-3$ bits about $X \in \{0,1\}^n$ using a key of $4 \log n$ bits. 
 Recently, \citeN{Dup09} and \citeN{FDHL10} prove that random encodings exhibit a locking behaviour in a stronger sense and that it is possible to lock up to $n-\delta$ bits for any arbitrarily small constant $\delta$ while still using a key of $O(\log n)$ bits. 
To obtain an explicit construction, standard derandomization techniques are not known to work in this setting. For example, unitary designs \cite{DCEL09} define an exponential number of bases. Moreover, using a $\delta$-biased subset of the set of Pauli matrices \cite{AS04,DD10} fails to produce a locking scheme unless the subset has a size of close to $2^n$ (see Appendix \ref{sec:app-pauli}).

\subsection{Our results}
In this paper, we study uncertainty relations in the light of a connection with low-distortion embeddings of $(\CC^d, \ell_2)$ into $(\CC^{d'}, \ell_1)$. The intuition behind this connection is very simple. Consider the measurements defined by a set of orthonormal bases $\{\cB_0, \cB_1, \dots, \cB_{t-1}\}$ of $(\CC^2)^{\otimes n}$. The bases $\{\cB_0, \cB_1, \dots, \cB_{t-1}\}$ satisfy an uncertainty relation if for every $n$-qubit state $\ket{\psi}$ and ``most'' bases $\cB_k$, the vector representing $\ket{\psi}$ in $\cB_k$ is ``spread''. One way of quantifying the spread of a vector is by its $\ell_1$ norm, i.e., the sum of the absolute values of its components. A vector $\ket{\psi} \in (\CC^2)^{\otimes n}$ of unit $\ell_2$ norm is well spread if its $\ell_1$ norm is close to its maximal value of $\sqrt{2^n}$. 

Embeddings from a Euclidean space $(\RR^d, \ell_2)$ into $(\RR^{d'}, \ell_1)$ (or more generally, any finite-dimensional normed space) that approximately preserve the norm up to a scaling factor are typically studied in the area of asymptotic geometric analysis \cite{Dvo61,Mil71,FLM77,MS86}. More recently, low-distortion embeddings --- and in particular from $\ell_2$ into $\ell_1$ --- started to gain interest in the computer science community for their applications to approximation algorithms and compressed sensing \cite{Ind06,Ind07,GLR08}.
For our applications in quantum information theory, the relevant norm is not the $\ell_1$ norm but rather a closely related norm called $\ell_1(\ell_2)$.

Motivated by these considerations, we measure the uncertainty of a distribution by taking a marginal and measuring its closeness to the uniform distribution. This is a stronger requirement than having large Shannon entropy and it leads to the definition of a \emph{metric uncertainty relation}  (Definition \ref{def:metric-ur}). Using standard techniques from asymptotic geometric analysis, we prove the existence of strong metric uncertainty relations (Theorem \ref{thm:existence-ur}). This result can be seen as a strengthening of Dvoretzky's theorem \cite{Dvo61,Mil71} for the special case of the $\ell_1(\ell_2)$ norm. In addition to giving a stronger statement with better parameters, our analysis of the uncertainty relations satisfied by random bases is simpler than earlier proofs \cite{HLSW04,FDHL10}. In particular, for large $n$, we prove the existence of entropic uncertainty relations with average measurement entropy strictly increasing with the number of measurements. This result leads to better results on the existence of locking schemes (Corollary \ref{cor:existence-locking}). We also show in Theorem \ref{thm:hiding-fingerprint} how to use these locking schemes to build quantum hiding fingerprints as defined by \citeN{GI10}.

Moreover, adapting an explicit low-distortion embedding of $(\RR^d, \ell_2)$ to $(\RR^{d'}, \ell_1)$ with $d' = d^{1+o(1)}$ due to \citeN{Ind07}, we obtain explicit bases of $(\CC^2)^{\otimes n}$ that satisfy strong metric uncertainty relations for a number of bases that is polynomial in $n$. Measuring in these bases can be performed by almost linear size quantum circuits. The use of a strong permutation extractor is the main new ingredient that makes our ``quantization'' of Indyk's construction satisfy stronger uncertainty relations than do general mutually unbiased bases.
A strong permutation extractor (Definition \ref{def:perm-extractor}) is a small family of permutations of bit strings with the property that for any probability distribution on input bit strings with high min-entropy, applying a typical permutation from the family to the input induces an almost uniform probability distribution on a prefix of the output bits. It is a special kind of randomness extractor, a combinatorial object with many applications to the theory of pseudorandomness and to cryptography; see \cite{Sha02,Vad07}. Our construction of efficiently computable bases satisfying strong metric uncertainty relations involves an alternating application of approximately mutually unbiased bases and strong permutation extractors. Our approximately mutually unbiased bases consist of sets of single-qubit Hadamard gates. Moreover, we build efficiently computable and invertible permutations that define an extractor using the results of \citeN{GUV09}. 

\red{Even though the idea of combining mutually unbiased bases and extractors comes from \cite{Ind07}, in hindsight, it is
very natural from the point of view of quantum cryptography. Measurements in (approximately) unbiased bases are used in almost all quantum cryptographic protocols. The objective of such a step is usually to bound the probability that an adversary can guess the outcome of the associated measurement. Once such a bound is guaranteed, one can distill the randomness produced into almost uniform random bits using a step of privacy amplification which makes use of a randomness extractor. Our quantization of Indyk's construction can be seen as a repeated ``coherent'' application of these two steps.}

We use these uncertainty relations to build explicit locking schemes whose encoding and decoding operations can be performed by quantum circuits of size almost linear in the length of the message (see Table \ref{tbl:locking-results}). 
Moreover, we also obtain a locking scheme where both the encoding and decoding operations consist of a classical computation with polynomial runtime and a quantum computation using only a small number of single-qubit Hadamard gates (Corollary \ref{cor:explicit-locking}).
Performing these quantum operations can in principle be done using the same technology as implementing the BB84 quantum key distribution protocol \cite{BB84}, but our idealized scheme must still be made robust to noise and imperfect devices. It should be noted that for this simple scheme, the message is encoded in a slightly larger quantum system.
This locking scheme can be used to obtain string commitment protocols \cite{BCHLW08} that are efficient in terms of computation and communication.

We also give an application of our uncertainty relations to a problem called quantum identification. Quantum identification is a communication task for two parties Alice and Bob, where Alice is given a pure quantum state $\ket{\psi}$ and Bob wants to simulate measurements of the form $\{\proj{\ph}, \1 - \proj{\ph}\}$ on $\ket{\psi}$ where $\ket{\ph}$ is a pure quantum state. This task can be seen as a quantum analogue of the problem of equality testing \cite{AD89,KN97} where Alice and Bob hold $n$-bit strings $x$ and $y$ and Bob wants to determine whether $x=y$ using a one-way classical channel from Alice to Bob. 
%
\citeN{HW10} showed that classical communication alone is useless for quantum identification. However, having access to a negligible amount of quantum communication makes classical communication useful. Their proof is non-explicit. Here, we describe an efficient encoding circuit that also uses less quantum communication: it allows the identification of an $n$-qubit state by communicating only a single message of $O(\log^2 n)$ qubits and $n$ classical bits.

\begin{table}
\tbl{Comparison of different locking schemes. $n$ is the number of bits of the message $X$. The information leakage and the size of the key $K$ are measured in bits and the size of the ciphertext $\cE(X,K)$ in qubits. Efficient locking schemes have encoding and decoding quantum circuits of size polynomial in $n$. The locking schemes of the first and next to last actually have encoding circuits that are in principle implementable with current technology; they only use classical computations and simple single-qubit transformations. It should be noted that our locking definition implies all the previous definitions.
Note that the variable $\e$ can depend on $n$. For example, one can take $\e = \eta/n$ to make the information leakage arbitrarily small.  The symbol $O( \cdot )$ refers to constants independent of $\e$ and $n$, but there is a dependence on $\delta$ for the next to last row.\label{tbl:locking-results}}
{
\begin{tabular}{|c||c|c|c|c|}
\hline
			& Inf. leakage & Size of key & Size of ciphertext & Efficient ? \\ \hline \hline
\cite{DHLST04}
	& $n/2$ & $1$ & $n$ & yes \\ \hline
\cite{HLSW04} 	
& $3$ & $4 \log(n)$ & $n$ & no \\ \hline
\cite{FDHL10} 	
& $\e n$ & $2\log(n/\e^2) + O(1)$ & $n$ & no \\ \hline \hline
Corollary \ref{cor:existence-locking}	& $\e n$  & $2 \log(1/\e) + O(\log \log(1/\e))$ & $n + 2 \ceil{\log(9/\e)}$ & no \\ \hline
Corollary \ref{cor:existence-locking} 	& $\e n$  & $4 \log(1/\e) + O(\log \log(1/\e))$ & $n$ & no \\ \hline
Corollary \ref{cor:explicit-locking}	& $\e n$ & $O_{\delta}(\log(n/\e))$ & $(4+\delta) \cdot n$ & yes \\ \hline
Corollary \ref{cor:explicit-locking}	& $\e n$ & $O(\log(n/\e) \log(n))$ & $n$ & yes \\ \hline
\end{tabular}
}
\end{table}

\subsection{Related work}
Aubrun, Szarek and Werner \shortcite{ASW10,ASW10b} used a connection between low-distortion embeddings and quantum information. They show in \cite{ASW10} that the existence of large subspaces of highly entangled states follows from Dvoretzky's theorem for the Schatten $p$-norm\footnote{The Schatten $p$-norm of a matrix $M$ is defined as the $\ell_p$ norm of a vector of singular values of $M$.} for $p > 2$. This in turns shows the existence of channels that violate additivity of minimum output $p$-R\'enyi entropy as was previously demonstrated by \citeN{HW08}. Using a more delicate argument \cite{ASW10b}, they are also able to give an alternative proof of a violation of the additivity conjecture, which was previously found by \citeN{Has09}.

In a cryptographic setting, \citeN{DPS04} used ideas related to locking to develop quantum ciphers that have the property that the key used for encryption can be recycled. In \cite{DPS05}, they construct a quantum key recycling scheme (see also \cite{OH05}) with near optimal parameters by encoding the message together with its authentication tag using a full set of mutually unbiased bases.

\subsection{Notation and basic facts}

We use the following notation throughout the paper. For a positive integer $n$, we define $[n] = \{0, \dots, n-1\}$.
\subsubsection{Probability}
 Random variables are usually denoted by capital letters $X, K, \dots$, while $p_X$ denotes the distribution of $X$, i.e., $\pr{X = x} = p_X(x)$. The notation $X \sim p$ means that $X$ has distribution $p$. $\unif(S)$ is the uniform distribution on the set $S$. To measure the distance between probability distributions on a finite set $\cX$, we use the total variation distance or trace distance $\tracedist{p,q} = \frac{1}{2} \sum_{x \in \cX} |p(x) - q(x)|$. We will also write $\tracedist{X,Y}$ for $\tracedist{p_X, p_Y}$. When $\tracedist{X,Y} \leq \e$, we say that $X$ is $\e$-close to $Y$. A useful characterization of the trace distance is $\tracedist{p,q} = \min_{X \sim p, Y \sim q} \pr{X \neq Y}$ (this equality is known as Doeblin's coupling lemma \cite{Doe38}). Another useful measure of closeness between distributions is the fidelity $\fid{p,q} = \sum_{x \in \cX} \sqrt{p(x)q(x)}$, also known as the Bhattacharyya distance and related to the Hellinger distance. We have the following relation between the fidelity and the trace distance
\begin{equation}
\label{eq:ineq-tracedist-fid}
1 - \fid{p,q} \leq \tracedist{p,q} \leq \sqrt{1-\fid{p,q}^2}.
\end{equation}
The Shannon entropy of a distribution $p$ on $\cX$ is defined as $\entH(p) = -\sum_{x \in \cX} p(x) \log p(x)$ where the $\log$ is taken here and throughout the paper to be base two. We will also write $\entH(X)$ for $\entH(p_X)$. The mutual information between two random variables $X$ and $Y$ is defined as $\entI(X; Y) = \entH(X) + \entH(Y) - \entH(X,Y)$. The min-entropy of a distribution $p$ is defined as $\entHmin(p) = -\log \max_x p(x)$. We say that a random variable $X$ is a $k$-source if $\entHmin(X) \geq k$. To refer to the $i$-th component of a vector $v \in \RR^n$, we usually write $v_i$ except when $v$ already has a subscript, in which case we use $v(i)$. The Hamming weight of a binary vector $v$ (number of ones) is denoted by $\weight(v)$ and the Hamming distance between two binary vectors $v, v'$ (number of components that are different) is written as $\hamdist(v, v')$.

\subsubsection{Quantum mechanics}

The state of a pure quantum system is represented by a unit vector in a Hilbert space. Quantum systems  are denoted $A, B, C\dots$ and are identified with their corresponding Hilbert spaces. The dimension of a Hilbert space $A$ is denoted by $d_A$. Every Hilbert space $A$ comes with a preferred orthonormal basis $\{\ket{a}\}_{a \in [d_A]}$ that we call the computational basis. The elements of this basis are labeled by integers from $0$ to $d_A-1$. For a Hilbert space of the form $(\CC^{2})^{\otimes n}$, this canonical basis will also be labeled by strings in $\{0,1\}^n$. $A \isom B$ means that the Hilbert spaces $A$ and $B$ are isomorphic.  

To describe a distribution $\{p_1, \dots, p_r\}$ over quantum states $\{\ket{\psi_1}, \dots, \ket{\psi_r}\}$ (also called a mixed state), we use a density operator $\rho = \sum_{i=1}^r p_i \proj{\psi_i}$, where $\proj{\psi}$ refers to the projector on the line defined by $\ket{\psi}$. A density operator is a Hermitian positive semidefinite operator with unit trace. The density operator associated with a pure state is abbreviated by omitting the ket and bra $\psi \eqdef \proj{\psi}$. $\cS(A)$ is the set of density operators acting on $A$. The Hilbert space on which a density operator $\rho \in \cS(A)$ acts is sometimes denoted by a superscript, as in $\rho^A$. This notation is also used for pure states $\ket{\psi}^A \in A$.

In order to describe the joint state of a system $AB$, we use the tensor product Hilbert space $A \otimes B$, which is sometimes simply denoted $AB$.  If $\rho^{AB}$ describes the joint state on $AB$, the state on the system $A$ is described by the partial trace $\rho^A \eqdef \tr_B \rho^{AB}$. If $U$ is a unitary acting on $A$, and $\ket{\psi}$ a state in $A \ox B$, we sometimes use $U \ket{\psi}$ to denote the state $(U \ox \1^B) \ket{\psi}$, where the symbol $\1^B$ is reserved for the identity map on $B$.

The most general way to obtain classical information from a quantum state is by performing a measurement. A measurement is described by a positive operator-valued measure (POVM), which is a set $\{P_1, \dots, P_s\}$ of positive semidefinite operators that sum to the identity.  If the state of the quantum system is represented by the density operator $\rho$, the probability of observing the outcome labeled $i$ is $\tr[P_i \rho]$ for all $i \in \{1, \dots, s\}$.
For a state $\ket{\psi} \in A$, $p_{\ket{\psi}}$ denotes the distribution of the outcomes of the measurement of $\ket{\psi}$ in the basis $\{\ket{a}\}$. We have $p_{\ket{\psi}}(a) = |\braket{a}{\psi}|^2$. Similarly, for a mixed state $\rho$, we define $p_{\rho}(a) = \tr[ \proj{a} \rho]$.

The trace distance between density operators acting on $A$ is defined by $\tracedist{\rho, \sigma} = \frac{1}{2}\tr \sqrt{(\rho - \sigma)^2}$.
The von Neumann entropy of a quantum state $\rho^A$ is defined by $\entH(\rho^A) = - \tr \rho \log \rho$. It will also be denoted $\entH(A)_{\rho}$. For a bipartite state $\rho^{AB} \in \cS(AB)$, the quantum mutual information is $\entI(A;B)_{\rho} = \entH(A)_{\rho} + \entH(B)_{\rho} - \entH(A,B)_{\rho}$. We use Fannes' inequality \cite{Fan73}, or more precisely an improvement by \citeN{Aud07}, which states that for any states $\rho$ and $\sigma$ on $A$, 
\begin{equation}
|\entH(\rho) - \entH(\sigma)| \leq \tracedist{\rho, \sigma} \log d_A + \binent(\tracedist{\rho, \sigma}), 
\label{eq:fannes}
\end{equation}
with $\binent(\e) = -\e \log(\e) - (1-\e) \log(1-\e)$.



\section{Uncertainty relations}

\paragraph{Outline of the section}
In this section, we start by introducing uncertainty relations and setting up some notation (Section \ref{sec:background-ur}). Then, we define metric uncertainty relations in Section \ref{sec:metric-ur}. In Section \ref{sec:existence-ur}, we prove the existence of strong metric uncertainty relations. Explicit constructions are given in Section \ref{sec:explicit-ur}.

\subsection{Background}
\label{sec:background-ur}
Consider a set of orthonormal bases $\cB = \{ \cB_0, \dots, \cB_{t-1}\}$ of the Hilbert space $C$. Each basis $\cB_k = (v^k_0, \dots, v^k_{d_C-1})$ defines a measurement on $C$. The outcomes of these measurements are indexed by $x \in [d_C]$. The outcome distribution $p_{\cB_k , \ket{\psi}}$ when the measurement is performed on the state $\ket{\psi} \in C$ is defined by
$p_{\cB_k , \ket{\psi}}(x) = |\braket{v^k_x}{\psi}|^2$
for all $x \in [d_C]$.
An uncertainty relation for a set of orthonormal bases $\cB = \{ \cB_0, \dots, \cB_{t-1}\}$ expresses the property that for any state $\ket{\psi} \in C$, there are some measurements in $\cB$ whose outcomes given state $\ket{\psi}$ have some uncertainty. A common way of quantifying this uncertainty is by using the Shannon entropy. The set of bases $\cB$ is said to satisfy an \emph{entropic uncertainty relation} if there exists a positive number $h$ such that for all states $\ket{\psi} \in C$,
\[
\frac{1}{t} \sum_{k=0}^{t-1} \entH(p_{\cB_k, \ket{\psi}}) \geq h.
\]
Note that by choosing a state $\ket{\psi}$ in the basis $\cB_0$, we obtain $\entH(p_{\cB_0, \ket{\psi}}) = 0$. As $\entH(p_{\cB, \ket{\psi}}) \leq \log d_C$, this implies that $h$ cannot be larger than $(1-1/t) \log d_C$.

It is more convenient here to talk about uncertainty relations for a set of unitary transformations. Let $\{\ket{x}^C\}_x$ be the computational basis of $C$. We associate to the unitary transformation $U$ the basis $\{U^{\dg} \ket{x}\}_x$. On a state $\ket{\psi}$, the outcome distribution is described by
\[
p_{U \ket{\psi}}(x) = |\bra{x} U \ket{\psi}|^2.
\]
As can be seen from this equation, we can equivalently talk about measuring the state $U \ket{\psi}$ in the computational basis. An entropic uncertainty relation for $U_0, \dots, U_{t-1}$ can be written as
\begin{equation}
\label{eq:entropic-ur}
\frac{1}{t} \sum_{k=0}^{t-1} \entH( p_{U_k \ket{\psi}} ) \geq h.
\end{equation}
Entropic uncertainty relations have been used to prove the security of quantum key distribution and of cryptographic protocols in the bounded and noisy quantum storage models \cite{TR11,DFRSS07,BFW12}. Other applications of entropic uncertainty relations are given in Section \ref{sec:locking}. For more details on entropic uncertainty relations and their applications, see the recent survey \cite{WW10}.


\subsection{Metric uncertainty relations}
\label{sec:metric-ur}

Here, instead of using the entropy as a measure of uncertainty, we use closeness to the uniform distribution. In other words, we are interested in sets of unitary transformations $U_0, \dots, U_{t-1}$ that for all $\ket{\psi} \in C$ satisfy
\[
\frac{1}{t} \sum_{k=0}^{t-1} \tracedist{p_{U_k \ket{\psi}}, \unif([d_C])} \leq \e
\]
for some $\e \in (0,1)$. $\Delta(p,q)$ refers to the total variation distance between the distributions $p$ and $q$.
This condition is very strong, in fact too strong for our purposes, and we will see that a weaker definition is sufficient to imply entropic uncertainty relations. Let $C = A \otimes B$. (For example, if $C$ consists of $n$ qubits, $A$ might represent the first $n - \log n$ qubits and $B$ the last $\log n$ qubits.) Moreover, let the computational basis for $C$ be of the form $\{\ket{a}^A \otimes \ket{b}^B\}_{a,b}$ where $\{\ket{a}\}$ and $\{\ket{b}\}$ are the computational bases of $A$ and $B$. Instead of asking for the outcome of the measurement on the computational basis of the whole space to be uniform, we only require that the outcome of a measurement of the $A$ system in its computational basis $\{\ket{a}\}$ be close to uniform. More precisely, we define for $a \in [d_A]$,
\[
p^A_{U_k \ket{\psi}}(a) = \sum_{b=0}^{d_B - 1} |\bra{a}^A \bra{b}^B U_k \ket{\psi}|^2.
\]
We can then define a metric uncertainty relation. Naturally, the larger the $A$ system, the stronger the uncertainty relation for a fixed $B$ system.
\begin{definition}[Metric uncertainty relation]
\label{def:metric-ur}
Let $A$ and $B$ be Hilbert spaces. We say that a set $\{U_0, \dots, U_{t-1}\}$ of unitary transformations on $AB$ satisfies an $\e$-metric uncertainty relation on $A$ if for all states $\ket{\psi} \in AB$,
\begin{equation}
\label{eq:marginalA}
\frac{1}{t} \sum_{k=0}^{t-1} \tracedist{p^A_{U_k \ket{\psi}},\unif([d_A]) } \leq \e.
\end{equation}
\end{definition}
\begin{myremark}[]
Observe that this implies that \eqref{eq:marginalA} also holds for mixed states: for any $\psi \in \cS(A \otimes B)$,
$
\frac{1}{t} \sum_{k=0}^{t-1} \tracedist{p^A_{U_k \psi U_k^{\dagger}},\unif([d_A]) } \leq \e.
$
\end{myremark}

\paragraph{Metric uncertainty relations imply entropic uncertainty relations}
In the next proposition, we show that a metric uncertainty relation gives rise to an entropic uncertainty relation. It is worth stressing that there are no restrictions on measurements.
\begin{proposition}
\label{prop:metric-to-entropic}
Let $\e \in(0, 1)$ and $\{U_0, \dots, U_{t-1}\}$ be a set of unitaries on $AB$ satisfying an $\e$-metric uncertainty relation on $A$:
\[
\frac{1}{t} \sum_{k=0}^{t-1} \tracedist{p^A_{U_k \ket{\psi}}, \unif([d_A])}  \leq \e.
\]
Then
\[
\frac{1}{t} \sum_{k=0}^{t-1} \entH(p_{U_k \ket{\psi}}) \geq (1-\e) \log d_A - \binent(\e).
\]
where $\binent$ is the binary entropy function.
\end{proposition}
\begin{proof}
Recall that the distribution $p^A_{U_k \ket{\psi}}$ (see equation \eqref{eq:marginalA} for a definition) on $[d_A]$ is a marginal of the distribution $p_{U_k\ket{\psi}}$. Thus $\entH(p_{U_k\ket{\psi}}) \geq \entH(p^A_{U_k, \ket{\psi}})$. Using Fannes-Audenaert's inequality \eqref{eq:fannes}, we have for all $k$
\begin{align*}
\entH(p^A_{U_k \ket{\psi}}) &\geq \log d_A - \tracedist{p^A_{U_k \ket{\psi}},\unif([d_A])} \log d_A - \binent\left(\tracedist{p^A_{U_k \ket{\psi}},\unif([d_A])}\right).
\end{align*}
By averaging over $k$ and using the concavity of $\binent$, we get the desired result.
\end{proof}

\paragraph{Explicit link to low-distortion embeddings}
Even though we do not explicitly use the link to low-distortion embeddings, we describe the connection as it might have other applications. In the definition of metric uncertainty relations, the distance between distributions was computed using the trace distance. The connection to low-distortion metric embeddings is clearer when we measure closeness of distributions using the fidelity. We have
\begin{align}
\fid{p^A_{U_k \ket{\psi}},\unif([d_A])} &= \frac{1}{\sqrt{d_A}} \sum_{a=0}^{d_A-1} \sqrt{p^A_{U_k \ket{\psi}}(a)} \notag \\
											&= \frac{1}{\sqrt{d_A}}  \sum_{a=0}^{d_A -1} \sqrt{\sum_{b=0}^{d_B -1} |\bra{a}^A \bra{b}^B U_k \ket{\psi}|^2} \notag \\
											&= \frac{1}{\sqrt{d_A}}  \| U_k \ket{\psi} \|_{\ell_1^A(\ell_2^B)} \label{eq:fidelity-l1l2}	
\end{align}
where the norm $\ell_1^A(\ell_2^B)$ is defined by
\begin{definition}[$\ell_1(\ell_2)$ norm]
\label{def:l1l2}
For a state $\ket{\psi} = \sum_{a \in [d_A],b \in [d_B]} \alpha_{a,b} \ket{a}^A \ket{b}^B$,
\[
\big \| \ket{\psi} \big \|_{\ell_1^{A}(\ell_2^B)} = \sum_{a \in [d_A]} \big \| \{\alpha_{a,b}\}_b \big \|_2 = \sum_{a \in [d_A]} \sqrt{\sum_{b \in [d_B]} |\alpha_{a,b}|^2}.
\]
We use $\| \cdot \|_{12} \eqdef \| \cdot \|_{\ell_1^{A}(\ell_2^B)}$ when the systems $A$ and $B$ are clear from the context.
\end{definition}
Observe that this definition of norm depends on the choice of the computational basis. The $\ell^A_1(\ell^B_2)$ norm will always be taken with respect to the computational bases.

For $\{U_0, \dots, U_{t-1}\}$ to satisfy an uncertainty relation, we want
\[
\frac{1}{t}\sum_{k} \frac{1}{\sqrt{d_A}} \| U_k \ket{\psi} \|_{\ell^{A}_1(\ell_2^B)} \geq 1-\e.
\]
This expression can be rewritten by introducing a new register $K$ that holds the index $k$. We get for all $\ket{\psi}$ by writing $C=AB$
\begin{equation}
\label{eq:low-distortion1}
\left\| \frac{1}{\sqrt{t}} \sum_{k} U_k \ket{\psi}^C \ket{k}^K \right\|_{\ell^{AK}_1(\ell_2^B)} \geq (1-\e) \sqrt{t \cdot d_A}.
\end{equation}
Using the Cauchy-Schwarz inequality, which in this context reads $\| \ket{\phi} \|_{\ell_1^A(\ell_2^B)} \leq \sqrt{d_A} \| \ket{\phi} \|_{2}$ for any $\ket{\phi} \in AB$, we have that for all $\ket{\psi}$,
\begin{equation}
\label{eq:low-distortion2}
\left\| \frac{1}{\sqrt{t}} \sum_{k} U_k \ket{\psi}^C \ket{k}^K \right\|_{\ell^{AK}_1(\ell_2^B)} \leq \sqrt{t \cdot d_A} \left\| \frac{1}{\sqrt{t}} \sum_{k} U_k \ket{\psi}^C \ket{k}^K \right\|_2 = \sqrt{t \cdot d_A}.
\end{equation}
Rewriting \eqref{eq:low-distortion1} and \eqref{eq:low-distortion2} as
\[
(1-\e) \leq \frac{1}{ \sqrt{t \cdot d_A}} \cdot \frac{\left\| \frac{1}{\sqrt{t}} \sum_{k} U_k \ket{\psi}^C \ket{k}^K \right\|_{\ell^{AK}_1(\ell_2^B)}} { \left\| \frac{1}{\sqrt{t}} \sum_{k} U_k \ket{\psi}^C \ket{k}^K \right\|_2} \leq 1,
\]
we see that the image of $C$ by the linear map $\ket{\psi} \mapsto \frac{1}{\sqrt{t}} \sum_k U_k \ket{\psi} \otimes \ket{k}$ is an almost Euclidean subspace of $(A \otimes K \otimes B, \ell^{AK}_1(\ell_2^B))$. In other words, as the map $\ket{\psi} \mapsto \frac{1}{\sqrt{t}} \sum_k U_k \ket{\psi} \otimes \ket{k}$ is an isometry (in the $\ell_2$ sense), it is an embedding of $(C, \ell_2)$ into $(AKB, \ell^{AK}_1(\ell_2^B))$ with distortion $1/(1-\e)$ \cite{Mat02}.

Observe that a general low-distortion embedding of $(C, \ell_2)$ into $(AKB, \ell_1^{AK}(\ell_2^B))$ does not necessarily give a metric uncertainty relation as it need not be of the form $\ket{\psi} \mapsto \frac{1}{\sqrt{t}} \sum_k U_k \ket{\psi} \otimes \ket{k}$. When $t = 2$, a metric uncertainty relation is related to the notion of Kashin decomposition \cite{Kas77}; see also \cite{Pis89,Sza06}.

\paragraph{A remark on the composition of metric uncertainty relations}
There is a natural way of building an uncertainty relation for a Hilbert space from uncertainty relations on smaller Hilbert spaces. This composition property is also important for the cryptographic applications of metric uncertainty relations presented in the second half of the paper, in which setting it ensures the security of parallel composition of locking-based encryption.
\begin{proposition}
Consider Hilbert spaces $A_1$, $A_2$, $B_1$, $B_2$. For $i \in \{0,1\}$, let $\{U^{(i)}_{k_i}\}_{k_i \in [t_i]}$ be a set of unitary transformations of $A_i \otimes B_i$ satisfying an $\e$-metric uncertainty relation on $A_i$.
Then, $\{U^{(1)}_{k_1} \otimes U^{(2)}_{k_2}\}_{k_1, k_2 \in [t_1] \times [t_2]}$ satisfies a $2\e$-metric uncertainty relation on $A_1 \otimes A_2$. 
\end{proposition}
\begin{proof}
Let $\ket{\psi} \in (A_1 \otimes B_1) \otimes (A_2 \otimes B_2)$ and let $p_{k_1, k_2}$ denote the distribution obtained by measuring $U^{(1)}_{k_1} \otimes U^{(2)}_{k_2} \ket{\psi}$ in the computational basis of $A_1 \otimes A_2$.
Our objective is to show that
\begin{equation}
\label{eq:parallel-compose}
\frac{1}{t_1 t_2} \sum_{k_1 \in [t_1], k_2 \in [t_2]} \tracedist{p_{k_1, k_2}, \unif([d_{A_1}] \times [d_{A_2}])} \leq 2\e.
\end{equation}
We have
\begin{align}
& \tracedist{p_{k_1, k_2}, \unif([d_{A_1}] \times [d_{A_2}])} \notag \\
 &= \frac{1}{2} \sum_{a_1, a_2} \left|p_{k_1, k_2}(a_1, a_2) - \frac{1}{d_{A_1} d_{A_2}}\right| \notag \\
											&\leq \frac{1}{2} \sum_{a_1, a_2} \left|p_{k_1, k_2}(a_1, a_2) - \frac{ p^{A_1}_{k_1, k_2}(a_1) }{d_{A_2}} \right| +  \frac{1}{2} \sum_{a_1, a_2} \left| \frac{ p^{A_1}_{k_1, k_2}(a_1) }{d_{A_2}} - \frac{1}{d_{A_1} d_{A_2}} \right| \notag \\
											&= \frac{1}{2} \sum_{a_1} p^{A_1}_{k_1, k_2}(a_1) \sum_{a_2} \left|\frac{p_{k_1, k_2}(a_1, a_2)}{p^{A_1}_{k_1, k_2}(a_1) } - \frac{1}{d_{A_2}} \right| +  \frac{1}{2} \sum_{a_1} \left| p_{k_1, k_2}^{A_1}(a_1)  - \frac{1}{d_{A_1}} \right| \label{eq:parallel-compose2}
\end{align}
where $p^{A_1}_{k_1, k_2}(a_1) \eqdef \sum_{a_2} p_{k_1, k_2}(a_1, a_2)$ is the outcome distribution of measuring the  $A_1$ system of $U^{(1)}_{k_1} \otimes U^{(2)}_{k_2} \ket{\psi}$. The distribution $p^{A_1}_{k_1, k_2}$ can also be seen as the outcome of measuring the mixed state 
\[
U^{(1)}_{k_1} \psi^{A_1B_1} {U^{(1)}_{k_1}}^{\dagger}
\]
in the computational basis $\{ \ket{a_1} \}$. Thus, we have for any $k_2 \in [t_2]$,
\[
\frac{1}{t_1} \sum_{k_1} \tracedist{p^{A_1}_{k_1, k_2}, \unif([d_{A_1}])} \leq \e.
\]
Moreover, for $a_1 \in [d_{A_1}]$, the distribution on $[d_{A_2}]$ defined by $\frac{p_{k_1, k_2}(a_1, a_2)}{p^{A_1}_{k_1, k_2}(a_1) }$ is the outcome distribution of measuring in the computational basis of $A_2$ the state
\[
U^{(2)}_{k_2} \psi^{A_2B_2}_{k_1, a_1} {U^{(2)}_{k_2}}^{\dagger}
\]
where $\psi^{A_2B_2}_{k_1, a_1}$ is the density operator describing the state of the system $A_2B_2$ given that the outcome of the measurement of the $A_1$ system is $a_1$. We can now use the fact that $\{U^{(2)}_{k_2} \}$ satisfies a metric uncertainty relation. Taking the average over $k_1$ and $k_2$ in equation \eqref{eq:parallel-compose2}, we get
\[
\frac{1}{t_1 t_2} \sum_{k_1, k_2} \tracedist{p_{k_1, k_2}, \unif([d_{A_1}] \times [d_{A_2}])} \leq 2\e.
\]
\end{proof}
This observation is in the same spirit as \cite[Proposition 1]{IS10}, and can in fact be used to build large almost Euclidean subspaces of $\ell^A_1(\ell^B_2)$.


\subsection{Metric uncertainty relations: existence}
\label{sec:existence-ur}
In this section, we prove the existence of families of unitary transformations satisfying strong uncertainty relations. The proof proceeds by showing that choosing random unitaries according to the Haar measure defines a metric uncertainty relation with positive probability. The techniques used are quite standard and date back to Milman's proof of Dvoretzky's theorem \cite{Mil71,FLM77}. In fact, using the connection to embeddings of $\ell_2$ into $\ell_1(\ell_2)$ presented in the previous section, this existential theorem can be viewed as a strengthening of Dvoretzky's theorem for the $\ell_1(\ell_2)$  norm \cite{MS86}. Explicit constructions of uncertainty relations are presented in the next section.

In order to use metric uncertainty relations to build quantum hiding fingerprints, we require an additional property for $\{U_0, \dots, U_{t-1}\}$. A set of unitary transformations $\{U_0, \dots, U_{t-1}\}$ of $\CC^d$ is said to define $\gamma$-approximately mutually unbiased bases ($\gamma$-MUBs) if for all elements $\ket{x}$ and $\ket{y}$ of the computational basis and all $k \neq k'$, we have
\begin{equation}
\label{eq:def-approx-mub}
|\bra{x} U_k^{\dagger} U_{k'} \ket{y}| \leq \frac{1}{d^{\gamma/2}}.
\end{equation}
$1$-MUBs correspond to the usual notion of mutually unbiased bases.

%

\begin{theorem}[Existence of metric uncertainty relations]
\label{thm:existence-ur}
Let $c=16$ and $\e \in (0,1)$. Let $A$ and $B$ be Hilbert spaces with $\dim B \geq 9/\e^2$ and $d \eqdef \dim A \ox B$. Then, for all $t > \frac{72 c \cdot \ln(9/\e)}{\e^2}$, there exists a set $\{U_{0}, \dots, U_{t-1}\}$ of unitary transformations of $AB$ satisfying an $\e$-metric uncertainty relation on $A$: for all states $\ket{\psi} \in AB$,
\[
\frac{1}{t} \sum_{k=0}^{t-1} \tracedist{p^A_{U_k \ket{\psi}}, \unif([d_A])} \leq \e.
\]
Moreover, for $\gamma \in (0,1)$ and $d$ such that $4t^2d^2 \exp{-d^{1-\gamma}} < 1/2$, the unitaries $\{U_0, \dots, U_{t-1}\}$ can be chosen to also form $\gamma$-MUBs.
\end{theorem}
\begin{myremark}[]
The proof proceeds by choosing a set of unitary transformations at random. See \eqref{eq:prur} and \eqref{eq:pr-approx-mub} for a precise bound on the probability that such a set does not form a metric uncertainty relation or a $\gamma$-MUB.
\end{myremark}
\begin{proof}
The basic idea is to evaluate the expected value of $\tracedist{p^A_{U\ket{\psi}}, \unif([d_A])}$ for a fixed state $\ket{\psi}$ when $U$ is a random unitary chosen according to the Haar measure. Then, we use a concentration argument to show that with high probability, this distance is close to its expected value. After this step, we show that the additional averaging $\frac{1}{t} \sum_{k=0}^{t-1} \tracedist{p^A_{U_k\ket{\psi}}, \unif([d_A])}$ of $t$ independent copies results in additional concentration at a rate that depends on $t$. We conclude by showing the existence of a family of unitaries that makes this expression small for all states $\ket{\psi}$ using a union bound over a $\delta$-net. The main ingredients of the proof are stated here but only proved in Appendix \ref{sec:app-existence-ur}. 

We start by computing the expected value of the fidelity $\ex{\fid{p^A_{U\ket{\psi}}, \unif([d_A])}}$, which can be seen as an $\ell_1(\ell_2)$ norm. 

\begin{lemma}[Expected value of $\ell^A_1(\ell^B_2)$ over the sphere]
\label{lem:expl1l2}
Let $\ket{\ph}^{AB}$ be a random pure state on $AB$. Then,
\[
\ex{\fid{p^A_{\ket{\ph}}, \unif([d_A])}} 
\geq \sqrt{1 - \frac{1}{d_B}}.
\]
\end{lemma}

We then use the inequality $\tracedist{p, q} \leq \sqrt{1 - \fid{p, q}^2}$ to get 
\[
\ex{ \tracedist{ p^A_{\ket{\ph}}, \unif([d_A])} } \leq \ex{\sqrt{ 1 - \fid{p^A_{\ket{\ph}}, \unif([d_A])} ^2 } }.
\]
By the concavity of the function $x \mapsto \sqrt{1-x^2}$ on the interval $[0,1]$,
\begin{align*}
\ex{ \tracedist{ p^A_{\ket{\ph}}, \unif([d_A])} }
								&\leq \sqrt{ 1 - \ex{\fid{p^A_{\ket{\ph}}, \unif([d_A])}}^2 } \\
								&\leq \sqrt{ 1 - \left(1 - \frac{1}{d_B} \right) } \\
								&\leq \e/3. \\
\end{align*}
The last inequality comes from the hypothesis of the theorem that $d_B \geq 9/\e^2$.
In other words, for any fixed $\ket{\psi}$, the average over $U$ of the trace distance between $p^A_{U \ket{\psi}}$ and the uniform distribution is at most $\e/3$. 
\red{Setting $\mu = \ex{\tracedist{p^A_{\ket{\ph}}, \unif([d_A])}}$, the next step is to evaluate $\pr{ \left| \frac{1}{t}\sum_{k=0}^{t-1} \tracedist{p^A_{U_k \ket{\psi}}, \unif([d_A])}  - \mu \right|  \geq \e/3 }$. This is done using a concentration inequality on the product of spheres. For completeness, an approach that is more elementary is presented in the appendix (Lemma \ref{lem:levy} and Lemma \ref{lem:avgconc}). While this second approach is more elementary, it leads to slightly worse constants and additional technical constraints on the relation between $t$ and $d$. However, these constraints do not substantially affect our conclusion.

\begin{lemma}[Concentration inequality on the product of spheres] 
\label{lem:conc-product-sphere}
Let $f : (\CC^d)^{\times t}  \to \RR$ and $\eta > 0$ be such that for all pure states $\ket{\ph_0}, \dots \ket{\ph_{t-1}}$ and $\ket{\psi_0}, \dots, \ket{\psi_{t-1}}$ in $\CC^d$, 
\[
\big| f(\ket{\ph_0}, \dots, \ket{\ph_{t-1}}) - f(\ket{\psi_0}, \dots, \ket{\psi_{t-1}}) \big| \leq \eta \sqrt{\sum_{i=0}^{t-1} \| \ket{\ph_i} - \ket{\psi_i} \|^2_2}.
\]
Let $\ket{\ph_0}, \dots, \ket{\ph_{t-1}}$ be independent random pure states in dimension $d$. Then for all $\delta \geq 0$,
\[
\pr{\Big| f(\ket{\ph_0}, \dots, \ket{\ph_{t-1}}) - \ex{f(\ket{\ph_0}, \dots, \ket{\ph_{t-1}})} \Big| \geq \delta } \leq 4 \exp{- \frac{\delta^2 d}{c \eta^2} }
\]
where $c$ is a constant. We can take $c = 16$.
\end{lemma}
\begin{proof}
We start by applying Example 6.5.2 of \cite{MS86} to obtain concentration around the median. Then to prove concentration around the mean, we can use the general Proposition V.4 also from \cite{MS86}.
\end{proof}
We apply this concentration result to $f : \ket{\ph_0}^{AB}, \dots, \ket{\ph_{t-1}}^{AB} \mapsto \frac{1}{t}\sum_{k=0}^{t-1} \tracedist{p^A_{\ket{\ph_k}}, \unif([d_A])}$. We start by finding an upper bound on the Lipshitz constant $\eta$. For any pure states $\{\ket{\ph_k}^{AB}\}_k$ and $\{\ket{\psi_k}^{AB}\}_k$, we have
\begin{align*}
&| f(\ket{\ph_0}, \dots, \ket{\ph_{t-1}}) - f(\ket{\psi_0}, \dots, \ket{\psi_{t-1}}) | 
\\&\leq \frac{1}{t} \sum_{k=0}^{t-1} \tracedist{p^A_{\ph_k}, p^A_{\psi_k}}  \\
			&\leq \frac{1}{t} \sum_{k=0}^{t-1} \frac{1}{2}\sum_{a,b} \left| |\bra{a} \bra{b} \ket{\ph_k}|^2 - |\bra{a} \bra{b} \ket{\psi_k}|^2  \right| \\
			&\leq \frac{1}{t} \sum_{k=0}^{t-1} \frac{1}{2} \sqrt{\sum_{a,b} \big| |\bra{a} \bra{b} \ket{\ph_k}| + |\bra{a} \bra{b} \ket{\psi_k}| \big|^2 \cdot  \sum_{a,b} \left| |\bra{a} \bra{b} \ket{\ph_k}|  - |\bra{a} \bra{b} \ket{\psi_k}| \right|^2}.
\end{align*}
The first two inequalities follow from the triangle inequality. The third inequality is due to the Cauchy Schwarz inequality. Continuing,
			\begin{align}
				| f(\ket{\ph_0}, \dots, \ket{\ph_{t-1}}) - f(\ket{\psi_0}, \dots, \ket{\psi_{t-1}}) | &\leq \frac{1}{t} \sum_{k=0}^{t-1} \| \ket{\ph_k} - \ket{\psi_k} \|_2 \notag \\
			&\leq \frac{1}{\sqrt{t}} \sqrt{\sum_{k=0}^{t-1} \| \ket{\ph_k} - \ket{\psi_k} \|^2_2}.
			\label{eq:ineq-tracedist-l2}
\end{align}
 The first inequality follows again from the triangle inequality and the last inequality follows from the Cauchy Schwarz inequality. Applying Lemma \ref{lem:conc-product-sphere} with $\eta = 1/\sqrt{t}$, we obtain
\[
\pr{ \left| \frac{1}{t}\sum_{k=0}^{t-1} \tracedist{p^A_{U_k \ket{\psi}}, \unif([d_A])}  - \mu \right|  \geq \e/3 } \leq 4 \exp{-  \frac{(\e/3)^2 t d}{c} }. 
\]
This inequality can also be achieved with a slightly worse constant by first applying L\'evy's Lemma \ref{lem:levy} to the function $\tracedist{p^A_{U_k \ket{\psi}}, \unif([d_A])}$ and then using Lemma \ref{lem:avgconc} to prove that averaging of $t$ independent copies results in additional concentration at a rate of $t$ .
} 

Continuing with Lemma \ref{lem:expl1l2}, we have
\begin{equation}
\label{eq:pravgfixedstate}
\pr{ \frac{1}{t}\sum_{k=0}^{t-1} \tracedist{p^A_{U_k\ket{\psi}}, \unif([d_A])} \geq 2\e/3 } \leq 4 \exp{- \frac{\e^2 t d}{9 c} }. 
\end{equation}

We would like to have the event described in \eqref{eq:pravgfixedstate} hold for all $\ket{\psi} \in AB$. For this, we construct a finite set $\cN$ of states (a $\delta$-net) for which we can ensure that $\frac{1}{t}\sum_{k=0}^{t-1} \tracedist{p^A_{U_k\ket{\psi}}, \unif([d_A])} < 2\e/3$ for all $\ket{\psi} \in \cN$ holds with high probability. See e.g., \cite[Lemma II.4]{HLSW04} for a proof.
\begin{lemma}[$\delta$-net]
\label{lem:eps-net}
Let $\delta \in (0,1)$. There exists a set $\cN$ of pure states in $\CC^d$ with $|\cN| \leq (3/\delta)^{2 d}$ such that for every pure state $\ket{\psi} \in \CC^d$ (i.e., $\| \ket{\psi} \|_2 = 1$), there exists $\ket{\tilde{\psi}} \in \cN$ such that
\[
\| \ket{\psi} - \ket{\tilde{\psi}} \|_{2} \leq \delta.
\] 
\end{lemma}

Let $\cN$ be the $\e/3$-net obtained by applying this lemma to the space $AB$ with $\delta = \e/3$. We have
\begin{align*}
\pr{ \exists \ket{\psi} \in \cN: \frac{1}{t}\sum_{k=0}^{t-1} \tracedist{p^A_{U_k\ket{\psi}}, \unif([d_A])} \geq 2\e/3 } &\leq |\cN| \cdot 4\exp{- \frac{\e^2 t d }{9c} } \\
																&\leq 4\exp{-d \left(\frac{\e^2 t}{9c}  - 2\ln(9/\e) \right)}.
\end{align*}
Now for an arbitrary state $\ket{\psi} \in AB$, we know that there exists $\ket{\tilde{\psi}} \in \cN$ such that $ \| \ket{\psi} - \ket{\tilde{\psi}} \|_{2} \leq \e/3$. As a consequence, for any unitary transformation $U$,
\begin{align*}
\tracedist{p^A_{U \ket{\psi}}, \unif([d_A])} &\leq \tracedist{ p^A_{U \ket{\tilde{\psi}}}, \unif([d_A])} +  \tracedist{p^A_{U \ket{\tilde{\psi}}}, p^A_{U \ket{\psi}}}   \\
							&\leq \tracedist{ p^A_{U \ket{\tilde{\psi}}}, \unif([d_A])} +  \| U \ket{\tilde{\psi}} - U \ket{\psi} \|_2 \\							 
							&\leq \tracedist{p^A_{U \ket{\tilde{\psi}}}, \unif([d_A])} + \e/3.
\end{align*}
In the first inequality, we used the triangle inequality and the second inequality can be derived as in \eqref{eq:ineq-tracedist-l2}.
Thus,
\begin{equation}
\label{eq:prur}
\pr{ \exists \ket{\psi} \in AB:  \frac{1}{t}\sum_{k=0}^{t-1} \tracedist{ p_{U_k \ket{\psi}}, \unif([d_A])} \geq \e } \leq 4\exp{-d \left(\frac{\e^2 t}{9 c}  - 2\ln(9/\e) \right)}.
\end{equation}
If $t > \frac{4 \cdot 18 c \cdot \ln(9/\e)}{\e^2}$, this bound is strictly smaller than $1/2$ and the result follows.

To prove that we can suppose that $\{U_0, \dots, U_{t-1}\}$ define $\gamma$-MUBs, consider the function $f : \ket{\ph} \mapsto \braket{\psi}{\ph}$ for some fixed vector $\ket{\psi}$. Then, if $\ket{\ph}$ is a random pure state, we have $\ex{f(\ket{\ph})} = 0$. Moreover, $f$ is $1$-Lipschitz. Thus, Lemma \ref{lem:conc-product-sphere} for $t = 1$, or more simply L\'{e}vy's Lemma \ref{lem:levy}, with $\delta = d^{-\gamma/2}$ gives
\[
\pr{ |\braket{\psi}{\ph}| \geq d^{-\gamma/2}} \leq 4 \exp{-\frac{d^{1-\gamma}}{c}}.
\]
As a result,
\begin{equation}
\label{eq:pr-approx-mub}
\pr{ \exists k \neq k', x ,y \in [d], |\bra{x}  U^{\dagger}_k U_{k'} \ket{y}| \geq d^{-\gamma}} \leq 4t^2 d^2 \exp{-\frac{d^{1-\gamma}}{c}}
\end{equation}
which completes the proof.
\end{proof}

\begin{corollary}[Existence of entropic uncertainty relations]
\label{cor:existence-entropic-ur}
Let $C$ be a Hilbert space of dimension $d > 2$. There exists a universal constant $c' \geq 1$ such that for any integer $t > 2$, there exists a set $\{U_0, \dots, U_{t-1}\}$ of unitary transformations of $C$ satisfying the following entropic uncertainty relation: for any state $\ket{\psi}$,
\[
\frac{1}{t} \sum_{k=0}^{t-1} \entH(p_{U_k \ket{\psi}}) \geq \left(1-\sqrt{\frac{c' \log t}{t}}\right) \log d - \log\left(\frac{18t}{c' \log t}\right) - \binent\left(\sqrt{\frac{c' \log t}{t}} \right) 
\]
where $\binent(\e)$ is the binary entropy function.
In particular, in the limit $d \to \infty$, we obtain the existence of a sequence of sets of $t$ bases satisfying
\[
\lim_{d \to \infty} \frac{\frac{1}{t} \sum_{k=0}^{t-1} \entH(p_{U_k \ket{\psi}})}{\log d} \geq 1-\sqrt{\frac{c' \log t}{t}}.
\]
\end{corollary}
\begin{myremark}[]
Recall that the bases (or measurements) that constitute the uncertainty relation are defined as the images of the computational basis by $U^{\dagger}_k$. 
Note that as argued previously, we can always choose a state $\ket{\psi}$ that is a basis vector for one of the bases $U_0, \dots, U_{t-1}$. Thus, for any set of unitaries $\{U_0, \dots, U_{t-1}\}$, there exists $\ket{\psi}$ such that
\[
\frac{1}{t} \sum_{k=0}^{t-1} \entH(p_{U_k \ket{\psi}}) \leq \left(1- \frac{1}{t} \right) \log d.
\]
It is an open question whether there exists uncertainty relations matching this bound, even asymptotically as $d \to \infty$ \cite{WW10}. \citeN{WW10} ask whether there even exists a function $f$ growing to $+\infty$ such that  
\[
\lim_{d \to \infty} \frac{1}{t} \frac{\sum_{k=0}^{t-1} \entH(p_{U_k \ket{\psi}})}{\log d} \geq 1- \frac{1}{f(t)}.
\]
The corollary answers this question in the affirmative with $f(t) = \sqrt{\frac{t}{c' \log t}}$. 
\end{myremark}
\begin{proof}
Define $c' = 100 c$ where $c$ comes from Lemma \ref{lem:conc-product-sphere}, $\e = \sqrt{\frac{c' \log t}{t}}$ and decompose $C = A \otimes B$ with $d_B = \ceil{9/\e^2}$. 
Applying Theorem \ref{thm:existence-ur} and observing that our choice of $\e$ is such that $t > 72 c \log(9/\e)/\e^2$, we get a family $U_0, \dots, U_{t-1}$ of unitary transformations that satisfies
\[
\frac{1}{t} \sum_{k=0}^{t-1} \tracedist{p^A_{U_k \ket{\psi}}, \unif([d_A])} \leq \e.
\]
By Proposition \ref{prop:metric-to-entropic}, these unitary transformations also satisfy an entropic uncertainty relation:
\begin{align*}
\frac{1}{t} \sum_{k=0}^{t-1} \entH(p^A_{U_k \ket{\psi}}) &\geq (1-\e) \log \left(\frac{d}{\ceil{9/\e^2}}\right) - \binent(\e) \\
											&\geq (1-\e) \log d - \log(18/\e^2) - \binent(\e).
\end{align*}
\end{proof}


\subsection{Metric uncertainty relations: explicit construction}
\label{sec:explicit-ur}

In this section, we are interested in obtaining families $\{U_0, \dots, U_{t-1}\}$ of unitaries satisfying metric uncertainty relations where $U_0, \dots, U_{t-1}$ are explicit and efficiently computable using a quantum computer. For this section, we consider for simplicity a Hilbert space composed of qubits, i.e., of dimension $d = 2^n$ for some integer $n$. This Hilbert space is of the form $A \otimes B$ where $A$ describes the states of the first $\log d_A$ qubits and $B$ the last $\log d_B$ qubits. Note that we assume that both $d_A$ and $d_B$ are powers of two.


We construct a set of unitaries by adapting an explicit low-distortion embedding of $(\RR^d, \ell_2)$ into $(\RR^{d'}, \ell_1)$ with $d' = d^{1+o(1)}$ found by \citeN{Ind07}. The construction has two main ingredients: a set of mutually unbiased bases and an extractor. \red{More      specifically the embedding consists of repeated applications of the following procedure. The input vector is encoded into a small number of mutually unbiased bases from \cite{HSP06} and all these encodings are concatenated. The components of this longer vector are then shuffled in a way specified by an extractor construction due to \cite{Zuc97} so that the coefficients on some fixed known coordinates are flat. In the following step, the same procedure is applied on the remaining coordinates that do not necessarily have flat coefficients.} Our construction uses the same paradigm while requiring additional properties of both the mutually unbiased bases and the extractor. 

In order to obtain a locking scheme that only needs simple quantum operations, we construct sets of \emph{approximately} mutually unbiased bases from a restricted set of unitaries that can be implemented with single-qubit Hadamard gates. Moreover, we impose three additional properties on the extractor: we need our extractor to be strong, to define a permutation and to be efficiently invertible. An extractor is said to be strong if the output is close to uniform even given the seed. We want the extractor to be strong because we are constructing metric uncertainty relations as opposed to a norm embedding. The property of being a permutation extractor is needed to ensure that the induced transformation on $(\CC^{2})^{\otimes n}$ preserves the $\ell_2$ norm. We also require the efficient invertibility condition to be able to build an efficient quantum circuit for the permutation. See Definition \ref{def:perm-extractor} for a precise formulation.

The intuition behind Indyk's construction is as follows. Let $V_0, \dots, V_{r-1}$ be unitaries defining (approximately) mutually unbiased bases (defined in equation \eqref{eq:def-approx-mub}) and let $\{P_y\}_{y \in S}$ be a permutation extractor (defined later in Definition \ref{def:perm-extractor}). The role of the mutually unbiased bases is to guarantee that for all states $\ket{\psi}$ and for most values of $j \in [r]$, most of the mass of the state $V_j\ket{\psi}$ is ``well spread'' in the computational basis. This spread is measured in terms of the min-entropy of the distribution $p_{V_j \ket{\psi}}$. Then, the extractor $\{P_y\}_y$ will ensure that on average over $y \in S$, the masses $\sum_{b} |\bra{a} \bra{b} P_y V_j \ket{\psi}|^2$ are almost equal for all $a \in [d_A]$. More precisely, the distribution $p^A_{P_y V_j \ket{\psi}}$ is close to uniform.

%
%

As shown in the following lemma, there is a construction of mutually unbiased bases that can be efficiently implemented \cite{WF89}. The proof of the lemma is deferred to Appendix \ref{sec:app-explicitmub}.
%
%
\begin{lemma}[Quantum circuits for MUBs]
\label{lem:explicitmub}
Let $n$ be a positive integer and $d = 2^n$. For any integer $r \leq d+1$, there exists a family $V_0, \dots,  V_{r-1}$ of unitary transformations of $\CC^d$ that define mutually unbiased bases. Moreover, there is a randomized classical algorithm with runtime $O(n^2 \polylog n)$ that takes as input $j \in [r]$ and outputs a binary vector $\alpha_j \in \{0,1\}^{2n-1}$, and a quantum circuit of size $O(n \polylog n)$ that when given as input the vector $\alpha_j$ (classical input) and a quantum state $\ket{\psi} \in \CC^d$ outputs $V_j \ket{\psi}$.
\end{lemma}
\begin{myremark}[]
The randomization in the algorithm is used to find an irreducible polynomial of degree $n$ over $\FF_2[X]$. It could be replaced by a deterministic algorithm that runs in time $O(n^4 \polylog n)$. Observe that if $n$ is odd and $r \leq (d+1)/2$, it is possible to choose the unitary transformations to be real (see \cite{HSP06}).
\end{myremark}
%
%
It is also possible to obtain approximately mutually unbiased bases that use smaller circuits. In fact, the following lemma shows that we can construct large sets of approximately mutually unbiased bases defined by unitaries in the restricted set
\[
\cH = \{ H^{v} \eqdef H^{v_1} \otimes \dots \otimes H^{v_n}, v \in \{0,1\}^n \},
\]
where $H$ is the Hadamard transform on $\CC^2$ defined by 
\[
H = \frac{1}{\sqrt{2}} \left( \begin{array}{cc}
1 & 1 \\
1 & -1 
\end{array} \right).
\]
In our construction of metric uncertainty relations (Theorem \ref{thm:explicit-ur1}), we could use the $1$-MUBs of Lemma \ref{lem:explicitmub} or the $(1/2-\delta)$-MUBs of Lemma \ref{lem:approx-mub}. As the construction of approximate MUBs is simpler and can be implemented with simpler circuits, we will mostly be using Lemma \ref{lem:approx-mub}.
%
%
\begin{lemma}[Approximate MUBs in $\cH$]
\label{lem:approx-mub}
Let $n'$ be a positive integer and $n = 2^{n'}$. 
\begin{enumerate}
\item For any integer $r \leq n$, there exists a family $V_0, \dots,  V_{r-1} \in \cH$ that define $1/2$-MUBs.
\item For any $\delta \in (0,1/2)$, there exists a constant $c > 0$ independent of $n$ such that for any $r \leq 2^{c n}$  there exists a family $V_0, \dots,  V_{r-1}$ of unitary transformations in $\cH$ that define $(1/2-\delta)$-MUBs.
\end{enumerate}
Moreover, in both cases, given an index $j \in [r]$, there is a polynomial time (classical) algorithm that computes the vector $v \in \{0,1\}^n$ that defines the unitary $V_j = H^{v}$.
\end{lemma}
\begin{proof}
Observe that for any $v, v'$ and $x, y \in \{0,1\}^n$
\begin{equation}
\label{eq:mub-had}
|\bra{x} H^{v} H^{v'} \ket{y}| = \prod_{j} \bra{x_j} H^{v_j} H^{v'_j} \ket{y_j} \leq \frac{1}{2^{\hamdist(v,v')/2}}.
\end{equation}
Using this observation, we see that a binary code $C \subseteq \{0,1\}^n$ with minimum distance $\gamma n$ defines a set of $\gamma$-MUBs in $\cH$. It is now sufficient to find binary codes with minimum distance as large as possible. 
For the first construction, we use the Hadamard code that has minimum distance $n/2$. The Hadamard codewords are indexed by $x \in \{0,1\}^{n'}$; the codeword corresponding to $x$ is the vector $v \in \{0,1\}^n$ whose coordinates are $v_z = x \cdot z$ for all $z \in \{0,1\}^{n'}$. This code has the largest possible minimum distance for a non-trivial binary code but its shortcoming is that the number of codewords is only $n$. For our applications, it is sometimes desirable to have $r$ larger than $n$ (this is useful to allow the error parameter $\e$ of our metric uncertainty relation to be smaller than $n^{-1/2}$).

For the second construction, we use families of linear codes with minimum distance $(1/2 - \delta)n$ with a number of codewords that is exponential in $n$. For this, we can use Reed-Solomon codes concatenated with linear codes on $\{0,1\}^{\Theta(n')}$ that match the performance of random linear codes; see for example Appendix E in \cite{Gol08}. For a simpler construction, note that we can also get $2^{\Omega(\sqrt{n})}$ codewords by using a Reed-Solomon code concatenated with a Hadamard code.
\end{proof}

The next lemma shows that for any state $\ket{\psi}$, for most values of $j$, the distribution $p_{V_j \ket{\psi}}$ is close to a distribution with large min-entropy provided $\{V_j\}$ define $\gamma$-MUBs. This result might be of independent interest. In fact, \citeN{DFRSS07} prove a lower bound close to $n/2$ on the min-entropy of a measurement in the computational basis of the state $U \ket{\psi}$ where $U$ is chosen uniformly from the full set of the $2^n$ unitaries of $\cH$. They leave as an open question the existence of small subsets of $\cH$ that satisfy the same uncertainty relation. When used with the $\gamma$-MUBs of Lemma \ref{lem:approx-mub}, the following lemma partially answers this question by exhibiting such sets of size polynomial in $n$ but with a min-entropy lower bound close to $n/4$ instead. 
%
%
\begin{lemma}
\label{lem:minentropymub}
Let $n \geq 1, d = 2^n$ and $\e \in (0,1)$ and consider a set of $r = \ceil{\frac{2}{\e^2}}$ unitary transformations $V_0, \dots, V_{r-1}$ of $\CC^d$ defining $\gamma$-MUBs. For all $\ket{\psi} \in \CC^d$,
\[
\left| \left\{ j \in [r] : \exists \text{ distribution } q_j, \tracedist{p_{V_j \ket{\psi}}, q_j} \leq \e \text{ and } \entHmin(q_j) \geq \frac{\gamma n}{2} - 2 \right\} \right| \geq (1-\e) r.
\]
\end{lemma}
\begin{proof}
This proof proceeds along the lines of \cite[Lemma 4.2]{Ind07}. Similar results can also be found in the sparse approximation literature; see \cite[Proposition 4.3]{Tro04} and references therein. 

Consider the $rd \times d$ matrix $V$ obtained by concatenating the rows of the matrices $V_0, \dots, V_{r-1}$. For $S \subseteq [rd]$, $V_S$ denotes the submatrix of $V$ obtained by selecting the rows in $S$. The coordinates of the vector $V \ket{\psi} \in \CC^{rd}$ are indexed by $z \in [rd]$ and denoted by $(V \ket{\psi})_z$.
\begin{claim}[]
We have for any set $S \subseteq [rd]$ of size at most $d^{\gamma/2}$ and any unit vector $\ket{\psi}$, then
\begin{equation}
\label{eq:upper-bound-vs}
\| (V \ket{\psi})_S \|_2^2 \leq 1 + \frac{ |S| }{d^{\gamma/2} }.
\end{equation}
\end{claim}
To prove the claim, we want an upper bound on the operator $2$-norm of the matrix $(V_S)$, which is the square root of the largest eigenvalue of $G = V_S V^{\dagger}_S$. As two distinct rows of $V$ have an inner product bounded by $\frac{1}{d^{\gamma/2}}$, the non-diagonal entries of $G$ are bounded by $\frac{1}{d^{\gamma/2}}$. Moreover, the diagonal entries of $G$ are all $1$. By the Gershgorin circle theorem, all the eigenvalues of $G$ lie in the disc centered at $1$ of radius $\frac{|S|-1}{d^{\gamma/2}}$. We conclude that \eqref{eq:upper-bound-vs} holds.

Now pick $S$ to be the set of indices of the $d^{\gamma/2}$ largest entries of the vector $\{|(V \ket{\psi})_z|^2\}_{z \in [rd]}$. Using the previous claim, we have
$\| (V \ket{\psi})_S \|_2^2 \leq 2.$
Moreover, since $S$ contains the $d^{\gamma/2}$ largest entries of $\{|(V \ket{\psi})_z|^2\}_z$, we have that for all $z \notin S$, $|(V\ket{\psi})_z|^2 d^{\gamma/2} \leq \| (V \ket{\psi})_S \|_2^2 \leq 2$. Thus, for all $z \notin S$, $|(V\ket{\psi})_z|^2 \leq \frac{2}{d^{\gamma/2}}$.

We now build the distributions $q_j$. For every $j \in [r]$, define 
\[
w_j = \sum_{z \in S \cap \{jr, \dots, (j+1)r-1\}} |(V\ket{\psi})_z|^2,
\]
which is the total weight in $S$ of $V_j \ket{\psi}$. Defining $T_{\e} = \{j : w_j > \e\}$, we have $|T_{\e}| \e \leq \| (V \ket{\psi})_S \|^2_2 \leq 2$. Thus,
\[
|T_{\e}| \leq 2/\e \leq \e r.
\]
We define the distribution $q_j$ for $j \in [r]$ by
\[
q_j(x) = \left\{ \begin{array}{ll}
|\bra{x} V_j \ket{\psi}|^2 + \frac{w_j}{d} & \textrm{if $jd + x \notin S$}\\
\frac{w_j}{d} & \textrm{if $jd + x \in S$}.
\end{array} \right. 
\]
Since
\[
\sum_x q_j(x) = w_j + \sum_{x \in [d] : jd+x \notin S} |\bra{x} V_j \ket{\psi}|^2 = \sum_{x \in [d]} |\bra{x} V_j \ket{\psi}|^2 = 1,
\]
$q_j$ is a probability distribution. Moreover, we have that for $j \notin T_{\e}$
\[
\tracedist{p_{V_j \ket{\psi}}, q_j} \leq \frac{1}{2} \left( \sum_{x : jd+x \notin S} \frac{w_j}{d} + \sum_{x : jd+x \in S} \left( \frac{w_j}{d} + |\bra{x} V_j \ket{\psi}|^2 \right) \right) = w_j \leq \e.
\]
The distribution $q_j$ also has the property that 
for all $x \in [d]$, $q_j(x) \leq \frac{2}{d^{\gamma/2}} + \frac{1}{d} \leq \frac{3}{d^{\gamma/2}}$. In other words, $\entHmin(q_j) \geq \frac{\gamma n}{2} - 2$.
\end{proof}

We now move to the second building block in Indyk's construction: randomness extractors.
Randomness extractors are functions that extract uniform random bits from weak sources of randomness.

\begin{definition}[Strong permutation extractor]
\label{def:perm-extractor}
Let $n$ and $m \leq n$ be positive integers, $\ell \in [0,n]$ and $\e \in (0,1)$. A family of permutations $\{P_y\}_{y \in S}$ of $\{0,1\}^n$ where each permutation $P_y$ is described by two functions $P^E_y: \{0,1\}^n \to \{0,1\}^m$ (the first $m$ output bits of $P_y$) and $P^R_y: \{0,1\}^n \to \{0,1\}^{n-m}$ (the last $n-m$ output bits of $P_y$) is said to be an explicit $(n, \ell) \to_{\e} m$ strong permutation extractor if:
\begin{itemize}
\item For any random variable $X$ on $\{0,1\}^n$ such that $\entHmin(X) \geq \ell$, and an independent seed $U_S$ uniformly distributed over $S$, we have
\[
\tracedist{p_{\left(U_S, P_{U_S}^E(X)\right)}, \unif( S \times \{0,1\}^m)} \leq \e,
\]
which is equivalent to
\begin{equation}
\label{eq:strongext}
\frac{1}{|S|}\sum_{y \in S} \tracedist{p_{P_y^E(X)}, \unif(\{0,1\}^m)} \leq \e.
\end{equation}
\item For all $y \in S$, both the function $P_y$ and its inverse $P_y^{-1}$ are computable in time polynomial in $n$.
\end{itemize}
\end{definition}
\begin{myremark}[]
The usual definition of extractors does not require the functions to be permutations or to be efficiently invertible. But for this paper, these conditions are needed.
A similar definition of permutation extractors was used in \cite{RVW00} in order to avoid some entropy loss in an extractor construction. Here, the reason we use permutation extractors is different; it is because we want the induced transformation $P_y$ on $\CC^{2^n}$ to preserve the $\ell_2$ norm.
\end{myremark}

We can adapt an extractor construction of \cite{GUV09} to obtain a permutation extractor with the following parameters. The details of the construction are presented in Appendix \ref{sec:app-perm-extractor}.


\begin{theorem}[Explicit strong permutation extractors]
\label{thm:perm-extractor}
For all (constant) $\delta \in (0,1)$, all positive integers $n$, all $\ell \in [c \log(n/\e), n]$ ($c$ is a constant independent of $n$ and $\e$), and all $\e \in (0,1/2)$, there is an explicit $(n, \ell) \to_{\e} (1-\delta) \ell$ strong permutation extractor $\{P_y\}_{y \in S}$ with $\log |S| \leq O(\log(n/\e))$. Moreover, the functions $(x,y) \mapsto P_y(x)$ and $(x,y) \mapsto P^{-1}_y(x)$ can be computed by circuits of size $O(n \polylog(n/\e))$.
\end{theorem}

A permutation $P$ on $\{0,1\}^n$ defines a unitary transformation on $(\CC^2)^{\otimes n}$ that we also call $P$.  The permutation extractor $\{P_y\}$ will be seen as a family of unitary transformations over $n$ qubits. Moreover, just as we decomposed the space $\{0,1\}^n$ into the first $m$ bits and the last $n-m$ bits, we decompose the space $(\CC^2)^{\otimes n}$ into $A \otimes B$, where $A$ represents the first $m$ qubits and $B$ represents the last $n-m$ qubits. The properties of $\{P^E_y\}$ will then be reflected in the system $A$.

\red{Recall that Lemma \ref{lem:minentropymub} stated that after a measurement in a randomly selected basis from an approximately mutually unbiased set, the min-entropy of the outcome is quite large. This means that if we follow this measurement with an extractor, the output will be almost uniform, which is exactly what we need. Thus, using the construction of Theorem \ref{thm:perm-extractor}, we obtain a set of unitaries satisfying a metric uncertainty relation.}
%
%
\begin{theorem}[Explicit uncertainty relations: key optimized]
\label{thm:explicit-ur1}
Let $\delta > 0$ be a constant, $n$ be a positive integer, $\e \in (2^{-c'n},1)$ ($c'$ is a constant independent of $n$). Then, there exist $t \leq \left(\frac{n}{\e}\right)^{c}$ (for some constant $c$ independent of $n$ and $\e$) unitary transformations $U_0, \dots, U_{t-1}$ acting on $n$ qubits such that: if $A$ represents the first $(1-\delta)n/4 - O(1)$ qubits and $B$ represents the remaining qubits, then for all $\ket{\psi} \in AB$,
\[
\frac{1}{t}\sum_{k=0}^{t-1} \tracedist{p^A_{U_k \ket{\psi}}, \unif([d_A])} \leq \e.
\]
Moreover, the mapping that takes the index $k \in [t]$ and a state $\ket{\psi}$ as inputs and outputs the state $U_k \ket{\psi}$ can be performed by a classical computation with polynomial runtime and a quantum circuit that consists of single-qubit Hadamard gates on a subset of the qubits followed by a permutation in the computational basis. This permutation can be computed by (classical or quantum) circuits of size $O(n \polylog(n/\e))$.
\end{theorem}
\begin{myremark}[]
Observe that in terms of the dimension $d$ of the Hilbert space, the number of unitaries $t$ is polylogarithmic.
\end{myremark}
\begin{proof}
Let $\e' = \e/6$.
Lemma \ref{lem:approx-mub} gives $r = \ceil{2/\e'^2}$ unitary transformations $V_0, \dots, V_{r-1}$ that define $\gamma$-mutually unbiased bases with $\gamma = 1/2 - \delta/4$. Moreover, all theses unitaries can be performed by a quantum circuit that consists of single-qubit Hadamard gates on a subset of the qubits. Theorem \ref{thm:perm-extractor} with $\ell = (1-\delta/2)n/4 - 2$ and error $\e'$ gives $|S| \leq 2^{c\log(n/\e')}$ permutations $\{P_y\}_{y \in S}$ of $\{0,1\}^n$ that define an $(n, \ell) \mapsto_{\e'} (1-\delta/2)\ell$ extractor and are computable by classical circuits of size $O(n \polylog(n/\e))$. We now argue that this classical circuit can be used to build a quantum circuit of size $O(n \polylog n)$ that computes the unitaries $P_y$. 

Given classical circuits that compute $P$ and $P^{-1}$, we can construct reversible circuits $C_P$ and $C_{P^{-1}}$ for $P$ and $P^{-1}$. The circuit $C_P$ when given input $(x,0)$ outputs the binary string $(x, P(x))$, so that it keeps the input $x$. Such a circuit can readily be transformed into a quantum circuit that acts on the computational basis states as the classical circuit. We also call these circuits $C_P$ and $C_{P^{-1}}$. Observe that we want to compute the unitary $P$, so we have to erase the input $x$. For this, we combine the circuits $C_P$ and $C_{P^{-1}}$ as described in Figure \ref{fig:qcircuit}. Note that the size of this quantum circuit is the same as the size of the original classical circuit up to some multiplicative constant. Thus, this quantum circuit has size $O(n \polylog (n/\e))$.
\begin{center}
\begin{figure}[t]
\begin{center}
\includegraphics[width=0.8\textwidth]{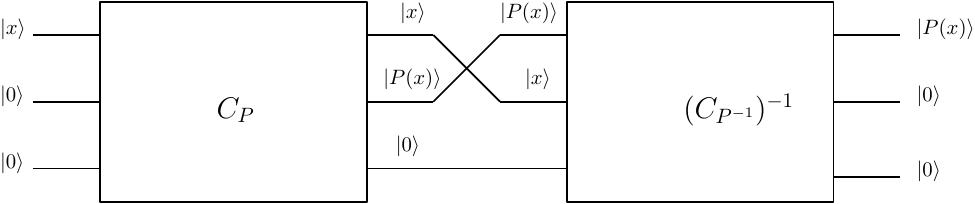}
\end{center}
\caption{Quantum circuit to compute the permutation $P$ using quantum circuits $C_P$ for $P$ and $C_{P^{-1}}$ for $P^{-1}$. $(C_{P^{-1}})^{-1}$ is simply the circuit $C_{P^{-1}}$ taken backwards. The bottom register is an ancilla register.}
\label{fig:qcircuit}
\end{figure}
\end{center}

The unitaries $\{U_0, \dots, U_{t-1}\}$ are obtained by taking all the possible products $P_y V_j$ for $j \in [r], y \in S$. Note that $t = r|S|$.
We now show that the set $\{U_0, \dots, U_{t-1}\}$ satisfies the uncertainty relation property. Using Lemma \ref{lem:minentropymub}, for any state $\ket{\psi}$, the set 
\[
T_{\ket{\psi}} \eqdef \left\{ j : \exists q_j, \tracedist{p_{V_j \ket{\psi}}, q_j} \leq \e' \text{ and }  \entHmin(q_j) \geq (1-\delta/2)n/4 - 2\right\}
\]
has size at least $(1-\e') r$. Moreover, for all $a \in [d_A]$, $p^A_{P_y V_i \ket{\psi}}(a) = \sum_{b} |\bra{a} \bra{b} P_y V_i \ket{\psi}|^2 = \pr{P^E_y(X) = a}$ where $X$ has distribution $p_{V_i \ket{\psi}}$. By definition, for $i \in T_{\ket{\psi}}$, we have $\tracedist{p_{V_i \ket{\psi}}, q_i} \leq \e'$ with $\entHmin(q_i) \geq (1-\delta/2)n/4 - 2$. Using the fact that $\{P^E_y\}$ is a strong extractor (see \eqref{eq:strongext}) for min-entropy $(1-\delta/2)n/4 - 2$, it follows from the triangle inequality that
\[
\frac{1}{|S|}\sum_{y \in S} \tracedist{p^A_{P_y V_i \ket{\psi}}, \unif([d_A])} \leq 2\e'
\]
%
%
for all $i \in T_{\ket{\psi}}$. As $|T_{\ket{\psi}}| \geq (1-\e')r$, we obtain
\[
\frac{1}{t}\sum_{k=0}^{t-1} \tracedist{p^A_{U_k \ket{\psi}}, \unif([d_A])} \leq 3\e' = \e/2.
\]
To conclude, we show that $t$ can be taken to be a power of two at the cost of multiplying the error by at most two. In fact, let $p$ be the smallest integer satisfying $t \leq 2^p$, so that $2^p \leq 2t$. By choosing an arbitrary subset of $2^p - t$ unitaries and duplicating them, we obtain a set of $2^p$ unitaries. It is easily seen that we obtain an $\e$-metric uncertainty relation with $2^p$ unitaries from an $\e/2$-metric uncertainty relation with $t$ unitaries.
%
%
%
%
\end{proof}

Note that the $B$ system we obtain is quite large and to get strong uncertainty relations, we want the system $B$ to be as small as possible. For this, it is possible to repeat the construction of the previous theorem on the $B$ system. The next theorem gives a construction where the $A$ system is composed of $n-O(\log \log n) - O(\log(1/\e))$ qubits. Of course, this is at the expense of increasing the number of unitaries in the uncertainty relation.

\begin{theorem}[Explicit uncertainty relation: message length optimized]
\label{thm:explicit-ur2}
Let $n$ be a positive integer and $\e \in (2^{-c'n}, 1)$ where $c'$ is a constant independent of $n$. Then, there exist $t \leq \left(\frac{n}{\e}\right)^{c \log n}$ (for some constant $c$ independent of $n$ and $\e$) unitary transformations $U_0, \dots, U_{t-1}$ acting on $n$ qubits that are all computable by quantum circuits of size $O(n \polylog(n/\e))$ such that: if $A$ represents the first $n - O(\log \log n) - O( \log(1/\e))$ qubits and $B$ represents the remaining qubits, then for all $\ket{\psi} \in AB$,
\begin{equation}
\label{eq:cond-explicit-ur2}
\frac{1}{t} \sum_{k=0}^{t-1} \tracedist{p^A_{U_k \ket{\psi}}, \unif([d_A])} \leq \e.
\end{equation}
Moreover, the mapping that takes the index $k \in [t]$ and a state $\ket{\psi}$ as inputs and outputs the state $U_k \ket{\psi}$ can be performed by a classical precomputation with polynomial runtime and a quantum circuit of size $O(n \polylog(n/\e))$. The number of unitaries $t$ can be taken to be a power of two.
\end{theorem}
\red{
\begin{myremark}[]
Later in Theorem \ref{thm:identification}, we will need to apply different $U_k$'s in superposition. In this case, to analyze the size of the quantum circuit, we need to bound the running time of the classical computation that describes the quantum circuit for $U_k$ given $k$. For this, it is simpler to use the mutually unbiased bases construction of Lemma \ref{lem:explicitmub}, for which we can bound this running time by $O(n^2 \polylog n)$.
\end{myremark}
}
\begin{proof}
Using the construction of Theorem \ref{thm:explicit-ur1}, we obtain a system $A$ over which we have some uncertainty relation and a system $B$ that we do not control. In order to decrease the dimension of the system $B$, we can apply the same construction to that system. The system $B$ then gets decomposed into $A_2B_2$, and we know that the distribution of the measurement outcomes of system $A_2$ in the computational basis is close to uniform. As a result, we obtain an uncertainty relation on the system $AA_2$ (see Figure \ref{fig:composition-indyk}).

\begin{center}
\begin{figure}[tp]
\begin{center}
\includegraphics[width=0.6\textwidth]{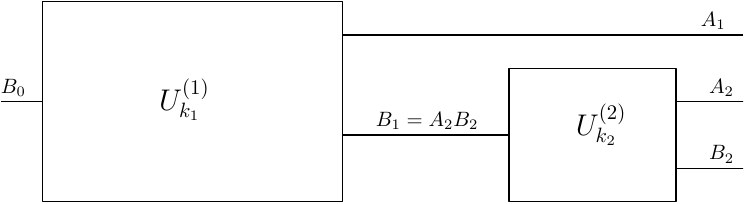}
\end{center}
\caption{Composition of the construction of Theorem \ref{thm:explicit-ur1}: In order to reduce the dimension of the $B$ system, we can re-apply the uncertainty relation to the $B$ system.}
\label{fig:composition-indyk}
\end{figure}
\end{center}

More precisely, we start by demonstrating a simple property about the composition of metric uncertainty relations. Note that this composition is different from the one described in \eqref{eq:parallel-compose}, but the proof is quite similar.
\begin{claim}[]
Suppose the set $\{U^{(1)}_{0}, \dots, U^{(1)}_{t_1-1}\}$ of unitaries on $A_1B_1$ satisfies a $(t_1, \e_1)$-metric uncertainty relation on system $A_1$ and the $\{U^{(2)}_{0}, \dots, U^{(2)}_{t_2-1}\}$ of unitaries on $B_1 = A_2B_2$ satisfies a $(t_2, \e_2)$-metric uncertainty relation on $A_2$. 
Then the set of unitaries $\left\{(\1^{A_1} \ox U^{(2)}_{k_2}) \cdot U^{(1)}_{k_1} \right\}_{k_1, k_2 \in [t_1] \times [t_2]}$ satisfies a $(t_1t_2, \e_1 + \e_2)$-metric uncertainty relation on $A_1 A_2$: for all $\ket{\psi} \in A_1A_2B_2$,
\[
\frac{1}{t_1t_2} \sum_{k_1, k_2 \in [t_1] \times [t_2]} \tracedist{ p_{U^{(2)}_{k_2} U^{(1)}_{k_1}\ket{\psi}}, \unif([d_{A_1}d_{A_2}]) } \leq \e_1 + \e_2.
\]
\end{claim}
We now prove the claim.
For a fixed value of $k_1 \in [t_1]$ and $a_1 \in [d_{A_1}]$, we can apply the second uncertainty relation to the state $\frac{\bra{a_1}^{A_1} U_{k_1} \ket{\psi}}{\| \bra{a_1}^{A_1} U_{k_1} \ket{\psi} \|_2} = \frac{1}{\sqrt{p^{A_1}_{U_{k_1} \ket{\psi}}(a_1)}} \sum_{b_1} \left(\bra{a_1} \bra{b_1} U_{k_1} \ket{\psi} \right) \ket{b_1} \in B_1 = A_2B_2$. As $\{\ket{b_1}\}_{b_1} = \{\ket{a_2} \ket{b_2} \}_{a_2, b_2}$, we have
\begin{align*}
\frac{1}{t_2}\sum_{k_2} \sum_{a_2} \left| \frac{1}{p^{A_1}_{U_{k_1} \ket{\psi}}(a_1)}\sum_{b_2} | \bra{a_1}^{A_1} \bra{a_2}^{A_2} \bra{b_2}^{B_2} (\1^{A_1} \ox U_{k_2}) U_{k_1} \ket{\psi} |^2 - \frac{1}{d_{A_2}} \right| \leq \e_2. 
\end{align*}
We can then calculate, in the same vein as \eqref{eq:parallel-compose2}
\begin{align*}
&\frac{1}{t_1t_2}\sum_{k_1, k_2} \sum_{a_1, a_2} \left| \sum_{b_2} | \bra{a_1}^{A_1} \bra{a_2}^{A_2} \bra{b_2}^{B_2} (\1^{A_1} \ox U_{k_2}) U_{k_1} \ket{\psi} |^2 - \frac{1}{d_{A_1} d_{A_2}} \right| \\
	&\leq \frac{1}{t_1t_2} \sum_{k_1, k_2} \sum_{a_1} 
	\left| \sum_{b_2} | \bra{a_1}^{A_1} \bra{a_2}^{A_2} \bra{b_2}^{B_2} (\1^{A_1} \ox U_{k_2}) U_{k_1} \ket{\psi} |^2 - \frac{p^{A_1}_{U_{k_1} \ket{\psi}}(a_1)}{d_{A_2}} \right| \\
	&+ \frac{1}{t_1}\sum_{k_1} \sum_{a_1, a_2} \left| \frac{p^{A_1}_{U_{k_1} \ket{\psi}}(a_1)}{d_{A_2}} - \frac{1}{d_{A_1} d_{A_2}} \right|  \\
	& \leq \frac{1}{t_1} \sum_{k_1} \sum_{a_1} p^{A_1}_{U_{k_1} \ket{\psi}}(a_1) \e_2 + \e_1 \\
	&\leq \e_2+ \e_1.
\end{align*}
This completes the proof of the claim.

To obtain the claimed dimensions, we compose the construction of Theorem \ref{thm:explicit-ur1} $h$ times with an error parameter $\e' = \e/h$ and $\delta = 1/2$. Starting with a space of $n$ qubits, the dimension of the $B$ system (after one step) can be bounded by
\[
\frac{7}{8}n \leq \log d_B \leq \frac{7}{8}n + O(1).
\]
So after $h$ steps, we have
\[
\left(7/8\right)^h n  \leq \log d_{B_h} \leq \left(7/8\right)^h n + O(1).
\]
Note that $h$ cannot be arbitrarily large: in order to apply the construction of Theorem \ref{thm:explicit-ur1} on a system of $m$ qubits with error $\e'$, we should have $\e' \geq 2^{-c'm}$. In other words, if 
\begin{equation}
\label{eq:cond-repeat-cons}
\log d_{B_h} \geq \frac{1}{c'} \log(h/\e),
\end{equation}
then we can apply the construction $h$ times. 
%
Let $c''$ be a constant to be chosen later and $h = \floor{\frac{1}{\log(8/7)} \left(\log n - 
\log(c''\log \log n + c''\log(1/\e)) \right)}$.
First, this choice of $h$ satisfies \eqref{eq:cond-repeat-cons}:
\begin{align*}
\log d_{B_h} &\geq c'' \log \log n + c'' \log(1/\e) \\
			&\geq \frac{1}{c'} \log(h/\e)
\end{align*}
if $c''$ is chosen large enough. Moreover, we get 
\[
\log d_{B_h} = 2^{-\log n} \cdot 2^{ \log O\left(\log \log n + \log(1/\e)\right) } \cdot n = O( \log \log n + \log(1/\e))
\]
as stated in the theorem.

Each unitary of the obtained uncertainty relation is a product of $h$ unitaries each obtained from Theorem \ref{thm:explicit-ur1}. The overall number of unitaries is the product of the number of unitaries for each of the $h$ steps. As a result, we have $t \leq \left(\frac{n}{\e}\right)^{c \log n}$ for some constant $c$. $t$ can be taken to be a power of two as the number of unitaries at each step can be taken to be a power of two. As for the running time, every unitary transformation of the uncertainty relation can be computed by a quantum circuit of size $O(n \polylog n)$ as it is a product of $O(\log n)$ unitaries each computed by a quantum circuit of size $O( n \polylog n)$.
\end{proof}

It is of course possible to obtain a trade-off between the key size and the dimension of the $B$ system by choosing the number of times the construction of Theorem \ref{thm:explicit-ur1} is applied. \red{For example, in the next corollary, we choose the number of repetitions of the construction of Theorem \ref{thm:explicit-ur1} so that we get an average entropy of $(1-\e)n$ while keeping the number of unitaries polynomial. Note that if we are concerned with (Shannon) entropic uncertainty relations, there is not much to be gained from repeating the construction  $n^{\Omega(\log n)}$ times as in Theorem \ref{thm:explicit-ur2} because we always lose a multiplicative factor $(1-\e)$ from Fannes' inequality.}

\begin{corollary}[Explicit entropic uncertainty relations]
\label{cor:explicit-entropic-ur}
Let $n \geq 100$ be an integer, and $\e \in ( 10n^{-1/2}, 1)$. Then, there exists $t \leq \left(\frac{n}{\e}\right)^{c \log (1/\e)}$ (for some constant $c$ independent of $n$ and $\e$) unitary transformations $U_0, \dots, U_{t-1}$ acting on $n$ qubits that are all computable by quantum circuits of size $O(n \polylog n)$ satisfying an entropic uncertainty relation: for all pure states $\ket{\psi} \in \left(\CC^2\right)^{ \otimes n} $,
\begin{equation}
\frac{1}{t} \sum_{k=0}^{t-1} \entH(p_{U_k \ket{\psi}}) \geq (1 - 2\e) n - \binent(\e)
\end{equation}
where $\binent$ is the binary entropy function.
Moreover, the mapping that takes the index $k \in [t]$ and a state $\ket{\psi}$ as inputs and outputs the state $U_k \ket{\psi}$ can be performed by a classical precomputation with polynomial runtime and a quantum circuit of size $O(n \polylog (n/\e))$. The number of unitaries $t$ can be taken to be a power of two.
\end{corollary}
\begin{proof}
The proof is basically the same as the proof of Theorem \ref{thm:explicit-ur2}, except that we repeat the construction $h = \Theta(\log(1/\e)/\log(8/7))$ times. We thus have 
\[
\log d_{B_h} \leq \left(7/8\right)^{h} n + O(1) \leq \e n.
\]
We obtain a set of $t \leq \left(\frac{n}{\e}\right)^{c\log(1/\e)}$ unitary transformations. Applying Proposition \ref{prop:metric-to-entropic}, we get
\begin{align*}
\frac{1}{t} \sum_{i=0}^{t-1} \entH(p_{U_k \ket{\psi}}) &\geq (1 - \e) (1-\e) n - \binent(\e) \\
										&\geq (1-2\e) n - \binent(\e).
\end{align*}
\comment{
We added $\e \geq 10/\sqrt{n}$ to make sure that 
\[
\log d_{B_{h-1}} \geq \e n - 2 \log(8 \cdot 6^2 \cdot h^2/\e^2) \geq 10\sqrt{n} - 4\log(h/\e) - 18 \geq 2 \log(h/\e) + 10
\]
The reason this is an issue here is that the error $\e$ (in the MUR) and the size $\e n$ of the B system are taken to be the same.
}
\end{proof}


\section{Locking classical information in quantum states}
\label{sec:locking}

\paragraph{Outline of the section}
We apply the results on metric uncertainty relations of the previous section to obtain locking schemes. After an introductory section on locking classical correlations (Section \ref{sec:locking-background}), we show how to obtain a locking scheme using a metric uncertainty relation in Section \ref{sec:locking-ur}. Using the constructions of the previous section, this leads to locking schemes presented in Corollaries \ref{cor:existence-locking} and   \ref{cor:explicit-locking}. In Section \ref{sec:hiding-fingerprint}, we show how to construct quantum hiding fingerprints by locking a classical fingerprint. In Section \ref{sec:string-commitment}, we observe that these locking schemes can be used to construct efficient string commitment protocols. Section \ref{sec:ent-formation} discusses the link to locking entanglement of formation.

\subsection{Background}
\label{sec:locking-background}
Locking of classical correlations was first described in \cite{DHLST04} as a violation of the incremental proportionality of the maximal classical mutual information that can be obtained by local measurement on a bipartite state. More precisely, for a bipartite state $\omega^{AB}$, the maximum classical mutual information $\entI_c$ is defined by
\[
\entI_c(A;B)_{\omega} = \max_{\{M_i^A\}, \{M_i^B\} } \entI( I_{A}; I_{B}),
\]
where $\{M_i^A\}$ and $\{M_i^B\}$ are measurements on $A$ and $B$, and $I_{A}, I_B$ are the (random) outcomes of these measurements on the state $\omega^{AB}$.
Incremental proportionality is the intuitive property that $\ell$ bits of communication between two parties can increase their mutual information by at most $\ell$ bits. \citeN{DHLST04} considered the states
\begin{equation}
\label{eq:lockingstate1}
\omega^{XKC} = \frac{1}{2d} \sum_{k=0}^1 \sum_{x=0}^{d-1} \proj{x}^X \ox \proj{k}^K \ox (U_k \proj{x} U_k^{\dagger})^C 
\end{equation}
for $k \in \{0,1\}$ where $U_0 = \1$ and $U_1$ is the Hadamard transform. 
It was shown by \citeN{DHLST04} that the classical mutual information $\entI_c(XK; C)_{\omega} = \frac{1}{2} \log d$. However, if the holder of the $C$ system also knows the value of $k$, then we can represent the global state by the following density operator
\[
\omega^{XKCK'} = \frac{1}{2d} \sum_{k=0}^1 \sum_{x=0}^{d-1}  \proj{x}^X \ox \proj{k}^K \ox (U_k \proj{x} U_k^{\dagger})^C \ox \proj{k}^{K'}.
\]
It is easy to see that $\entI_c(XK;CK')_{\omega} = 1 + \log d$. This means that with only one bit of communication (represented by the register $K'$), the classical mutual information between systems $XK$ and $C$ jumped from $\frac{1}{2} \log d$ to $1+\log d$. In other words, it is possible to unlock $\frac{1}{2} \log d$ bits of information (about $X$) from the quantum system $C$ using a single bit.

\citeN{HLSW04} proved an even stronger locking result. They generalize the state in equation \eqref{eq:lockingstate1} to
\begin{equation}
\label{eq:lockingstate2}
\omega^{XKC} = \frac{1}{t d} \sum_{x=0}^{d-1} \sum_{k=0}^{t-1} \proj{x}^X \ox \proj{k}^K \otimes (U_k \proj{x} U_k^{\dagger})^C \ox \proj{k}^{K'}
\end{equation}
where $U_k$ are chosen independently at random according to the Haar measure. They show that for any $\e > 0$, by taking $t = (\log d)^3$ and if $d$ is large enough, 
\[
\entI_c(X; C)_{\omega} \leq \e \log d \qquad \text{ and } \qquad \entI_c(XK; CK')_{\omega} = \log d+ \log t
\]
with high probability. Note that the size of the key measured in bits is only $\log t = O(\log \log d)$ and it should be compared to the $(1-\e) \log d$ bits of unlocked (classical) information. It should be noted that their argument is probabilistic, and it does not say how to construct the unitary transformations $U_k$. It is worth stressing that standard derandomization techniques are not known to work in this setting. For example, unitary $t$-designs use far too many bits of randomness \cite{DCEL09}. Moreover, using a $\delta$-biased subset of the set of Pauli matrices fails to produce a locking scheme unless the subset has a size of the order of the dimension $d$ \cite{AS04,DD10} (see Appendix \ref{sec:app-pauli}).

Here, we view locking as a cryptographic task in which a message is encoded into a quantum state using a key whose size is much smaller than the message. Having access to the key, one can decode the message. However, an eavesdropper who does not have access to the key and has complete uncertainty about the message can extract almost no classical information about the message. 
We should stress here that this is not a composable cryptographic task, namely because an eavesdropper could choose to store quantum information about the message instead of measuring. In fact, as shown in \cite{KRBM07}, using the communicated message $X$ as a key for a one-time pad encryption might not be secure; see also \cite{FDHL10}.

One could compare a locking scheme to an entropically secure encryption scheme \cite{RW02,DS05}. These two schemes achieve the same task of encrypting a high entropy message using a small key. The security definition of a locking scheme is strictly stronger. In fact, for a classical eavesdropper (i.e., an eavesdropper that can only measure) an $\e$-locking scheme is secure in a strong sense. This additional security guarantee comes at the cost of upgrading classical communication to quantum communication. With respect to quantum entropically secure encryption \cite{Des09,DD10}, the security condition of a locking scheme is also more stringent (see Appendix \ref{sec:app-pauli} for an example of an entropically secure encryption scheme that is not $\e$-locking). However, a quantum entropically secure scheme allows the encryption of quantum states.

Nonetheless, we note that an $\e$-locking scheme hides the message in a stronger sense if the adversary is limited to a small  quantum memory. In fact, using the same technique as \cite[Corollary 2]{HMRRS10} based on \cite{RRS09}, if the adversary is allowed to store $m$ qubits, then the joint state of the message and the knowledge of the adversary is $(c2^{m} \e)$-close to a product state for some universal constant $c$. For example, if $m = O(\log n)$, then a key of logarithmic size can still be used. This is especially interesting for the scheme presented in Corollary \ref{cor:explicit-locking} below, for which the sender and the receiver do not use any quantum memory. One could then use such a scheme for key distribution in the bounded quantum storage model, where the adversary is only allowed to have a quantum memory of logarithmic size in $n$ and an arbitrarily large classical memory. Note that even though this is a strong assumption compared to the unconditional security of BB84 \cite{BB84}, one advantage of such a protocol for key distribution is that it only uses one-way communication between the two parties. In contrast, the BB84 quantum key distribution protocol needs interaction between the two parties.


\begin{definition}[$\e$-locking scheme]
\label{def:locking} Let $n$ be a positive integer, $\ell \in [0, n]$ and $\e \in [0,1]$.
An encoding $\cE : [2^n] \times [t] \to \cS(C)$ is said to be $(\ell, \e)$-\emph{locking} for the quantum system $C$ if:
\begin{itemize}
\item For all $x \neq x' \in [2^n]$ and all $k \in [t]$, $\tracedist{\cE(x,k),\cE(x', k)} = 1$.
\item Let $X$ (the message) be a random variable on $[2^n]$ with min-entropy $\entHmin(X) \geq \ell$, and $K$ (the key) be an independent uniform random variable on $[t]$. For any measurement $\{M_i\}$ on $C$ and any outcome $i$,
\begin{equation}
\label{eq:randomization}
\tracedist{p_{X|\event{I=i}}, p_{X}} \leq \e.
\end{equation}
where $I$ is the outcome of measurement $\{M_i\}$ on the (random) quantum state $\cE(X, K)$. 

When the min-entropy bound $\ell$ is not specified, it should be understood that $\ell = n$ meaning that $X$ is uniformly distributed on $[2^n]$. The state $\cE(x, k)$ for $x \in [2^n]$ and $k \in [t]$ is referred to as the ciphertext.
\end{itemize}
\end{definition}
\begin{myremark}[]
The relevant parameters of a locking scheme are: the number of bits $n$ of the (classical) message, the dimension $d$ of the (quantum) ciphertext, the number $t$ of possible values of the key and the error $\e$. Strictly speaking, a classical one-time pad encryption, for which $t = 2^{n}$, is $(0,0)$-locking according to this definition. However, here we seek locking schemes for which $t$ is much smaller than $2^n$, say polynomial in $n$. This cannot be achieved using a classical encryption scheme.
\comment{To get that guess a key $K'$ and take the measurement $U_K \ket{i}$, we then have $\pr{X=i|\event{I=i}} \geq 1/t$.}
\end{myremark}

Note that we used the statistical distance between $p_{X|\event{I=i}}$ and $p_X$ instead of the mutual information between $X$ and $I$ to measure the information gained about $X$ from a measurement. Using the trace distance is a stronger requirement as demonstrated by the following proposition.
\begin{proposition}
\label{prop:mutinfo}
Let $\e \in [0,1)$ and $\cE : [2^n] \times [t] \to \cS(C)$ be an $\e$-locking scheme. Define the state 
\[
\omega^{XKCK'} = \frac{1}{t d} \sum_{k=0}^{t-1} \sum_{x=0}^{2^n-1} \proj{x}^X \ox \proj{k}^K  \ox \cE(x, k)^C \ox \proj{k}^{K'}.
\] 
Then,
\[
\entI_c(X; C)_{\omega} \leq \e n + \binent(\e) \qquad \text{ and } \qquad \entI_c(XK; CK')_{\omega} = n + \log t
\]
where $\binent(\e) = -\e \log (\e) - (1-\e) \log(1-\e)$.
\end{proposition}
\begin{proof}
First, we can suppose that the measurement performed on the system $X$ is in the basis $\{ \ket{x}\}_{x}$. In fact, the outcome distribution of any measurement on the $X$ system can be simulated classically  using the values of the random variables $X$.

Now let $I$ be the outcome of a measurement performed on the $C$ system. Using Fannes' inequality, we have for any $i$
\begin{align*}
\entH(X) - \entH(X|I=i) &\leq \tracedist{p_X, p_{X|\event{I=i}}}n - \binent\left(\tracedist{p_X, p_{X|\event{I=i}}}\right) \\
					&\leq \e n + \binent(\e)
\end{align*}
using the fact that $\cE$ defines an $\e$-locking scheme. Thus,
\begin{align*}
\entI(X;I) &= \entH(X) - \sum_i \pr{I=i} \entH(X|I=i) \\
		&\leq \e n + \binent(\e).
\end{align*}
As this holds for any measurement, we get $\entI_c(X; C)_{\omega} \leq \e n + \binent(\e)$.
\end{proof}

The trace distance was also used in \cite{Dup09,FDHL10} to define a locking scheme. To measure the leakage of information about $X$ caused by a measurement, they used the trace distance between the joint distribution of $p_{(X, I)}$ and the product distribution $p_{X} \times p_I$. Note that our definition is stronger, in that for all outcomes of the measurement $i$, $\tracedist{p_{X|\event{I=i}}, p_X} \leq \e$ whereas the definition of \cite{FDHL10} says that this only holds on average over $i$. To the best of our knowledge, even the existence of such a strong locking scheme with small key was unknown.

For a survey on locking classical correlations, see~\cite{Leu09}.

\subsection{Locking using a metric uncertainty relation}
\label{sec:locking-ur}

The following theorem shows that a locking scheme can easily be constructed using a metric uncertainty relation.
\begin{theorem}
\label{thm:ur-locking}
Let $\e \in (0,1)$ and $\{U_0, \dots, U_{t-1}\}$ be a set of unitary transformations of $A \otimes B$ that satisfies an $\e$-metric uncertainty relation on $A$, i.e., for all states $\ket{\psi} \in AB$,
\[
\frac{1}{t} \sum_{k=0}^{t-1} \tracedist{p^A_{U_k \ket{\psi}}, \unif([d_A])} \leq \e.
\]
Assume $d_A = 2^n$. Then, the mapping $\cE : [2^n] \times [t] \to \cS(A B)$ defined by
\[
\cE(x, k) = \frac{1}{d_B} \sum_{b=0}^{d_B-1} U_k^{\dagger} \left( \proj{x}^A \otimes \proj{b}^B \right) U_k. 
\]
is $\e$-locking. Moreover, for all $\ell \in [0, n]$ such that $2^{\ell - n} > \e$, it is $(\ell, \frac{2\e}{2^{\ell-n} - \e})$-locking.
\end{theorem}
\begin{myremark}[]
The state that the encoder inputs in the $B$ system is simply private randomness. The encoder chooses a uniformly random $b \in [d_B]$ and sends the quantum state $U_k^{\dagger} \ket{x}^A \ket{b}^B$. Note that $b$ does not need to be part of the key (i.e., shared with the receiver). This makes the dimension $d = d_A d_B$ of the ciphertext larger than the number of possible messages $2^n$. If one insists on having a ciphertext of the same size as the message, it suffices to consider $b$ as part of the message and apply a one-time pad encryption to $b$. The number of possible values taken by the key increases to $t \cdot d_B$.
\end{myremark}
\begin{proof}
First, it is clear that different messages are distinguishable. In fact, for $x \neq x'$ and any $k$,
\[
\tracedist{\cE(x,k), \cE(x', k)} = \frac{1}{2}\tr \left[ \sqrt{ \proj{x}^A \otimes \frac{\1^B}{d_B} - \proj{x'}^A \otimes \frac{\1^B}{d_B} } \right] = 1.
\]
We now prove the locking property. Let $X$ be the random variable representing the message. Assume that $X$ is uniformly distributed over some set $S \subseteq [d_A]$ of size $|S| \geq 2^{\ell}$. Let $K$ be a uniformly random key in $[t]$ that is independent of $X$. Consider a POVM $\{M_i\}$  on the system $AB$. Without loss of generality, we can suppose that the POVM elements $M_i$ have rank $1$. Otherwise, by writing $M_i$ in its eigenbasis, we could decompose outcome $i$ into more outcomes that can only reveal more information. So we can write the elements as weighted rank one projectors: $M_i = \xi_i\proj{e_i}$ where $\xi_i > 0$. Our objective is to show that the outcome $I$ of this measurement on the state $\cE(X,K)$ is almost independent of $X$. More precisely, for a fixed measurement outcome $I=i$, we want to compare the conditional distribution $p_{X|\event{I=i}}$ with $p_X$. The trace distance between these distributions can be written as
\begin{equation}
\label{eq:avgoutcome}
\frac{1}{2} \sum_{x=0}^{d_A-1} \big| \pr{X=x |  I=i} - \pr{X=x} \big|.
\end{equation}

Towards this objective, we start by computing the distribution of the measurement outcome $I$, given the value of the message $X=x$ (note that the receiver does not know the key):
\begin{align*}
\pr{ I = i | X = x} 
	&= \frac{\xi_i}{t d_B} \sum_{k=0}^{t-1} \sum_{b=0}^{d_B-1} \tr \big[U_k \proj{e_i} U_k^\dg \cdot \proj{x}^A \otimes \proj{b}^B \big] \\
	&= \frac{\xi_i}{t d_B} \sum_{k=0}^{t-1} \sum_{b=0}^{d_B-1} \bra{x}^A \bra{b}^B U_k \proj{e_i}  U^\dg_k \ket{x}^A \ket{b}^B \\
	&= \frac{\xi_i}{t d_B} \sum_{k=0}^{t-1} \sum_{b=0}^{d_B-1} \left|\bra{x}^A \bra{b}^B U_k \ket{e_i} \right|^2 \\
	&= \frac{\xi_i}{d_B} \frac{1}{t} \sum_{k=0}^{t-1} p^A_{U_k \ket{e_i}}(x).
\end{align*}
Since $X$ is uniformly distributed over $S$, we have that for all $x \in S$
\begin{align}
\pr{X=x|I=i} &= \frac{\pr{X=x} \pr{ I=i | X=x} }{ \sum_{x' \in S} \pr{X=x'} \pr{ I=i | X=x'} }   \notag \\ 
	&=  \frac{(1/t) \cdot\sum_k p^A_{U_k \ket{e_i}}(x) }{ (1/t) \cdot \sum_{x' \in S} \sum_k p^A_{U_k \ket{e_i}}(x') }.
\end{align}
Observe that in the case where $X$ is uniformly distributed over $[2^n]$ ($S = [2^n]$), it is simple to obtain directly that
\[
\tracedist{p_{X| \event{I=i}}, p_X} = \frac{1}{2} \sum_{x=0}^{d_A-1} \left|\frac{1}{t}\sum_{k=0}^{t-1} p^A_{U_k \ket{e_i}}(x) - \frac{1}{2^n} \right| \leq \e
\]
using the fact that $\{U_k\}$ satisfies a metric uncertainty relation on $A$.
Now let $S$ be any set of size at least $2^{\ell}$ and let $\alpha = \frac{1}{t} \sum_{x' \in S} \sum_k p^A_{U_k \ket{e_i}}(x')$. We then bound	
\begin{align*}
\frac{1}{2} \sum_{x=0}^{d_A-1} \big| \pr{X=x |  I=i} - \pr{X=x} \big| &= \frac{1}{2} \sum_{x \in S} \left| \frac{(1/t) \cdot \sum_k p^A_{U_k \ket{e_i}}(x) }{ \alpha } - \frac{1}{|S|}\right| \\
						&= \frac{1}{2 \alpha} \cdot  \sum_{x \in S} 
					\left| \frac{1}{t} \sum_{k=0}^{t-1} p^A_{U_k \ket{e_i}}(x) - \frac{\alpha}{|S|}\right| \\
						&\leq \frac{1}{2 \alpha} \cdot \frac{1}{t} \sum_k \left( \sum_{x \in S} 
					\left|  p^A_{U_k \ket{e_i}}(x) - \frac{1}{2^n} \right| + \left| \frac{1}{2^n} - \frac{\alpha}{|S|}\right| \right).
\end{align*}
We now use the fact that $\{U_k\}$ satisfies a metric uncertainty relation on $A$:  we get
\[
\frac{1}{t} \sum_k \frac{1}{2} \sum_{x \in S} \left|  p^A_{U_k \ket{e_i}}(x) - \frac{1}{2^n} \right| \leq \frac{1}{t} \sum_k \frac{1}{2} \sum_{x \in [d_A]} \left|  p^A_{U_k \ket{e_i}}(x) - \frac{1}{2^n} \right| \leq \e
\]
and
\begin{equation}
\label{eq:alpha}
\frac{1}{2} \left| \frac{|S|}{2^n} - \alpha\right| = \frac{1}{2} \left| \frac{|S|}{2^n} - \frac{1}{t} \sum_{x' \in S} \sum_{k=0}^{t-1} p^A_{U_k \ket{e_i}}(x') \right| \leq \e.
\end{equation}
As a result, we have
\[
\tracedist{p_{X| \event{I=i}}, p_X} \leq \frac{2\e}{\alpha}.
\]
Using \eqref{eq:alpha}, we have $\alpha \geq |S|2^{-n} - \e \geq 2^{\ell - n} - \e$. If $\e < 2^{\ell - n}$, we get
\[
\tracedist{p_{X| \event{I=i}}, p_X} \leq \frac{2\e}{2^{\ell - n} - \e}.
\]

In the general case when $X$ has min-entropy $\ell$, the distribution of $X$ can be seen as a mixture of uniform distributions over sets of size at least $2^{\ell}$. So there exist independent random variables $J \in \NN$ and $\{X_j\}$ uniformly distributed on sets of size at least $2^{\ell}$ such that $X = X_J$. One can then write
\begin{align*}
&\frac{1}{2} \sum_x \left| \pr{X = x| I = i} - \pr{X=x} \right| \\
&= \frac{1}{2} \sum_{x,j} \left| \pr{J = j} \left( \pr{X_j=x | I = i, J=j} - \pr{X_j = x | J=j} \right) \right| \\
								&\leq \frac{2\e}{2^{\ell-n} - \e}.
\end{align*}
\end{proof}

Using Theorem \ref{thm:ur-locking} together with the existence of metric uncertainty relations (Theorem \ref{thm:existence-ur}), we show the existence of $\e$-locking schemes whose key size depends only on $\e$ and not on the size of the encoded message. This result was not previously known.
\begin{corollary}[Existence of locking schemes]
\label{cor:existence-locking}
Let $n$ be a large enough integer and $\e \in (0,1)$. Then there exists an $\e$-locking scheme encoding an $n$-bit message using a key of at most $2\log(1/\e) + O(\log \log(1/\e))$ bits into at most $n + 2 \log(18/\e)$ qubits.
\end{corollary}
\begin{myremark}[]
Observe that in terms of number of bits, the size of the key is only a factor of two larger (up to smaller order terms) than the lower bound of $\log(1/(\e + 2^{-n}))$ bits that can be obtained by guessing the key. In fact, consider the strategy of performing the decoding operation corresponding to the key value $0$. In this case, we have $\pr{X=i|I=i} \geq \pr{K=0} = 1/t$. Thus, $\tracedist{p_{X|I=i}, p_X} \geq 1/t - 2^{-n}$.

Recall that we can increase the size of the message to be equal to the number of qubits of the ciphertext. The key size becomes at most $4\log(1/\e) + O(\log(\log(1/\e))$.
\end{myremark}
\comment{The $+10$ is an upper bound on $2 \log(18)$}
\begin{proof}
Use the construction of Theorem \ref{thm:existence-ur} with $d_A = 2^n$ and $d_B = 2^q$ such that $2^{q-1} < 9/\e^2 \leq 2^q$ and $d = d_A d_B$. Take $t = 2^p$ to be the power of two with $2^{p-1} \leq \frac{4 \cdot 18c\log (9/\e)}{\e^2} < 2^p$.
\end{proof}

The following corollary gives explicit locking schemes. We mention the constructions based on Theorems \ref{thm:explicit-ur1} and \ref{thm:explicit-ur2}. Of course, one could obtain a tradeoff between the key size and the dimension of the quantum system.
\begin{corollary}[Explicit locking schemes]
\label{cor:explicit-locking}
Let $\delta > 0$ be a constant, $n$ be a positive integer, $\e \in (2^{-c'n},1)$ ($c'$ is a constant independent of $n$). 
\begin{itemize}
\item Then, there exists an efficient $\e$-locking scheme encoding an $n$-bit message in a quantum state of $n' \leq (4+\delta) n + O(\log(1/\e))$ qubits using a key of size $O(\log(n/\e))$ bits. In fact, both the encoding and decoding operations are computable using a classical computation with polynomial running time and a quantum circuit with only Hadamard gates and preparations and measurements in the computational basis. 

\item There also exists an efficient $\e$-locking scheme $\cE'$ encoding an $n$-bit message in a quantum state of $n$ qubits using a key of size $O(\log(n/\e) \cdot \log n)$ bits. $\cE'$ is computable by a classical algorithm with polynomial runtime and a quantum circuit of size $O(n \polylog(n/\e))$.
\end{itemize}
\end{corollary}
\begin{proof}
For the first result, we observe that the construction of Theorem \ref{thm:ur-locking} encodes the message in the computational basis. Recall that the untaries $U_k$ of Theorem \ref{thm:explicit-ur1} are of the form $U_k = P_k V_k$ where $P_k$ is a permutation of the computational basis. Hence, it is possible to \emph{classically} compute the label of the computational basis element $P^{\dagger}_k \ket{x} \ket{b}$. One can then prepare the state $P^{\dagger}_k \ket{x} \ket{b}$ and apply the unitary $V^{\dagger}_k$ to obtain the ciphertext. The decoding is performed in a similar way. One first applies the unitary $V_k$, measures in the computational basis and then applies the permutation $P_k$ to the $n$-bit string corresponding to the outcome.


For the second construction, we apply Theorem \ref{thm:explicit-ur2} with $n' =  n + c'\ceil{\log\log n + \log(1/\e)}$ for some large enough constant $c'$.  We can then use a one-time pad encryption on the input to the $B$ system. This increases the size of the key by only $c'\ceil{\log\log n + \log(1/\e)}$ bits.
\end{proof}

As mentioned earlier (see equation \eqref{eq:lockingstate1}), explicit states that exhibit locking behaviour have been presented in \cite{DHLST04}. However, this is the first explicit construction of states $\omega$ that achieves the following strong locking behaviour: for any $\delta > 0$, for $n$ large enough, the state $\omega^{XCK}$ satisfies $\entI_c(X; C)_{\omega} \leq \delta$ and $\entI_c(X; CK)_{\omega} = n + \log d_K$ where $K$ is a classical $O(\log(n/\delta))$-bit system. This is a direct consequence of Corollary \ref{cor:explicit-locking} taking $\e = \delta/(20n)$, and Proposition \ref{prop:mutinfo}. 
We should also mention that \citeN{KRBM07} explicitly construct a state exhibiting locking behaviour where the key is large but the stronger condition $\entI_c(XK;C)_{\omega} \leq \delta$ is satisfied.

\subsection{Quantum hiding fingerprints}
\label{sec:hiding-fingerprint}

In this section, we show that the locking scheme of Corollary \ref{cor:existence-locking} can be used to build mixed state quantum hiding fingerprints as defined by \citeN{GI10}. A quantum fingerprint \cite{BCWW01} encodes an $n$-bit string into a quantum state $\rho_x$ of $n' \ll n$ qubits such that given $y \in \{0,1\}^n$ and the fingerprint $\rho_x$, it is possible to decide with small error probability whether $x=y$. The additional hiding property ensures that measuring $\rho_x$ leaks very little information about $x$. Here, we prove that such a hiding fingerprint can be obtained by locking a classical fingerprint. The hiding property is then a direct consequence of the locking property. In order to prove the fingerprinting property, we use the mutually unbiased property of the bases involved in  the scheme of Corollary \ref{cor:existence-locking}.
 
\citeN{GI10} used the accessible information\footnote{The accessible information about $X$ in a quantum system $C$ refers to the maximum over all measurements of the system $C$ of $\entI(X; I)$ where $I$ is the outcome of that measurement.} as a measure of the hiding property. Here, we strengthen this definition by imposing a bound on the total variation distance instead (see Proposition \ref{prop:metric-to-entropic}). 

\begin{definition}[Quantum hiding fingerprint]
Let $n$ be a positive integer, $\delta, \e \in (0,1)$ and $C$ be a Hilbert space. An encoding $f : \{0,1\}^n \to \cS(C)$ together with a set of measurements $\{M^y, \1 - M^y\}$ for each $y \in \{0,1\}^n$ is a $(\delta, \e)$-hiding fingerprint if 
\begin{enumerate}
\item (Fingerprint property) For all $x \in \{0,1\}^n$, $\tr \left[ M^x f(x) \right] = 1$ and for $y \neq x$, $\tr \left[ M^y f(x) \right] \leq \delta$.
\item (Hiding property) Let $X$ be uniformly distributed on $\{0,1\}^n$. Then, for any POVM $\{N_i\}$ on the system $C$ whose outcome on $f(X)$ is denoted $I$, we have for all possible outcomes $i$,
\[
\tracedist{p_{X| \event{I=i}}, p_X} \leq \e.
\]
\end{enumerate}
\end{definition}
We usually want the Hilbert space $C$ to be composed of $O(\log n)$ qubits. \citeN{GI10} proved that for any constant $c$, there exists efficient quantum hiding fingerprinting schemes for which the number of qubits in the system $C$ is $O(\log n)$ and both the error probability $\delta$ and the accessible information are bounded by $1/n^c$. Here, we prove that the same result can be obtained by locking a classical fingerprint. The general structure of our quantum hiding fingerprint for parameters $n, \delta$ and $\e$ is as follows:
\begin{enumerate}
\item Choose a random prime $p \in \cP_{n,\e, \delta}$ uniformly from the set $\cP_{n, \e, \delta}$.
\item Set $t = \ceil{c \log(1/\e) \e^{-2}}$, $d_A = p$ and $d_B = \ceil{c'/\e^2}$ and generate $t$ random unitaries $U^p_0, \dots, U^p_{t-1}$ acting on $A \otimes B$.
\item The fingerprint consists of the random prime $p$ and the state $(U^p_k)^{\dagger} \ket{x\bmod{p}}^A \ket{b}^B$ where $k \in [t]$ and $b \in [d_B]$ are chosen uniformly and independently. The density operator representing this state is denoted $f(x) \eqdef \frac{1}{t d_B} \sum_{k, b} (U^p_k)^{\dagger} \proj{x\bmod{p}}^A \proj{b}^B U^p_k$.
\end{enumerate}
Observe that even though this protocol is randomized because the unitaries are chosen at random, it is possible to implement it with polynomial resources in $n$ as the size of the message to be locked is $O(\log n)$ bits. In fact, one can approximately sample a random unitary in dimension $2^{O(\log n)}$ using a polynomial number of public random bits. The mixed state protocol of \cite{GI10} achieves roughly the same parameters. Their construction is also randomized but it uses random codes instead of random unitaries. For this reason, the protocol of \cite{GI10} would probably be more efficient in practice.

\begin{theorem}
\label{thm:hiding-fingerprint}
There exist constants $c,c'$ and $c''$, such that for all positive integer $n$, $\delta, \e \in (0,1/4)$ if we define $\cP_{n,\delta, \e}$ to be the set of primes in the interval $[l, u]$ where 
\[
l = \left( \frac{c''}{\delta} \cdot \frac{\log^2(1/\e)}{\e^{8}} \right)^{1/0.9} + 10n \quad \text{ and } \quad u = l + (2n/\delta)^2
\]
and provided $u \leq 2^{n-2}$, the scheme described above is a $(\delta, \e)$-hiding fingerprint with probability $1-2^{-\Omega(n)}$ over the choice of random unitaries.
%
\end{theorem}
The proof of this result involves two parts. First, we need to show that the fingerprint of a uniformly distributed $X \in \{0,1\}^n$ does not give away much information about $X$. This follows easily from Theorem \ref{thm:existence-ur} and Theorem \ref{thm:ur-locking}. We also need to show that for every $y \in \{0,1\}^n$, there is a measurement that Bob can apply to the fingerprint to determine with high confidence whether it corresponds to a fingerprint of $y$ or not.  In order to prove this, we use the following proposition on the Gram-Schmidt orthonormalisation of a set of almost orthogonal vectors.

%
%

\begin{proposition}
\label{prop:gram-schmidt}
Let $v'_1, \ldots, v'_r$ be a sequence of unit length vectors 
in a Hilbert space. Let $0 < \delta \leq \frac{1}{16 r}$. 
For any $i \neq j$, suppose $|\braket{v'_i}{v'_j}| \leq \delta$. 
Let $v_1, \ldots, v_r$ be the corresponding sequence of vectors obtained by Gram-Schmidt
orthonormalising $v'_1, \ldots, v'_r$. Then for any $i$, 
$\norm{v_i - v'_i} \leq \delta \sqrt{32(i-1)}$.
\end{proposition}
\begin{proof}
Since $|\braket{v'_i}{v'_j}| < \delta < 1/r$ for any $i \neq j$, the
vectors $v'_1, \ldots, v'_r$ are linearly independent.
Define $\Pi_0$ to be the zero linear operator. For $i \geq 1$,
define $\Pi_i$ to be the orthogonal projection onto the subspace spanned by the
vectors $v'_1, \ldots, v'_i$. Observe that for any $i$, $v'_1, \ldots, v'_i$ and
$v_1, \ldots, v_i$ span the same space, and
$v_{i+1} = \frac{v'_{i+1} - \Pi_i(v'_{i+1})}{\norm{v'_{i+1} - \Pi_i(v'_{i+1})}}$.
We shall prove by induction on $i$ that
$\norm{\Pi_i(v'_k)} \leq 4 \delta \sqrt{i}$ for all $i$ and all $k > i$. 
This will prove the desired statement since
\begin{eqnarray*}
\norm{v_i - v'_i}^2 
&   =  & \norm{\Pi_{i-1}(v'_i)}^2 + 
         \left(1 - \norm{v'_i - \Pi_{i-1}(v'_i)}\right)^2 \\
&   =  & \norm{\Pi_{i-1}(v'_i)}^2 +
         \left(1 - \sqrt{1 - \norm{\Pi_{i-1}(v'_i)}^2}\right)^2 \\
&   =  & 2 - 2 \sqrt{1 - \norm{\Pi_{i-1}(v'_i)}^2} 
\;\leq\; 2 - 2 \sqrt{1 - 16 \delta^2 (i-1)} \\
& \leq & 32 \delta^2 (i-1).
\end{eqnarray*}
For the first equality, we write $v_i' = \Pi_{i-1}(v'_i) + \norm{v'_i - \Pi_{i-1}(v'_i)} \: v_i$ and we use the fact that $\Pi_{i-1}(v'_i)$ and $v_i$ are orthogonal.
 
The base case of $i = 1$ is trivial. Assume that the induction hypothesis holds
for a particular $i$. Let $1 \leq j \leq i+1$ and $k > i+1$. 
Observe that 
$v'_j = \Pi_{j-1}(v'_j) + \sqrt{1 - \norm{\Pi_{j-1}(v'_j)}^2} \; v_j$.
We have
\begin{eqnarray*}
|\braket{v'_k}{v'_j}| 
&   =  & \left|\braket{v'_k}{\Pi_{j-1}(v'_j)} + 
          \sqrt{1 - \norm{\Pi_{j-1}(v'_j)}^2} \; \braket{v'_k}{v_j} \right| \\
&   =  & \left|\braket{\Pi_{j-1}(v'_k)}{\Pi_{j-1}(v'_j)} + 
          \sqrt{1 - \norm{\Pi_{j-1}(v'_j)}^2} \; \braket{v'_k}{v_j} \right| \\
& \geq & \sqrt{1 - \norm{\Pi_{j-1}(v'_j)}^2} \; |\braket{v'_k}{v_j}| 
         - \norm{\Pi_{j-1}(v'_k)} \, \norm{\Pi_{j-1}(v'_j)},
\end{eqnarray*}
which implies that
\begin{eqnarray*}
|\braket{v'_k}{v_j}| 
& \leq & \frac{|\braket{v'_k}{v'_j}|+\norm{\Pi_{j-1}(v'_k)}\,\norm{\Pi_{j-1}(v'_j)}}
              {\sqrt{1 - \norm{\Pi_{j-1}(v'_j)}^2}} \\
&   \leq  & \frac{\delta + 16 \delta^2 (j-1)}{\sqrt{1 - 16 \delta^2 (j-1)}} 
\;  \leq \; \frac{\delta + \delta}{\sqrt{1 - \delta}} \\
&   \leq  & 4 \delta.
\end{eqnarray*}
Thus, 
$
\norm{\Pi_{i+1}(v'_k)}^2 = \sum_{j=1}^{i+1} |\braket{v'_k}{v_j}|^2 \leq
16 \delta^2 (i+1)
$,
which gives $\norm{\Pi_{i+1}(v'_k)} \leq 4 \delta \sqrt{i+1}$ completing the
induction.
\end{proof}
Using this result we can prove the following lemma.
\begin{lemma}
\label{lem:fingerprint-error}
Let $\{U_0, \dots, U_{t-1}\}$ be a set of unitary transformations on $AB$ that 
define $\gamma$-MUBs and $d^{-\gamma/2} \leq 1/(16td_B)$ where $d \eqdef d_A d_B$. Define for $y \in [d_A]$ the subspace $F_y = \textrm{span} \{ U_k^{\dagger} \ket{y} \ket{b}, k \in [t], b \in [d_B] \}$. Then for any $x \in [d_A]$, $y \neq x$, $k_0 \in [t]$ and $b_0 \in [d_B]$,
\[
\tr \left[ \Pi_{F_y} U^{\dagger}_{k_0} \ket{x} \ket{b_0} \right] \leq 2 \sqrt{32} (t d_B)^2 d^{-\gamma}.
\]
where $\Pi_{F}$ is the projector on the subspace $F$.
\end{lemma}
\begin{proof}
Consider the set of vectors $\{U_k^{\dagger} \ket{y} \ket{b}\}_{k \in [t], b \in [d_B]}$. We have for all $(k,b) \neq (k',b')$, 
\[
|\bra{y} \bra{b'} U_{k'} U^{\dagger}_{k} \ket{y} \ket{b}| \leq d^{-\gamma/2}.
\]
Picking any fixed ordering on $[t] \times [d_B]$, define $\{\ket{e_{k,b}(y)}\}_{k,b}$ to be the set of vectors obtained by Gram-Schmidt orthonormalising $\{U_k^{\dagger} \ket{y} \ket{b}\}_{k \in [t], b \in [d_B]}$. Using Proposition \ref{prop:gram-schmidt}, we have $\norm{\ket{e_{k,b}(y)} - U_k^{\dagger} \ket{y} \ket{b}} \leq d^{-\gamma/2} \sqrt{ 32 t d_B}$. Thus,
\begin{align*}
\tr \left[ \Pi_{F_y} U^{\dagger}_{k_0} \ket{x} \ket{b_0} \right] &= \sum_{k,b} |\bra{e_{k,b}(y)} U^{\dagger}_{k_0} \ket{x} \ket{b_0}|^2 \\
											&\leq \sum_{k,b} \left| | \bra{y} \bra{b'} U_{k'} U^{\dagger}_{k_0} \ket{x} \ket{b_0}| + \norm{\ket{e_{k,b}(y)} - U_k^{\dagger} \ket{y} \ket{b}} \right|^2 \\
											&\leq t d_B \cdot d^{-\gamma} \left(\sqrt{ 32 t d_B} + 1\right)^2 \\
											&\leq 2 \sqrt{32} (t d_B)^2 d^{-\gamma}.
\end{align*}
\end{proof}

\begin{proof}[of Theorem \ref{thm:hiding-fingerprint}]
We start by proving the hiding property. For any fixed $p$, the random variable $Z \eqdef X \bmod{p}$ is almost uniformly distributed on $[p]$. In fact, we have for any $z \in [p]$, $\pr{Z = z} \leq \frac{2^n/p + 1}{2^n}$. In other words, $\entHmin(Z) \geq \log p - \log(1+p2^{-n})$. Thus, using Theorem \ref{thm:existence-ur} and Theorem \ref{thm:ur-locking}, we have that except with probability exponentially small in $n$ (on the choice of the random unitary), the fingerprinting scheme satisfies for any measurement outcome $i$
\[
\tracedist{p_{Z|\event{I=i}}, p_Z} \leq \frac{2\e}{\frac{1}{1+p2^{-n}} - \e} \leq 4\e
\]
where $I$ denotes the outcome of a measurement on the state $f(X)$.
Recall that we are interested in the information leakage about $X$ not $Z$. For this, we note that the random variables $X, Z, I$ form a Markov chain. Thus,
\begin{align*}
&\tracedist{p_{X|\event{I=i}}, p_X} \\
 &= \sum_{x \in \{0,1\}^n} \left| \sum_{z \in [p]} \pr{Z=z|I=i} \pr{X=x|I=i, Z=z} - \pr{Z=z} \pr{X=x|Z=z}  \right| \\
						&= \sum_{x \in \{0,1\}^n} \left| \sum_{z \in [p]} \pr{Z=z|I=i} \pr{X=x|Z=z} - \pr{Z=z} \pr{X=x|Z=z}  \right| \\
						&\leq \sum_{z \in [p]} \left| \pr{Z=z|I=i} - \pr{Z=z} \right| \sum_{x \in \{0,1\}^n} \pr{X=x|Z=z}   \\
						&= \tracedist{p_{Z|\event{I=i}}, p_Z} \leq 4\e.
\end{align*}
This proves the hiding property.

We then analyse the fingerprint property. Let $x, y \in [2^n]$ and $p$ be the random prime of the fingerprint. We define the measurements by $M^y = \Pi_{F_y}$ for all $y \in \{0,1\}^n$ where $\Pi_{F_y}$ is the projector onto the subspace $F_y = \textrm{span} \{ {U^p_k}^{\dagger} \ket{y \bmod{p}} \ket{b}, k \in [t], b \in [d_B] \}$. If $x = y$, then $f(x)$ is a mixture of states in $\textrm{span} \{ {U^p_k}^{\dagger} \ket{y \bmod{p}} \ket{b}, k \in [t], b \in [d_B] \}$. Thus $\tr [ M^y f(x) ] = 1$.

We now suppose that $x \neq y$. First, we have for a random choice of prime $p \in \cP_{n, \e, \delta}$, $\pr{ x \bmod{p} = y \bmod{p} } = \pr{x-y \bmod{p} = 0} \leq \delta/2$ as the number of distinct prime divisors of $x-y$ is at most $n$ and the number of primes in $[l, u]$ is at least $2n/\delta$ for $n$ large enough. Then, whenever $x \bmod{p} \neq y \bmod{p}$, Lemma \ref{lem:fingerprint-error} with $\gamma = 0.9$ gives \begin{align*}
\tr \left[ \Pi_{F_y} f(x) \right] &\leq 2 \sqrt{32} (t d_B)^2 (d_A d_B)^{-0.9} \\
					&\leq 2 \sqrt{32} \cdot 4 c^2 c'^2 \frac{\log^2(1/\e)}{\e^8} \cdot  \frac{\delta \e^8}{c'' \log^2(1/\e)} \\
					&\leq \delta/2
\end{align*}
for $c''$ large enough with probability $1-2^{-\Omega(d_A d_B)} = 1-2^{-\Omega(n)}$ over the choice of the random unitaries (using Theorem \ref{thm:existence-ur}).
Finally, we get $\tr \left[ \Pi_{F_y} f(x) \right] \leq \delta$ with probability $1-2^{-\Omega(n)}$.
\end{proof}

\subsection{String commitment}
\label{sec:string-commitment}

In this section, we show how to use a locking scheme to obtain a weak form of bit commitment \cite{BCHLW06}. Bit commitment is an important two-party cryptographic primitive defined as follows. Consider two mutually distrustful parties Alice and Bob who are only allowed to communicate over some channel. The objective is to be able to achieve the following: Alice secretly chooses a bit $x$ and communicates with Bob to convince him that she fixed her choice, without revealing the actual bit $x$. This is the commit stage. At the reveal stage, Alice reveals the secret $x$ and enables Bob to open the commitment. Bob can then check whether Alice was honest. 

Using classical or quantum communication, unconditionally secure bit commitment is known to be unrealizable \cite{May97,LC97}. However, commitment protocols with weaker security guarantees do exist \cite{SR01,DFSS05,BCHLW06,BCHLW08}. Here, we consider the string commitment scenario studied in \cite[Section III]{BCHLW08}.
In a string commitment protocol, Alice commits to an $n$-bit string. Alice's ability to cheat is quantified by the number of strings she can reveal successfully. The ability of Bob to cheat is quantified by the information he can obtain about the string to be committed. One can formalize these notions in many ways. 
We use a security criterion that is similar to the one of \cite{BCHLW08} except that we use the statistical distance between the outcome distribution and the uniform distribution, instead of the accessible information. Our definition is slightly stronger by virtue of Proposition \ref{prop:mutinfo}. For a detailed study of string commitment in a more general setting, see \cite{BCHLW08}.

\begin{definition} 
An \emph{$(n,\alpha,\beta)$-quantum bit string commitment} is a quantum communication protocol between Alice (the committer) and Bob (the receiver) which has two phases. When both players are honest the protocol takes the following form.
\begin{itemize}
\item (Commit phase) Alice chooses a string $X \in \{0,1\}^n$ uniformly. Alice and Bob communicate, after which Bob holds a state $\rho_X$.
\item (Reveal phase) Alice and Bob communicate and Bob learns $X$.
\end{itemize}
The parameters $\alpha$ and $\beta$ are security parameters.
\begin{itemize}
\item If Alice is honest, then for any measurement performed by Bob on her state $\rho_X$, we have $\tracedist{p_X, p_{X|\event{I=i}}} \leq \frac{\beta}{n}$ where $I$ is the outcome of the measurement.
\item If Bob is honest, then for all commitments of Alice: $\sum_{x \in \{0,1\}^n} p_x \leq 2^{\alpha}$, where $p_x$ is the probability that Alice successfully reveals $x$.
\end{itemize}
\end{definition}

Following the strategy of \cite{BCHLW08}, the following protocol for string commitment can be defined using a locking scheme $\cE$.
\begin{itemize}
\item Commit phase: Alice has the string $X \in \{0,1\}^n$ and chooses a key $K \in [t]$ uniformly at random. She sends the state $\cE(X, K)$ to Bob.
\item Reveal phase: Alice announces both the string $X$ and the key $K$. Using the key, Bob measures some value $X'$. He accepts if $X = X'$.
\end{itemize}
A protocol is said to be efficient if both the communication (in terms of the number of qubits exchanged) is polynomial in $n$ and the computations performed by Alice and Bob can be done in polynomial time on a quantum computer. The protocol presented in \cite{BCHLW08} is not efficient in terms of computation and is efficient in terms of communication only if the cost of communicating a (random) unitary in dimension $2^n$ is disregarded. Using the efficient locking scheme of Corollary \ref{cor:explicit-locking}, we get
\begin{corollary}
Let $n$ be a positive integer and $\beta \in (n2^{-cn}, n)$ ($c$ is a constant independent of $n$). There exists an efficient $(n, c \log(n^2/\beta) , \beta)$-quantum bit string commitment protocol for some constant $c$ independent of $n$ and $\beta$.
\end{corollary}
\begin{proof}
We use the first construction of Corollary \ref{cor:explicit-locking} with $\e = \beta/n$.
If Bob is honest, the security analysis is exactly the same as \cite{BCHLW08}. If Alice is honest, the security follows directly from the definition of the locking scheme.
\end{proof}

\subsection{Locking entanglement of formation}
\label{sec:ent-formation}
The entanglement of formation is a measure of the entanglement in a bipartite quantum state that attempts to quantify the number of singlets required to produce the state in question using only local operations and classical communication
\cite{BDSW96}. For a bipartite state $\rho^{XY}$, the entanglement of formation is defined as
\begin{equation}
\label{eq:ent-formation}
\entF(X;Y)_{\rho} = \min_{ \{p_i, \ket{\psi_i} \}} \sum_i p_i \entH(X)_{\psi_i}.
\end{equation}
where the minimization is taken over all possible ways to write $\rho^{XY} = \sum_i p_i \proj{\psi_i}$ with $\sum_i p_i = 1$.
Entanglement of formation is related to the following quantity:
\[
\entIarrow(X;Y')_{\rho} = \max_{ \{M_i\} } \entI(X; I)
\]
where the maximization is taken over all measurements $\{M_i\}$ performed on the system $Y'$ and $I$ is the outcome of this measurement. \citeN{KW04} showed that for a pure state $\ket{\rho}^{XYY'}$, a simple identity holds:
\begin{equation}
\label{eq:info-entf}
\entF(X;Y)_{\rho} + \entIarrow(X;Y')_{\rho} = \entH(X)_{\rho}.
\end{equation}
Let $\{U_0, \dots, U_{t-1}\}$ be a set of unitary transformations of $A \otimes B \isom C$ and define
\[
\ket{\rho}^{ABCA'K} = \frac{1}{d_A d_B} \sum_{k \in [t],a \in [d_A], b \in [d_B]} \ket{a}^A \ket{b}^B \left(U^{\dagger}_k \ket{a} \otimes \ket{b} \right)^C \ket{a}^{A'} \ket{k}^K.
\]
If $\{U_0, \dots, U_{t-1}\}$ satisfies an $\e$-metric uncertainty relation, then we get a locking effect using Theorem \ref{thm:ur-locking} and Proposition \ref{prop:mutinfo}. In fact, we have $\entIarrow(A; C)_{\rho} \leq \e \log d_A + \binent(\e) $ and $\entIarrow(A; CK) = \log d_A$. Thus, using  \eqref{eq:info-entf}, we get
\[
\entF(A;A'BK)_{\rho} = \entH(A)_{\rho} - \entIarrow(A; C)_{\rho} \geq (1-\e)\log d_A - \binent(\e)
\]
and discarding the system $K$ of dimension $t$ we obtain a separable state 
\[
\entF(A;A'B)_{\rho} = 0.
\]
Explicit states for which the entanglement of formation increases by $n/2$ by adding a one-bit system $K$ have been presented in \cite{HHHO05} following \cite{DHLST04}. 
Here, using Corollary \ref{cor:explicit-entropic-ur}, we obtain explicit examples of states where the increase is $(1-\e)n$ by adding a system $K$ of $O(\log n \log(1/\e))$ bits for arbitrarily small $\e$.


\section{Quantum identification codes}
\label{sec:identification}
Consider the following quantum analogue of the equality testing communication problem. Alice is given an $n$-qubit state $\ket{\psi}$ and Bob is given $\ket{\ph}$. Namely, Bob wants to output $1$ with probability in the interval $[|\braket{\psi}{\ph}|^2 - \e, |\braket{\psi}{\ph}|^2 + \e]$ and $0$ with probability in the interval $[1 - |\braket{\psi}{\ph}|^2 - \e, 1 - |\braket{\psi}{\ph}|^2 + \e]$. This task is referred to as quantum identification \cite{Win04}. Note that communication only goes from Alice to Bob. There are many possible variations to this problem. One of the interesting models is when Alice receives the quantum state $\ket{\psi}$ and Bob gets a classical description of $\ket{\ph}$. An $\e$-quantum-ID code is defined by an encoder, which is a quantum operation that maps Alice's quantum state $\ket{\psi}$ to another quantum state which is transmitted to Bob, and a family of decoding POVMs $\{D_{\ph}, \1 - D_{\ph}\}$ for all $\ket{\ph}$ that Bob performs on the state he receives from Alice.
\begin{definition}[Quantum identification \cite{Win04}]
\label{def:qid-code} Let $\cH_1, \cH_2, C$ be Hilbert spaces and $\e \in (0,1)$.
An \emph{$\e$-quantum-ID code} for the space $C$ using the channel $\cN : \cS(\cH_1) \to \cS(\cH_2)$ consists of an encoding map $\cE : \cS(C) \to \cS(\cH_1)$ and a set of POVMs $\{D_{\ph}, \1 - D_{\ph}\}$ acting on $\cS(\cH_2)$, one for each pure state $\ket{\ph}$ such that
\[
\forall \ket{\psi}, \ket{\ph} \in C, \qquad \Big| \tr \left[ D_{\ph} \cN(\cE(\psi)) \right]  - | \braket{\ph}{\psi} |^2 \Big| \leq \e.
\]
\end{definition}
Here we consider channels $\cN$ transmitting noiseless qubits and noiseless classical bits. We also say that $\e$-quantum identification of $n$-qubit states can be performed using $\ell$ bits and $m$ qubits when there exists an $\e$-quantum-ID code for the space $C = (\CC^2)^{\otimes n}$ using the channel $\cN = \bitchannel^{\otimes \ell} \otimes  \qubitchannel^{\otimes m}$, where $\bitchannel$ and $\qubitchannel$ are the noiseless bit and qubit channels.  \citeN{HW10} showed that classical communication alone cannot be used for quantum identification. However, a small amount of quantum communication makes classical communication useful. Specifically, they proved that quantum identification of $n$-qubit states can be done using $m = o(n)$ qubits and $\ell = n$ bits of communication. Using our metric uncertainty relations, we prove better bounds on the number of qubits of communication and give an efficient encoder for this problem.

Our protocol is based on a duality between quantum identification and approximate forgetfulness of a quantum channel demonstrated in \cite[Theorem 7]{HW10}. Specialized to our setting, the direction of the duality we use
states that if  $V : C \to A\otimes B$ defines a low-distortion embedding of $(C, \ell_2)$ into $(AB, \ell^A_1(\ell^B_2))$, then the maps $\Gamma_a: C \to B$ for $a \in [d_A]$ defined by $\ket{\psi} \mapsto \sum_{b \in d_B} (\bra{a} \bra{b} V \ket{\psi}) \ket{b}$ approximately preserve inner products on average. The following lemma gives a precise statement. We give an elementary proof in the interest of making the presentation self-contained.
\begin{lemma}
\label{lem:dvo-jl}
Let $V : C \to A \otimes B$ be an isometry, i.e., for all $\ket{\psi} \in C$, $\| V \ket{\psi} \|_2 = \| \ket{\psi} \|_2$. For any vector $\ket{\psi} \in C$, we define the vectors $\ket{\psi_a} \in B$ by $V \ket{\psi} = \sum_{a \in [d_A]} \ket{a} \ket{\psi_a}$. Assume that $V$ satisfies the following property:
\begin{equation}
\label{eq:dvo-like}
\forall \ket{\psi} \in C \qquad \sum_{a \in [d_A]} \left| \| \ket{\psi_a} \|^2_2 - \frac{\| \ket{\psi} \|^2_2}{d_A} \right| \leq \e \| \ket{\psi} \|_2^2.
\end{equation}
Then we have for all unit vectors $\ket{\psi}, \ket{\ph} \in C$ with $V \ket{\psi} = \sum_{a \in [d_A]} \ket{a} \ket{\psi_a}$ and $V \ket{\ph} = \sum_{a \in [d_A]} \ket{a} \ket{\ph_a}$
\begin{equation}
\label{eq:jl-like}
\frac{1}{d_A} \sum_{a \in [d_A]} \left| \frac{|\braket{\psi_a}{\ph_a} |^2}{\| \ket{\psi_a} \|_2 \|\ket{\ph_a}\|_2} - |\braket{\psi}{\ph}|^2 \right| \leq 12 \e + 2 \sqrt{\e}.
\end{equation}
\end{lemma}
\begin{proof}
Let $\ket{\psi}$ and $\ket{\ph}$ be unit vectors in $C$. 
We use the triangle inequality to get 
\begin{align}
&\frac{1}{d_A} \sum_{a \in [d_A]} \left| \frac{|\braket{\psi_a}{\ph_a}|^2}{\| \ket{\psi_a} \|_2 \| \ket{\ph_a} \|_2} - |\braket{\psi}{\ph}|^2 \right| \notag \\
&\qquad \leq  \sum_{a \in [d_A]} \left| \frac{|\braket{\psi}{\ph}|^2}{d_A} - |\braket{\psi_a}{\ph_a}|^2 \right|
+  \sum_{a \in [d_A]} \left| |\braket{\psi_a}{\ph_a}|^2 - \frac{|\braket{\psi_a}{\ph_a}|^2}{d_A \| \ket{\psi_a} \|_2  \| \ket{\ph_a} \|_2} \right|. \label{eq:jl-dvo-1}
\end{align}

We start by dealing with the first term in \eqref{eq:jl-dvo-1}. Observe that
\begin{align}
\left| |\braket{\psi_a}{\ph_a}|^2 - \frac{|\braket{\psi}{\ph}|^2}{d_A} \right|
& \leq \left| (\re \braket{\psi_a}{\ph_a})^2 - \frac{(\re \braket{\psi}{\ph})^2 }{d_A} \right| +
\left| (\im \braket{\psi_a}{\ph_a})^2 - \frac{(\im \braket{\psi}{\ph})^2 }{d_A} \right| \notag \\
& \leq 2 \left| \re \braket{\psi_a}{\ph_a} - \frac{\re \braket{\psi}{\ph} }{d_A} \right| +
2 \left| \im \braket{\psi_a}{\ph_a} - \frac{\im \braket{\psi}{\ph} }{d_A} \right|.
 \label{eq:jl-dvo-2}
\end{align}
In the last inequality, we used the fact that $|x^2 - y^2| \leq 2 |x-y|$ whenever $|x+y| \leq 2$. To bound these terms, we apply the assumption about $V$ (equation \eqref{eq:dvo-like}) to the vector $\ket{\psi} - \ket{\ph}$:
\[
\sum_{a \in [d_A]} \left| \| \ket{\psi_a} - \ket{\ph_a} \|^2_2 - \frac{\| \ket{\psi} - \ket{\ph} \|^2_2}{d_A} \right| \leq \e \| \ket{\psi} - \ket{\ph} \|^2_2 \leq 4\e.
\]
By expanding $\| \ket{\psi_a} - \ket{\ph_a} \|^2_2$ and $\| \ket{\psi} - \ket{\ph} \|^2_2$, we obtain using the triangle inequality
\begin{align*}
\sum_{a \in [d_A]} \left| 2 \re \braket{\psi_a}{\ph_a} - \frac{2 \re \braket{\psi}{\ph}}{d_A} \right| &\leq 4\e + \sum_{a \in [d_A]} \left| \| \ket{\psi_a} \|^2_2 - \frac{\| \ket{\psi} \|^2_2}{d_A} \right| + \left| \| \ket{\ph_a} \|^2_2 - \frac{\| \ket{\ph} \|^2_2}{d_A} \right| \\
					&\leq 6\e.
\end{align*}
In the last inequality, we used equation \eqref{eq:dvo-like} for $\ket{\psi}$ and $\ket{\ph}$.
The same argument can be applied to $i \ket{\psi}$ and $\ket{\ph}$ to get 
\[
2 \sum_{a \in [d_A]} \left| \im \braket{\psi_a}{\ph_a} - \frac{\im \braket{\psi}{\ph}}{d_A} \right| \leq 6\e
\]
Thus, substituting in equation \eqref{eq:jl-dvo-2} we obtain
\[
\sum_a \left| |\braket{\psi_a}{\ph_a}|^2 - \frac{|\braket{\psi}{\ph}|^2}{d_A} \right| \leq 12\e.
\]

We now consider the second term in \eqref{eq:jl-dvo-1}. We have, using the Cauchy-Schwarz inequality, $\frac{|\braket{\psi_a}{\ph_a}|}{\| \ket{\psi_a} \|_2 \| \ket{\ph_a} \|_2} \leq 1$. Hence,
\begin{align*}
& \sum_{a \in [d_A]} \left| |\braket{\psi_a}{\ph_a}|^2 - \frac{|\braket{\psi_a}{\ph_a}|^2}{d_A \| \ket{\psi_a} \|_2  \| \ket{\ph_a} \|_2} \right| \\
&\leq \sum_{a \in [d_A]} \left| \| \ket{\psi_a} \|_2  \| \ket{\ph_a} \|_2 - \frac{1}{d_A} \right| \\
&\leq \sum_{a \in [d_A]} \| \ket{\psi_a} \|_2 \left| \| \ket{\ph_a} \|_2 - \frac{1}{\sqrt{d_A}} \right|
+  \sum_{a \in [d_A]} \left| \frac{\| \ket{\psi_a} \|_2}{\sqrt{d_A}}  - \frac{1}{d_A} \right| \\
&\leq \sqrt{ \sum_{a \in [d_A]} \| \ket{\psi_a} \|^2_2} \sqrt{\sum_{a \in [d_A]} \left| \| \ket{\ph_a} \|_2 - \frac{1}{\sqrt{d_A}} \right|^2} + \sqrt{\sum_{a \in [d_A]} \left| \| \ket{\psi_a} \|_2 - \frac{1}{\sqrt{d_A}} \right|^2} \\
&\leq \sqrt{  \sum_{a \in [d_A]} \left| \| \ket{\ph_a} \|^2_2 - \frac{1}{d_A} \right| } + \sqrt{  \sum_{a \in [d_A]} \left| \| \ket{\psi_a} \|^2_2 - \frac{1}{d_A} \right| } \\
&\leq 2 \sqrt{\e}.
\end{align*}
For the third inequality, we used once again the Cauchy-Schwarz inequality and for the fourth inequality, we used the fact that $\sum_{a \in [d_A]} \| \ket{\psi_a} \|^2_2 = \| V \ket{\psi} \|^2_2 = 1$ and the inequality $|x - y|^2 \leq |x-y||x+y| = |x^2 - y^2|$ for all nonnegative $x,y$. Plugging this bound into equation \eqref{eq:jl-dvo-1}, we obtain the desired result.
\end{proof}

\begin{center}
\begin{figure}[t]
\begin{center}
\includegraphics[scale=0.9]{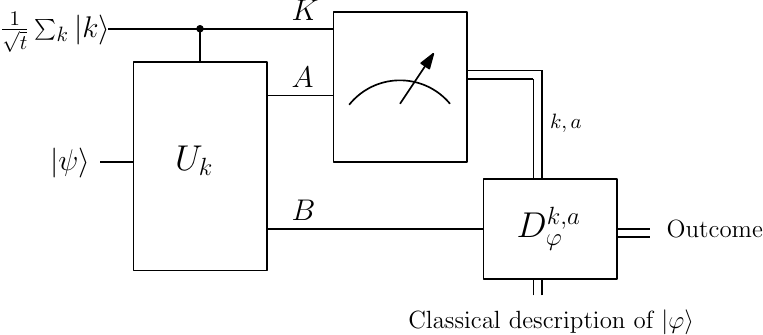}
\end{center}
\caption{Quantum identification based on a metric uncertainty relation. The system $K$ is prepared in a uniform superposition state $\frac{1}{\sqrt{t}} \sum_k \ket{k}$. Then, controlled by system $K$, the unitary $U_k$ is applied to $C = A \otimes B$, where the unitary transformations $\{U_k\}$ satisfy a metric uncertainty relation. The $KA$ system is then measured in its computational basis. The outcome $k,a$ of this measurement is sent through the classical channel. The system $B$ is sent using the noiseless quantum channel. The receiver constructs a POVM $D^{k,a}_{\ph}$ based on a classical description of the state $\ket{\ph}$ he wishes to test for and the classical communication $k,a$ he receives.}
\label{fig:qid-code}
\end{figure}
\end{center}
\begin{theorem}[Quantum identification using classical communication]
\label{thm:identification}
Let $n$ be a positive integer and $\e \in (2^{-c'n}, 1)$ where $c'$ is a constant independent of $n$. Then for some $m = O(\log(1/\e))$, $\e$-quantum identification of $n$-qubit states can be performed using a single message of $n$ bits and $m$ qubits.

Moreover, for some $m = O(\log(n/\e) \cdot \log(n))$, $\e$-quantum identification of $n$-qubit states can be performed using a single message of $n$ bits and $m$ qubits with an encoding quantum circuit of polynomial size.
\end{theorem}
\comment{
\begin{myremark}[]
Note that $\log d_A \leq n$ uses of the classical channel $\bitchannel$ would be sufficient. We do not include it to simplify the statement of the theorem.
\end{myremark}
}
\begin{proof}
Let $\{U_0, \dots, U_{t-1}\}$ be a set of unitaries on $n$ qubits satisfying an $\e'$- metric uncertainty relation with $\e' = 1/2 \cdot (\e/28)^2$.
We start by preparing the uniform superposition $\frac{1}{\sqrt{t}} \sum_{k=0}^{t-1} \ket{k}^K$ and apply the unitary $U_k$ on system $C$ controlled by the register $K$. We get the state $\frac{1}{\sqrt{t}} \sum_k \ket{k}^K (U_k \ket{\psi})^{AB} = \sum_{k,a} \ket{k}^K \ket{a}^A \ket{\psi_{k,a}}^B$ for some unnormalized vectors $\ket{\psi_{k,a}} \in B$. Alice then measures the system $KA$ in the computational basis obtaining an outcome $k,a$ and sends $k,a$ and $\ket{\hat{\psi}_{k,a}}$ to Bob, where $\ket{\hat{\psi}_{k,a}} = \ket{\psi_{k,a}}/\| \ket{\psi_{k,a}} \|_2$. Observe that $\sum_{k,a} \| \ket{\psi_{k,a}} \|^2_2 = 1$ and $\| \ket{\psi_{k,a}} \|^2_2 = \frac{1}{t} \cdot p^A_{U_k \ket{\psi}}(a)$ so that the metric uncertainty relation property can be written as
\begin{equation}
\label{eq:identification-1}
\frac{1}{2} \sum_{k,a} \left| \| \ket{\psi_{k,a}} \|^2_2 - \frac{1}{t d_A} \right| \leq \e'.
\end{equation}
This shows that the isometry $\ket{\psi} \mapsto \frac{1}{\sqrt{t}} \sum_k \ket{k}^K (U_k \ket{\psi})^{AB}$ satisfies the condition \eqref{eq:dvo-like} of Lemma \ref{lem:dvo-jl}.

The decoding POVMs for received classical information $k,a$ and state $\ket{\ph}$ are defined by $D^{k,a}_{\ph} = \proj{\hat{\ph}_{k,a}}$ where $\frac{1}{\sqrt{t}} \sum_k \ket{k}^K (U_k \ket{\ph})^{AB} = \sum_{k,a} \ket{k}^K \ket{a}^A \ket{\ph_{k,a}}^B$ and $\ket{\hat{\ph}_{k,a}} = \ket{\ph_{k,a}}/\| \ph_{k,a} \|_2$. 
The protocol is illustrated in Figure \ref{fig:qid-code}.

We now analyse the probability that Bob outputs $1$. Recall that outcome $1$ corresponds to the projector $\proj{\ph}$. The probability that the protocol in Figure \ref{fig:qid-code} outputs $1$ is  
\[
\sum_{k,a} \| \ket{\psi_{k,a}} \|^2_2 \cdot \tr\left[ D^{k,a}_{\ph} \proj{\hat{\psi}_{k,a}}\right] = \sum_{k,a} \| \ket{\psi_{k,a}} \|^2_2 |\braket{\hat{\psi}_{k,a}}{\hat{\ph}_{k,a}} |^2.
\]
Applying Lemma \ref{lem:dvo-jl}, we get
\begin{equation}
\label{eq:identification-2}
\frac{1}{td_A} \sum_{k,a} \left| |\braket{\hat{\psi}_{k,a}}{\hat{\ph}_{k,a}} |^2 - |\braket{\psi}{\ph}|^2\right| \leq 14 \sqrt{2\e'} = \e/2.
\end{equation}
Using the triangle inequality, equations \eqref{eq:identification-2} and \eqref{eq:identification-1}, we obtain
\begin{align*}
&\sum_{k,a}  \| \ket{\psi_{k,a}} \|^2_2 \left| |\braket{\hat{\psi}_{k,a}}{\hat{\ph}_{k,a}} |^2
- |\braket{\psi}{\ph}|^2\right| \\
&\leq \sum_{k,a} \frac{1}{td_A} \left| |\braket{\hat{\psi}_{k,a}}{\hat{\ph}_{k,a}} |^2 - |\braket{\psi}{\ph}|^2\right| + \sum_{k,a} \left| \| \ket{\psi_{k,a}} \|^2_2 - \frac{1}{td_A} \right| \cdot 2 \\
&\leq \e/2 + 4 \e' \leq \e.
\end{align*}
Thus, the probability of obtaining outcome $1$ is in the interval $[|\braket{\psi}{\ph}|^2 - \e, |\braket{\psi}{\ph}|^2 + \e]$. 

We conclude by using the metric uncertainty relations of Theorems \ref{thm:existence-ur} and \ref{thm:explicit-ur2}. For the explicit construction, we still need to argue that the encoding can be computed by a quantum circuit of size $O(n^2 \polylog(n/\e))$ and depth $O(n\polylog(n/\e))$ using classical precomputation. \red{Note that we need to apply the unitaries from Theorem \ref{thm:explicit-ur2} in superposition. This means that in order to obtain an efficient circuit, we need an efficient way of mapping the integer $k$ to a description of the circuit computing $U_k$. In order to obtain the desired size for the quantum circuit, we substitue the $1$-MUBs of Lemma \ref{lem:explicitmub} in the construction of Theorem \ref{thm:explicit-ur2}. The reason is that Lemma \ref{lem:approx-mub} concerning approximate MUBs does not give an explicit bound on the running time of the procedure that determines the circuit as a function of $k$.} The only thing we need to precompute is an irreducible polynomial of degree $n$ over $\FF_2[X]$. Then, using the same argument as in the proof of Lemma \ref{lem:explicitmub}, we can compute the unitary operation that takes as input the state $\ket{j} \otimes \ket{\psi}$ and outputs the state $\ket{j} \otimes V_j \ket{\psi}$ using a circuit of size $O(n^2 \polylog n)$ and depth $O(n \polylog n)$. Since the permutation extractor we use can be implemented by a quantum circuit of size $O(n \polylog(n/\e))$, the unitary transformation $\ket{k} \otimes \ket{\psi} \mapsto \ket{k} \otimes U_k \ket{\psi}$ can be computed by a quantum circuit of size $O(n^2 \polylog(n/\e))$ and depth $O(n \polylog(n/\e))$.
\end{proof}




This result can be thought of as an analogue of the well-known fact that the public-coin randomized communication complexity of equality is $O(\log(1/\e))$ for an error probability $\e$ \cite{KN97}. Quantum communication replaces classical communication and classical communication replaces public random bits. Classical communication can be thought of as an extra resource because on its own it is useless for quantum identification \cite[Theorem 11]{HW10}.

\section{Conclusion}

We have seen how the problem of finding uncertainty relations is closely related to the problem of finding large almost Euclidean subspaces of $\ell_1(\ell_2)$. Even though we did not use any norm embedding result directly, many of the ideas presented here come from the proofs and constructions in the study of the geometry of normed spaces.
In particular, we obtained an explicit family of bases that satisfy a strong metric uncertainty relation by adapting a construction of \citeN{Ind07}. Moreover, using standard techniques from asymptotic geometric analysis, we were able to prove a strong result on the uncertainty relations defined by random unitaries \cite{HLSW04}.

We used these uncertainty relations to exhibit strong locking effects. In particular, we obtained the first explicit construction of a method for encrypting a random $n$-bit string in an $n$-qubit state using a classical key of size polylogarithmic in $n$. Moreover, our non-explicit results give better key sizes than previous constructions while simultaneously meeting a stronger locking definition. In particular, we showed that an arbitrarily long message can be locked with a constant-sized key. Our results on locking are summarized in Table \ref{tbl:locking-results}. We should emphasize that, even though we presented information locking from a cryptographic point of view, it is not a composable primitive because an eavesdropper could choose to store quantum information about the message instead of measuring. For this reason, a locking scheme has to be used with great care when composed with other cryptographic primitives.

We also used uncertainty relations to construct quantum identification codes. We proved that it is possible to identify a quantum state of $n$ qubits by communicating $n$ classical bits and $O(\log(1/\e))$ quantum bits. We also presented an efficient encoder for this problem that uses $O(\log n \log (n/\e))$ qubits of communication instead. The main weakness of this result is that the decoder uses a classical description of the state $\ket{\ph}$ that is in general exponential in the number of qubits of $\ket{\ph}$. But as shown in \cite[Remark 16]{Win04}, if Bob was to receive a copy of the quantum state $\ket{\ph}$, there is no strategy for Alice that is asymptotically more efficient in terms of communication than sending her state to Bob.

We expect to see more applications to quantum information theory of the tools used in the theory of pseudorandomness. An interesting open question is whether these techniques can be helpful in constructing explicit subspaces containing only highly entangled states. We say that a subspace $S$ of the Hilbert space of a bipartite system $A \otimes B$ is a maximally entangled subspace if \emph{all} states in $S$ have a marginal on $A$ with almost maximal entropy. Using probabilistic arguments, \citeN{HLW04} showed that subspaces with this property can be surprisingly large, e.g., almost as large as $A \otimes B$ if we use the von Neumann entropy. Such subspaces are related to one of the central problems in quantum information theory:  the classical capacity of a quantum channel. Unlike for classical channels, there is no known computable formula for the classical capacity of a general quantum channel. The main difficulty in the quantum setting is the possibility of having entanglement between the inputs of subsequent uses of the channel. It remained unknown for a long time whether entangled inputs can be beneficial for the transmission of classical information, or in other words whether the Holevo information is additive. Recently, using a probabilistic construction of a maximally entangled subspace, \citeN{Has09} constructed a channel for which the Holevo information is not additive; see also \cite{HW08} for violations of related additivity questions.
In a different context, maximally entangled subspaces were also used to build protocols for superdense coding of quantum states \cite{HHL04}. These applications motivate the search for explicit constructions of maximally entangled subspaces. As shown by Aubrun, Szarek and Werner \shortcite{ASW10,ASW10b}, this problem  amounts to finding explicit almost Euclidean sections for matrix spaces endowed with Schatten $p$-norms, which corresponds to the $\ell_p$ norm of the singular values. In addition to the applications in quantum information theory, such almost Euclidean sections are closely related to rank minimization problems for which the nuclear norm heuristic allows exact recovery \cite{DF10}.

\appendix

\section*{APPENDIX}

\section{Existence of metric uncertainty relations}

In this section, we prove the lemmas used in Theorem \ref{thm:existence-ur}.

\label{sec:app-existence-ur}

We start with Lemma \ref{lem:expl1l2} on the average value of $\ell^A_1(\ell^B_2)$ on the sphere. Note that using  \eqref{eq:fidelity-l1l2}, the formulation here is equivalent to the original formulation.
\newtheorem{lemma-expl1l2}{Lemma \ref{lem:expl1l2} \: (Average value of $\ell^A_1(\ell^B_2)$ on the sphere)}
\begin{lemma-expl1l2}[] 
Let $\ket{\ph}^{AB}$ be a random pure state on $AB$. Then,
\[
\ex{ \| \ket{\ph}^{AB} \|_{\ell^A_1(\ell^B_2)} } =  d_A \frac{\Gamma(d_B+\frac{1}{2})}{\Gamma(d_B)} \frac{\Gamma(d_A d_B)}{\Gamma(d_A d_B + \frac{1}{2})} \geq \sqrt{1 - \frac{1}{d_B}} \sqrt{d_A}.
\]
where $\Gamma$ is the Gamma function $\Gamma(z) = \int_{0}^{\infty} u^{z-1} e^{-u} du$ for $z \geq 0$.
\end{lemma-expl1l2}
\begin{proof}
The presentation uses methods described in \cite{Bal97}.

Observe that the random variable $\| \ket{\ph}^{AB} \|_{12}$ is distributed as the $\ell_1^{d_A}(\ell_2^{2d_B})$ norm of a \emph{real} random vector chosen according to the rotation invariant measure on the sphere $\Sp^{2d_Ad_B-1}$. We define for integers $n$ and $m$ the norm $\ell_1^{n}(\ell_2^{m})$ of a real $n \cdot m$-dimensional vector $\{v_{i,j}\}_{i \in [n], j \in [m]}$ as for the complex case (Definition \ref{def:l1l2})
\[
\| v \|_{\ell_1^{n}(\ell_2^{m})} = \sum_i \sqrt{ \sum_j |v_{i,j}|^2 }.
\]
Note that we only specify the dimension of the systems as the systems themselves are not relevant here. In the rest of the proof, we use $\| \cdot \|_{12}$ as a shorthand for $\| \cdot \|_{\ell^{d_A}_1(\ell^{2d_B})}$. 
Our objective is to evaluate the expected value $\ex{\| \Theta \|_{12}}$ where $\Theta$ has rotation invariant distribution on the real sphere $\Sp^{s-1}$ and $s = 2d$  with $d = d_Ad_B$.  For this, we start by relating the $\ex{\|Z\|_{12}}$ and $\ex{\| \Theta \|_{12}}$ where $Z$ has a standard Gaussian distribution on $\RR^{s}$ . By changing to polar coordinates, we get
\begin{align*}
\ex{ \| Z \|_{12} }	 &= \int_{\RR^s} \| x \|_{12} \frac{e^{-\frac{1}{2}\sum_{i=1}^s x^2_i}}{(2\pi)^{s/2}} dx \\
			 &= \int_{0}^{\infty} \int_{\mathbb{S}^{s-1}} \| r \theta \|_{12} \frac{e^{-r^2/2}}{(2\pi)^{s/2}} \cdot \frac{s \pi^{s/2} d\sigma(\theta) }{\Gamma(\frac{s}{2}+1)} r^{s-1} dr
\end{align*}
where $\sigma$ is the normalized rotation-invariant measure on $\mathbb{S}^{s-1}$. The term $\frac{s \pi^{s/2}}{\Gamma(\frac{s}{2}+1)}$ is the surface area of the sphere in dimension $s-1$. Using the equality $\Gamma(z+1) = z \Gamma(z)$, we have $\frac{s \pi^{s/2}}{\Gamma(\frac{s}{2}+1)} = \frac{2 \pi^{s/2}}{\Gamma(\frac{s}{2})}$. Thus,
\begin{align*}
\ex{ \| Z \|_{12} }	 &= \frac{2 \pi^{s/2} }{(2\pi)^{s/2}\Gamma(\frac{s}{2})} \int_{0}^{\infty} r^{s} e^{-r^2/2} dr \cdot \int_{\mathbb{S}^{s-1}} \| \theta \|_{12} d\sigma(\theta) \\
				&= \frac{1}{2^{s/2-1}\Gamma(\frac{s}{2})} \int_{0}^{\infty} r^{s} e^{-r^2/2} dr \cdot \int_{\mathbb{S}^{s-1}} \| \theta \|_{12} d\sigma(\theta) \\
\end{align*}
We then perform a change of variable $u = r^2/2$:
\begin{align}
\ex{ \| Z \|_{12} }	 &=  \frac{1}{2^{s/2-1}\Gamma(\frac{s}{2})} \int_{0}^{\infty} (2u)^{(s-1)/2} e^{-u} du \cdot \int_{\mathbb{S}^{s-1}} \| \theta \|_{12} d\sigma(\theta) \notag \\
			&= \frac{2^{(s-1)/2} \Gamma(\frac{s-1}{2}+1)}{2^{s/2-1} \Gamma(\frac{s}{2})} \cdot \int_{\mathbb{S}^{s-1}} \| \theta \|_{12} d\sigma(\theta) \notag \\
			&= \frac{ \sqrt{2} \Gamma(\frac{s+1}{2})}{\Gamma(\frac{s}{2})} \cdot \ex{\| \Theta \|_{12}}. \label{eq:gaussian-haar}
\end{align}
Now, we compute
\begin{align*}
\ex{ \| Z \|_{12} }	 &= \int_{\RR^s} \| x \|_{12} \frac{e^{-\frac{1}{2}\| x \|_2^2}}{(2\pi)^{s/2}} dx \\
			&= \sum_{i=1}^{d_A} \int_{\RR^s}  \| x_i \|_2  \frac{e^{-\frac{1}{2} \| x \|_2^2 }}{(2\pi)^{s/2}} dx
\end{align*}
where we decomposed $x = (x_1, \dots, x_{d_A})$ where $x_i \in \RR^{2d_B}$. As all the terms of the sum are equal
\begin{align*}
\ex{ \| Z \|_{12} }&= d_A \int_{\RR^{2d_B}}  \| x_0 \|_2  \frac{e^{-\frac{1}{2} \| x_0 \|_2^2 }}{(2\pi)^{d_B}} dx_0 \left(\int_{\RR^{2d_B}}  \frac{e^{-\frac{1}{2} \| x_1 \|_2^2 }}{(2\pi)^{d_B}} dx_1 \right)^{d_A-1}  \\
			&= d_A \frac{\sqrt{2} \Gamma(\frac{2d_B+1}{2})}{\Gamma(d_B)} \int_{\mathbb{S}^{2d_B-1}} \|\theta \|_2 d\sigma(\theta) \\
			&= d_A \frac{ \sqrt{2} \Gamma(\frac{2d_B + 1}{2})}{\Gamma(d_B)}.
\end{align*}
To get the second equality, we use the same argument as for equation \eqref{eq:gaussian-haar}.
We conclude using equation \eqref{eq:gaussian-haar}
\begin{align*}
\ex{\| \ket{\ph} \|_{\ell_1^A(\ell_2^B)}} &= \ex{\| \Theta \|_{12}} \\
	&= d_A \frac{\Gamma(d_B+\frac{1}{2})}{\Gamma(d_B)} \cdot \frac{\Gamma(d_A d_B)}{\Gamma(d_A d_B + \frac{1}{2})}.
\end{align*}

We now prove the inequality in the statement of the lemma. We use the following two facts about the $\Gamma$ function: $\log \Gamma$ is convex and for all $z > 0$, $\Gamma(z+1) = z \Gamma(z)$. The first property can be seen by using H\"older's inequality for example and the second using integration by parts. Using these properties, we have
\begin{align*}
\log \Gamma\left(x+\frac{1}{2}\right) &\leq \frac{1}{2} \log \Gamma(x) + \frac{1}{2} \log \Gamma(x+1) \\
						&= \frac{1}{2} \log \left( x \Gamma(x)^2 \right) \\
						&= \log \left( \sqrt{x} \Gamma(x) \right).
\end{align*}
Thus, $\frac{\Gamma(x+\frac{1}{2})}{\Gamma(x)} \leq \sqrt{x}$. Similarly, we have $\frac{\Gamma(x)}{\Gamma(x-\frac{1}{2})} \leq \sqrt{x-\frac{1}{2}}$ which implies that $\frac{\Gamma(x+\frac{1}{2})}{\Gamma(x)}  \geq \sqrt{x-\frac{1}{2}}$ when writing $\Gamma(x+1/2) = (x-1/2) \Gamma(x-1/2)$.

We conclude that 
\begin{align*}
\ex{\| \ket{\ph} \|_{\ell_1^A(\ell_2^B)}} &\geq d_A \cdot \sqrt{d_B - \frac{1}{2}} \frac{1}{\sqrt{d_A d_B}} \\
		&= \sqrt{d_A} \cdot \sqrt{1 - \frac{1}{2d_B}} \geq \sqrt{d_A} \cdot \sqrt{1 - \frac{1}{d_B}}.
\end{align*}
\end{proof}

We now state two more standard results that can be used in place of Lemma \ref{lem:conc-product-sphere}. First, we state a version of L\'evy's lemma on the concentration of the rotation-invariant measure on the sphere presented in \cite{MS86}. Note that this is not the standard version of L\'evy's lemma which uses the median instead of the expectation.
\begin{lemma}[L\'evy's lemma]
\label{lem:levy}
Let $f : \CC^d \to \RR$ and $\eta > 0$ be such that for all pure states $\ket{\ph_1}, \ket{\ph_2}$ in $\CC^d$, 
\[
| f(\ket{\ph_1}) - f(\ket{\ph_2}) | \leq \eta \| \ket{\ph_1} - \ket{\ph_2} \|_2.
\]
Let $\ket{\ph}$ be a random pure state in dimension $d$. Then for all $0 \leq \delta$,
\[
\pr{| f(\ket{\ph}) - \ex{f(\ph)} | \geq \delta } \leq 4 \exp{- \frac{\delta^2 d}{c \eta^2} }
\]
where $c$ is a constant. We can take $c = 18 \pi^2$.
\end{lemma}
\begin{proof}
We can instead study the concentration of a Lipschitz function on the real sphere $\Sp^{2d-1}$. Note that the induced function (that we also call $f$) is still $\eta$-Lipschitz. Concentration on $\Sp^{2d-1}$ can be proved in a simple way using concentration of the standard Gaussian distribution. This proof is due to Maurey and Pisier and can be found in \cite[Appendix V]{MS86}. Specifically, using \cite[Corollary V.2]{MS86}, we get 
\begin{align*}
\pr{| f(\ket{\ph}) - \ex{f(\ket{\ph})} | \geq \delta} &\leq 2 \exp{-\frac{\delta^2 (2d)}{18 \pi^2 \eta^2}} + 2 \exp{-\frac{2d}{2 \pi^2}} \\
						&\leq 4\exp{-\frac{\delta^2 d}{9 \pi^2 \eta^2}}. 
\end{align*}
In the notation of the proof of \cite[Corollary V.2]{MS86}, we have set $\delta = 1/(2\sqrt{2})$. This can be done because using the same arguments as in the proof of Lemma \ref{lem:expl1l2}, we can show that the expected $\ell_2$ norm of the standard real Gaussian distribution in dimension $2d$ is at least $\sqrt{2}\sqrt{d-\frac{1}{2}} > \sqrt{d}$ for $d\geq 2$.

We used this version of L\'evy's lemma because it has an elementary proof and it gives directly the concentration about the expected value. Different versions involving the median of $f$ and giving better constants can be found in  \cite[Corollary 2.3]{MS86} or \cite[Proposition 1.3]{Led01} for example.
\end{proof}

The following lemma proves that the average of independent random variable having sub-Gaussian tails is well-concentrated around its expectation. The proof uses standard techniques for proving concentration inequalities.
\begin{lemma}
\label{lem:avgconc}
Let $a, b \geq 1$, and $t$ a positive integer. Suppose $X$ is a random variable with zero mean, satisfying the tail bounds for all $\eta > 0$ 
\[
\pr{X \geq \eta} \leq a e^{-b\eta^2} \quad \text{ and } \quad \pr{X \leq -\eta} \leq ae^{-b\eta^2}.
\]
Let $X_1, \dots X_t$ be independent copies of $X$. Then if $\delta \geq 0$ and $\delta^2 b \geq 16 a^2 \pi $,
\[
\pr{\left| \frac{1}{t} \sum_{k=1}^t X_k \right| \geq \delta} \leq \exp{-\frac{\delta^2 bt }{2}}.
\]
\end{lemma}
\begin{proof}
For any $\lambda > 0$, using Markov's inequality
\begin{align*}
\pr{\sum_{k=1}^t X_k \geq t \delta} &= \pr{\exp{\lambda \sum_{k=1}^t X_k} \geq \exp{\lambda t \delta} } \\
							&\leq \ex{\exp{\lambda \sum_{k=1}^t X_k}} e^{-\lambda t \delta} \\
							&= \ex{e^{\lambda X}}^t e^{-\lambda t \delta}.
\end{align*}
We now bound the moment generating function $\ex{e^{\lambda X}}$ of $X$ using the tail bounds.
\begin{align*}
\ex{e^{\lambda X}} &= \int_{0}^{\infty} \pr{e^{\lambda X} \geq u} du \\
				&= \int_{0}^{\infty} \pr{X \geq \frac{\ln u}{\lambda}} du \\
				&= \int_{0}^1 \pr{X \geq \frac{\ln u}{\lambda}} du +  \int_{1}^{\infty} \pr{X \geq \frac{\ln u}{\lambda}} du \\
				&\leq 1 + \int_{1}^{\infty} a \exp{-\frac{b \ln^2 u}{\lambda^2}} du \\
				&= 1 + a \int_{0}^{\infty} \exp{-\frac{b z^2}{\lambda^2}} e^{z} dz
\end{align*}
by making the change of variable $z = \log u$.
\begin{align*}				
\ex{e^{\lambda X}}			&\leq 1 + a \int_{0}^{\infty} \exp{-\frac{b}{\lambda^2}\left(z - \frac{\lambda^2}{2 b} \right)^2 + \frac{\lambda^2}{4b}} dz \\
				&\leq 1 + a \exp{\frac{\lambda^2}{4 b}} \int_{-\infty}^{\infty} \exp{-\frac{b}{\lambda^2}\left(z - \frac{\lambda^2}{2 b} \right)^2 } dz \\
				&= 1 + a \exp{\frac{\lambda^2}{4 b}} \frac{\lambda}{\sqrt{2 b}} \int_{-\infty}^{\infty} \exp{-\frac{u^2}{2} } du \\
				&= 1 + a\frac{\sqrt{2\pi}\lambda}{\sqrt{2b}} \cdot \exp{\frac{\lambda^2}{4 b}} \\
				&\leq 2 \max \left(1, a \frac{\sqrt{\pi}\lambda}{\sqrt{b}} \cdot \exp{\frac{\lambda^2}{4 b}} \right).
\end{align*}
We choose $\lambda = 2 \delta b$ (this is not the optimal choice but it makes expressions simpler),
\begin{align*}
\pr{\sum_{k=1}^t X_k \geq t \delta} &\leq \max\left( 2^t, \left(2 a \frac{\sqrt{\pi}\lambda}{\sqrt{b}}\right)^t \cdot \exp{\frac{\lambda^2 t}{4 b}} \right) \exp{-\lambda t \delta} \\
							&= \max \left( \exp{-2 \delta^2 bt + t \ln 2}, \exp{ \delta^2 bt - 2 \delta^2 bt + t \ln ( 4 a \sqrt{\pi} \delta \sqrt{b} ) } \right) \\
							&= \max \left\{ \exp{ \left(-2 \delta^2 b + \ln 2\right) t}, \exp{ \left( -\delta^2 b + \ln ( 4 a \sqrt{\pi} \delta \sqrt{b} ) \right) t } \right\}.
\end{align*}
\begin{claim}[]
For all $c \geq 1$ and $x \geq c$
\[
\frac{1}{2} \ln(cx) - x \leq -\frac{x}{2}.
\]
\end{claim}
The function $x \mapsto \frac{x}{2} - \frac{1}{2}\ln(cx)$ is increasing for $x \geq 1$. It suffices to show that it is nonnegative for $x = c$. To see that, we differentiate the function $y \mapsto y - \ln(y^2)$ to prove that for all $y \geq 1$, 
we have $y - \ln(y^2) \geq 0$. This proves the claim.

Using this inequality, we have for $\delta^2 b \geq 16 a^2 \pi$,
\[
-\delta^2 b + \ln ( 4 a \sqrt{\pi} \delta \sqrt{b} ) \leq -\frac{\delta^2b}{2} \quad \text{ and } \quad -2 \delta^2 b + \ln 2 \leq -\frac{\delta^2b}{2}.
\]
Finally,
\[
\pr{\sum_{k=1}^t X_k \geq t \delta} \leq \exp{-\frac{\delta^2bt}{2}}.
\]
\end{proof}


\section{Efficient mutually unbiased bases}
\label{sec:app-explicitmub}

\begin{proof}[of Lemma \ref{lem:explicitmub}]
We define $V_1 = \1$, and the remaining unitaries are indexed by binary vectors $u \in \{0,1\}^n$, for example the binary representations of integers from $0$ to $r-2$. The construction is based on operations in the finite field $\FF_{2^n}$. The field $\FF_{2^n}$ can be seen as an $n$-dimensional vector space over $\FF_2$. Choose $\theta \in \FF_{2^n}$ such that $1, \theta, \dots, \theta^{n-1}$ form a basis of $\FF_{2^n}$. For any $x,y \in [n]$, $\theta^{x} \cdot \theta^{y} \in \FF_{2^n}$ can be decomposed in our chosen basis as
$\theta^{x} \cdot \theta^{y} = \sum_{\ell=0}^{n-1} m_{\ell}(x,y) \theta^{\ell}$ for some $m_{\ell}(x,y) \in \FF_2$. We can thus define the matrices $M_0, M_1, \dots, M_{n-1}$ from the multiplication table
\[
\left( \begin{array}{c}
1 \\
\theta  \\
\vdots \\
\theta^{n-1} 
\end{array} \right) \cdot
\left( \begin{array}{cccc}
1 & \theta & \hdots & \theta^{n-1}
\end{array} \right) = M_0  + M_1 \theta + \dots + M_{n-1} \theta^{n-1}. 
\]
where $M_{\ell} = (m_{\ell}(x,y))_{x,y \in [n]}$.
For a given $u \in \{0,1\}^n$, we define the matrix
\[
N_u = \sum_{\ell=0}^{n-1} u_{\ell} M_\ell.
\]
Notice that as $\theta^x \cdot \theta^y = \theta^{x+y}$, the entry $N_u(x,y)$ of $N_u$ only depends on $x+y$, i.e., $N_u(x,y) = N_u(x',y')$ if $x+y = x'+y'$. So we can represent this matrix by a vector $\alpha_u(x+y) = N_u(x,y)$ of length $2n-1$.
We then define a $\ZZ_4$-valued quadratic form by: for $v \in \{0,1\}^n$,
\[
T_u(v) = v^{T} N_u v \mod 4.
\]
Note that the operations $v^{T} N_u v$ are not performed in $\FF_2$ but rather in $\ZZ$. Using the vector $\alpha_u$, we can write
\[
T_u(v) = \sum_{x,y \in [n]} v_x N_u(x,y) v_y \mod 4 = \sum_{z=0}^{2n-2} \left(\sum_{x=0}^{z} v_x v_{z-x} \right) \alpha_u(z) \mod 4
\]
if we define $v_x = 0$ for $x \geq n$.
We then define the diagonal matrix $D_u = \textrm{diag}\left(i^{T_u(v)} \right)_{v \in \FF_2^n}$.
Finally, we define for $1 \leq j \leq r-1$,
\[
V_j = D_{\textrm{bin}(j-1)} H^{\otimes n}
\]
where $\textrm{bin}(j) \in \{0,1\}^n$ is the binary representation of length $n$ of the integer $j$.

The fact that these unitaries define mutually unbiased bases was proved in \cite{WF89}. We now analyse how fast these unitary transformations can be implemented. Note that we want a circuit that takes as input a state $\ket{\psi}$ together with  the index $j$ of the unitary transformation and outputs $V_j \ket{\psi}$.

Given the index $j$ as input, we show it is possible to compute $u = \textrm{bin}(j-1)$ and compute the vector $\alpha_j \eqdef \alpha_u$ in time $O(n^2 \polylog n)$. In fact, we start by computing a representation of the field $\FF_{2^n}$ by finding an irreducible polynomial $Q$ of degree $n$ in $\FF_2[X]$, so that $\FF_{2^n} = \FF_2[X]/Q$. This can be done in expected time $O(n^2 \polylog n)$ (Corollary 14.43 in the book \cite{GG99}). There also exists a deterministic algorithm for finding an irreducible polynomial in time $O(n^{4} \polylog n)$ \cite{Sho90}. We then take $\theta = X$. Computing the polynomial $X^x \cdot X^y = X^{x+y} \mod Q$ can be done in time $O(n \polylog n)$ using the fast Euclidean algorithm (see Corollary 11.8 in \cite{GG99}). As $x+y \in [0, 2n-2]$, we can explicitly represent all the polynomials $X^{z}$ for $0 \leq z \leq 2n-2$ in time $O(n^2 \polylog n)$. It is then simple to compute the vector $\alpha_u$ using the vector $u$ in time $O(n^2)$.


To build the quantum circuit, we first observe that applying a Hadamard transform only takes $n$ single-qubit Hadamard gates. 
Then, to design a circuit performing the unitary transformation $D_{\textrm{bin}(j-1)}$, we start by building a classical circuit that computes 
\[
T_u(v) = \sum_{z=0}^{2n-2} \left(\sum_{x=0}^{z} v_x v_{z-x} \right) \alpha_u(z) \mod 4
\]
on inputs $v$ and $\alpha_u$.
Observing that $\sum_{x=0}^{z} v_x v_{z-x}$ is the coefficient of $Y^z$ in the polynomial $\left(\sum_{x=0}^{n-1} v_x Y^x\right)^2$, we can use fast polynomial multiplication to compute $T_u(v)$ in time $O(n \polylog n)$ (Corollary 8.27 in \cite{GG99}). 
This circuit can be transformed into a reversible circuit with the same size (up to some multiplicative constant) that takes as input $(v, \alpha_j, g)$ where $v \in \{0,1\}^n$, $\alpha_j \in \{0,1\}^{2n-1}$ and $g \in \ZZ_4$, and outputs $(v, \alpha_j, g + T_u(v) \mod 4)$.

This reversible classical circuit can be readily transformed into a quantum circuit that computes the unitary transformation defined by $W: \ket{v} \ket{g} \mapsto \ket{v} \ket{g + T_u(v) \mod 4}$. Recall that we want to implement the transformation $D_u:\ket{v} \mapsto i^{T_u(v)} \ket{v}$ efficiently. This is simple to obtain using the quantum circuit for $W$. In fact, if we use a catalyst state $\ket{\phi} = \ket{0} - i \ket{1} - \ket{2} + i \ket{3}$, we have
\[
W \ket{v} \ket{\phi} = i^{T_u(v)} \ket{v} \ket{\phi} = D_{\textrm{bin}(j-1)} \ket{v} \ket{\phi}.
\]
Finally, $D_{\textrm{bin}(j-1)} H^{\otimes n}$ can be implemented by a quantum circuit of size $O(n \polylog n)$.\end{proof}


\section{Permutation extractors}
\label{sec:app-perm-extractor}

In order to prove the existence of strong permutation extractors with good parameters, we use the construction of \citeN{GUV09} which is inspired by list decoding. Specifically we use their lossy condenser construction based on Reed-Solomon codes, which can be transformed into a permutation condenser. Then we use their general construction of an extractor based on the repeated application of a condenser.
The construction is described in this section. For completeness, we reproduce most of the proof here, except the results that are used exactly as stated in \cite{GUV09}. 

It is also worth mentioning that to obtain metric uncertainty relations, we want strong extractors. Even though \citeN{GUV09} mostly talk about weak extractors in their presentation, it is simple to convert their extractors into strong ones. In this section, we describe all the condensers and extractors as strong.
  
\begin{definition}[Condenser]
\label{def:condenser}
A function $C : \{0,1\}^n \times S \to \{0,1\}^{n'}$ is an $(n,k) \to_{\e} (n', k')$ \emph{condenser} if for every $X$ with min-entropy at least $k$, $C(X, U_S)$ is $\e$-close to a distribution with min-entropy $k'$ when $U_S$ is uniformly distributed on $S$. A condenser $C$ is \emph{strong} if $(U_S, C(X, U_S))$ is $\e$-close to $(U_S, Z)$ for some random variable $Z$ such that for all $y \in S$, $Z|_{U_S = y}$ has min-entropy at least $k$.

A condenser is \emph{explicit} if it is computable in polynomial time in $n$.
\end{definition}
\begin{myremark}[]
The set $S$ is usually of the form $\{0,1\}^d$ for some integer $d$. Here, it is convenient to take sets $S$ not of this form to obtain permutation extractors. Note also that an extractor is an $(n,k) \to_{\e} (m,m)$ condenser.
\end{myremark}

\begin{definition}[Permutation condenser]
A family  $\{P_y\}_{y \in S}$ of permutations of $\{0,1\}^n$  is an  $(n,k) \to_{\e} (n', k')$ \emph{strong permutation condenser} if the function $P^C: (x,y) \mapsto P^C_y(x)$ where $P_y^C(x)$ refers to the first $n'$ bits of $P_y(x)$ is an
$(n,k) \to_{\e} (n',k')$ strong condenser. 

A strong permutation condenser is \emph{explicit} if for all $y \in S$, both $P_y$ and $P_y^{-1}$ can be computed in polynomial time.
\end{definition}

The following theorem describes the condenser that will be used as a building block in the extractor construction. It is an analogue of Theorem 7.2 in \cite{GUV09}.
\begin{theorem}
\label{thm:perm-condenser}
For all positive integers $n$ and $\ell \leq n$, as well as $\alpha, \e \in (0,1/2)$, there exists an explicit family of permutations $\{RS_y\}_{y \in S}$ of $\FF_{2^t}^n$ that is an
\[
(nt, (\ell+1)t) \to_{\e} (\ell t, (1-\alpha) \ell t - 4)
\]
strong permutation condenser with $t = \ceil{1/\alpha \cdot \log(24n^2/\e)}$ and $\log |S| \leq t$. Moreover, the functions $(x, y) \mapsto RS_y(x)$ and $(x, y) \mapsto RS^{-1}_y(x)$ can be computed by a circuit of size $O(n \polylog(n/\e))$.
\end{theorem}
\begin{myremark}[]
Note that the input space of the condenser is $\{0,1\}^{nt}$ instead of $\{0,1\}^n$. But one can see such a condenser as a permutation condenser $(P'_y)$ on the smaller space $\{0,1\}^n$ defined by $P'_y(x) = P_y(x0^{t})$ for all $x \in \{0,1\}^n$ where $x0^t$ is obtained by appending $t$ zeros to $x$.
\end{myremark}
\begin{proof}
Set $q = 2^t$ and $\e_0 = \e/6$. Consider the function $C' : \FF_q^n \times \FF_q \to \FF_q^{\ell+1}$ defined by
\[
C'(f, y) = [y, f(y), f(\zeta y), \dots, f(\zeta^{\ell-1} y)]
\]
where $\FF_q^n$ is interpreted as the set of polynomials over $\FF_q$ of degree at most $n-1$ and $\zeta$ is a generator of the multiplicative group $\FF_q^{*}$. First, we compute the input and output sizes in terms of bits. The inputs can be described using $\log |\FF_q^n| = n \log q = nt$ bits, the seed using $\log |\FF_q| = t$ bits and the output using $\log |\FF_q^{\ell+1}| = (\ell + 1) t$. Using Theorem 7.1 in \cite{GUV09}, for any integer $h$, $C'$ is a 
\begin{equation}
\label{eq:rs-condenser}
\left(nt, \log \left( \frac{q^{\ell}-1}{\e_0} \right) \right) \to_{2\e_0} \left(\ell t + t, \log \left( \frac{Ah^{\ell}-1}{2\e_0} \right) \right)
\end{equation}
condenser where $A \eqdef \e_0 q - (n-1)(h-1)\ell$. We now choose $h = \ceil{q^{1-\alpha}}$. 
As $q \geq (4n^2/\e_0)^{1/\alpha}$, we have $A \geq \e_0 q - n^2 h \geq \e_0 q - \e_0 q^{\alpha}/4 \cdot (q^{1-\alpha} + 1) \geq \e_0 q/2$. Thus, we can compute the bounds we obtain on the condenser $C'$:
\[
\log \left( \frac{q^{\ell}-1}{\e_0} \right) = \ell t + \log(1/\e_0) \leq (\ell+1) t
\]
and 
\begin{align*}
\log \left( \frac{Ah^{\ell}-1}{2\e_0} \right) &= \log \left( \frac{Ah^{\ell}}{2\e_0} \right) + \log \left(1 - \frac{1}{Ah^\ell} \right) \\
							&\geq \log (q/4) + \ell \log h - 1 \\
							&\geq t + (1-\alpha) \ell t - 3.
\end{align*}
Plugging these values in equation \eqref{eq:rs-condenser}, we get that $C'$ is a 
\begin{equation}
\label{eq:paramscond}
\left(nt, (\ell+1) t \right) \to_{2\e_0} \left( \ell t + t, (1-\alpha)\ell t + t - 3) \right)
\end{equation}
condenser.

Observe that the seed $y$ is part of the output of the condenser. As we want to construct a strong condenser, we do not consider the seed as part of the output of the condenser. For this, we define $C : \FF_q^n \times \FF_q \to \FF_q^{\ell}$ by $C(f,y) = [f(y), \dots, f(\zeta^{\ell-1} y)]$. Moreover, as will be clear later when we try to build a permutation condenser, we take the seed to be uniform on $S \eqdef \FF_q^{*}$ instead of being uniform on the whole field $\FF_q$. Note that this increases the error of the condenser by at most $2^{-t} \leq \e_0$ (because one can choose $U_{\FF_q^{*}} = U_{\FF_q}$ with probability $1-2^{-t}$). Here and in the rest of this proof, we will be using Doeblin's coupling lemma.


Equation \eqref{eq:paramscond} then implies that if $X$ has min-entropy at least $(\ell + 1)t$ and $U_{S}$ is uniform on $S$, then the distribution of $(U_S, C(X, U_S))$ is $3\e_0$-close to a distribution with min-entropy at least $(1-\alpha) \ell t + t - 3$. Let $Y \in S$ and $Z \in \{0,1\}^{(\ell+1)t}$ be random variables such that $\entHmin(Y,Z) \geq (1-\alpha) \ell t + t - 3$ and $(U_S, C(X, U_S)) = (Y,Z)$ with probability at least $1-3\e_0$. If $Y$ was uniformly distributed on $S$, then it would follow directly that for all $y \in S$, $\entHmin(Z|Y=y) \geq (1-\alpha) \ell t$. However, $Y$ is not necessarily uniformly distributed. We define a new random variable $Z'$ by
\[
Z' = \left\{ \begin{array}{ll}
Z & \textrm{if $Y=U_S$}\\
U' & \textrm{if $Y \neq U_S$}\\
\end{array} \right.
\]
where $U'$ is uniformly distributed on $\{0,1\}^{(\ell+1)t}$ and independent of all the other random variables. We have for any $z \in \{0,1\}^{(\ell+1) t}$ and $y \in S$,
\begin{align*}
\pr{Z'=z|U_S=y} 	&= \frac{1}{\pr{U_S = y}} \left( \pr{Z' = z, Y=y, Y = U_S} +  
					\pr{Z'=z, U_S = y, Y \neq U_s} \right) \\
				&\leq \frac{1}{\pr{U_S = y}} \left(  2^{-(1-\alpha) \ell t - t + 3} + 2^{-(\ell+1)t} \cdot \frac{1}{|S|} \right) \\
				&\leq 2 \cdot 2^{-(1-\alpha) \ell t + 3}.
\end{align*}
Moreover, we have $(U_S, C(X, U_S)) = (U_S,Z')$ with probability at least $1 - 6\e_0$.

We conclude that $C$ is a
\begin{equation}
\label{eq:strongcond}
\left(nt, (\ell+1) t \right) \to_{\e} \left( \ell t, (1-\alpha)\ell t - 4) \right)
\end{equation}
strong condenser. 


To define our permutation condenser, we set the first $n' = \ell t$ bits $RS^C_y(x)$ of $RS_y(x)$ to be $RS^C_y(x) = C(x,y)$. 
We then define the remaining bits by defining $RS^R_y(f) = [f(\zeta^{\ell} y), \dots, f(\zeta^{n-1} y)]$. As $q \geq n-1$ and $\zeta$ is a generator of $\FF^*_q$, the elements $y, \zeta y, \dots, \zeta^{n-1} y$ are distinct provided $y \neq 0$. So for $y \neq 0$, $(RS^C,RS^R)_y(f)$ is the evaluation of the polynomial $f$ of degree at most $n-1$ in $n$ distinct points. Thus, $f \mapsto P_y(f)$ is a bijection in $\FF_q^n$ for all $y \neq 0$. This is why the value $0$ for the seed was excluded earlier.

Concerning the computation of the functions $RS^C_y$ and $RS^R_y$, they only require the evaluation of a polynomial on elements of the finite field $\FF_q$. Computations in the finite field $\FF_q$ can be performed efficiently by finding an irreducible polynomial of degree $\log q$ over $\FF_2$ and doing computations modulo this polynomial. In fact, finding an irreducible polynomial of degree $\log q$ over $\FF_2$ can be done in time polynomial in $\log q$ (see for example \cite{Sho90} for a deterministic algorithm and Corollary 14.43 in the book \cite{GG99} for a simpler randomized algorithm). Since addition, multiplication and finding the greatest common divisor of polynomials in $\FF_2[X]$ can be done using a number of operations in $\FF_2$ that is polynomial in the degrees, we conclude that computations in $\FF_q$ can be implemented in time $O(\polylog(n/\e))$. Moreover, one can efficiently find a generator $\zeta$ of the group $\FF^*_q$. For example, Theorem 1.1 in \cite{Sho92} shows the existence of a deterministic algorithm having a runtime $O(\poly(\log(q))) = O(\polylog(n/\e))$.

To evaluate $RS_y$ at a polynomial $f$, we compute the field elements $y, \zeta y, \dots, \zeta^{n-1} y$, and then evaluate the polynomial $f$ on these points. Using a fast multipoint evaluation, this step can be done in $O(n \polylog n)$ number of operations in $\FF_q$ (see Corollary 10.8 in \cite{GG99}). Moreover, given a list $[f(y), \dots, f(\zeta^{n-1} y)]$ for $y \neq 0$, we can find $f$ by fast interpolation in $\FF_q[X]$ (see Corollary 10.12 in \cite{GG99}). As a result $RS_y^{-1}$ can also be computed in $O(n \polylog n)$ operations in $\FF_q$.
\end{proof}

This condenser will be composed with other extractors, the following lemma shows how to compose condensers.

\begin{lemma}[Composition of strong permutation condensers]
\label{lem:compositioncond}
Let $(P_{1,y_1})_{y_1 \in S_1}$ be an $(n, k) \to_{\e} (n', k')$ strong permutation condenser and $(P_{2,y_2})_{y_2 \in S_2}$ be an $(n', k') \to_{\e} (n'', k'')$ strong permutation condenser. Then $(P_{y})_{y=(y_1, y_2) \in S_1 \times S_2} = (P^C_y, P^R_y)$ where $P^C_{y_1y_2} = P_{2,y_2}^C \comp P^C_{1,y_1} $ and $P^R_{y_1y_2} = (P^R_{2,y_2} \comp P^C_{1, y_1}) \concat P^R_{1, y_1}$ is an
$(n, k) \to_{2\e} (n'', k'')$ strong permutation extractor.
\end{lemma}
\begin{proof}
$P_y$ is clearly a permutation of $\{0,1\}^n$. We only need to check that $P^C$ is a strong condenser. By definition, if $\entHmin(X) \geq k$, $(U_{S_1}, P^C_{1, U_{S_1}}(X))$ is $\e$-close to $(U_{S_1}, Z)$ where $Z|_{U_{S_1} = y_1}$ has min-entropy at least $k'$. Now putting $Z$ into the condenser $P^C_{2}$, we get that for any $y_1$, $(U_{S_2}, P^C_{2, U_{S_2}}(Z_{U_{S_1}})$ is $\e$-close to $(U_{S_2}, Z_2)$ where $Z_2|_{U_{S_2} = y_2}$ has min-entropy at least $k''$ for any $y_2 \in S_2$. Thus, $Z_2|_{U_{S_1}U_{S_2} = y_1y_2}$ has min-entropy at least $k''$. Moreover, by the triangle inequality, we have $\tracedist{ (U_{S_1},U_{S_2}, P^C_{U_{S_1}U_{S_2}}(X)), (U_{S_1},U_{S_2}, Z_2) } \leq 2\e$.
\end{proof}

Next, we present one of the standard extractors that are used as a building block in many constructions. 

\begin{lemma}[``Leftover Hash Lemma'' \cite{ILL89}]
\label{lem:leftoverext}
For all positive integers $n$ and $k \leq n$, and $\e > 0$, there exists an explicit family $(P_y)_{y \in S}$ of permutations of $\{0,1\}^n$ that is an $(n,k) \to_{\e} m$ strong permutation extractor with $\log |S| = \log (2^n - 1)$ and $m \geq k - 2 \log(2/\e)$.
\end{lemma}
\begin{proof}
We view $\{0,1\}^n$ as the finite field $\FF_{2^n}$ and the set $S = \FF_{2^n}^{*}$. We then define the permutation $P_y(x) = x \cdot y$ where the product $x \cdot y$ is taken in the field $\FF_{2^n}$. The family of functions $P_y$ is pairwise independent. Applying the Leftover Hash Lemma \cite{ILL89}, we get that if $Y$ uniform on $\FF_{2^n}$, the distribution of the first $\ceil{k - 2\log(1/\e)}$ bits of $P_Y(X)$ together with $Y$ is $\e$-close to uniform. Now if $U_S$ is only uniform on $\FF_{2^n}^*$, $(U_S, P_{U_S}(X))$ is $\e + 2^{-n}$-close to the uniform distribution. The result follows from the fact that we can suppose $\e \geq 2^{-n}$ (otherwise, $k-2\log(1/\e) \leq 0$ and the theorem is true).
\end{proof}

The problem with this extractor is that it uses a seed that is as long as the input. Next, we introduce the notion of a block source.

\begin{definition}[Block source]
$X = (X_1, X_2, \dots, X_s)$ is a $(k_1, k_2, \dots, k_s)$ block source if for every $i \in \{1, \dots, s\}$ and $x_1, \dots, x_{i-1}$, $X|_{X_1=x_1, \dots, X_{i-1}=x_{i-1}}$ is a $k_i$-source. When $k_1 = \dots = k_s = k$, we call $X$ a $s \times k$ source.
\end{definition}

A block source has more structure than a general source. However, for a source of large min-entropy $k$ (or equivalently with small entropy deficiency $\Delta = n - k$), one does not lose too much entropy by viewing a general source as a block source where each block has entropy deficiency roughly $\Delta$. See \cite[Corollary 5.9]{GUV09} for a precise statement.

\begin{lemma}[Lemma 5.4 in \cite{GUV09}]
\label{lem:baseext}
Let $s$ be a (constant) positive integer. For all positive integers $n$ and $\ell \leq n$ and all $\e > 0$, setting $t = \ceil{8s \log (24n^2 \cdot (4s+1)/\e)}$, there is an explicit family $\{L_y\}_{y \in S}$ of permutations of $\{0,1\}^n$ that is an
\[
(n, 2 \ell t) \to_{\e} \ell t
\]
strong permutation extractor with $\log |S| \leq 2 \ell t / s + t$.
\end{lemma}
\begin{proof}
As the extractor is composed of many building blocks, each generating some error, we define $\e_0 = \e/(4s+1)$ where $\e$ is the target error of the final extractor. The idea is to first apply the condenser $RS$ of Theorem \ref{thm:perm-condenser} with $\alpha = \frac{1}{8s}$  to obtain a string $X' = C(X,U_{\FF^*_{2^t}})$ of length $n' = (2\ell-1) t$ which is $\e_0$-close to a $k'$-source where  
\[
k' = \left(1-\frac{1}{8s}\right)(2\ell-1) t - 4
\]
The entropy deficiency $\Delta$ of this $k'$-source can be bounded by $\Delta = n' - k' \leq \frac{(2\ell -1) t}{8s} + 4$.  
Then, we partition $X' = (X'_1, \dots, X'_{2s})$ (arbitrarily) into $2s$ blocks of size $n'' = \floor{n'/2s}$ or $n'' + 1$ . Using \cite[Corollary 5.9]{GUV09}, $(X'_1, \dots, X'_{2s})$ is $2s \e_0$-close to some $2s \times k''$-source where $k'' = (n'' - \Delta-\log(1/\e_0))$. 

We have $\Delta \leq \ell t/(4s) + 3 \leq \ell t/(3s)$ for $n$ large enough. Thus,
\[
k'' \geq \frac{2 \ell t}{2s} - \frac{\ell t}{3s} - \log(1/\e_0) = \frac{2}{3s} \ell t - \log(1/\e_0).
\]
We can then apply the extractor Lemma \ref{lem:leftoverext} to all the $2s$ blocks using the same seed of size $n''+1$. Note that we can reuse the same seed because we have a strong extractor and the seed is independent of all the blocks. This extractor extracts almost all the min-entropy of the sources. More precisely, if we input to this extractor a $2s \times k''$-source, the output distribution is $2s \e_0$-close to $m$ uniform bits where 
\[
m \geq 2s \cdot (k'' - 2 \log(2/\e_0)) \geq \frac{4}{3} \ell t - 6s \log(2/\e_0) \geq \ell t.
\]

Overall, the output of this extractor is $\e_0 + 2s \e_0 + 2s \e_0 = \e$-close to the uniform distribution on $m$ bits.

It only remains to show that the extractor we just described is strong and can be extended to a permutation. This follows from Lemma \ref{lem:compositioncond} and the fact the condensers (coming from Theorem \ref{thm:perm-condenser} and Lemma \ref{lem:leftoverext}) are strong permutation condensers.
\end{proof}
\begin{myremark}[]
As pointed out in \cite{GUV09}, a stronger version of this lemma (i.e., with larger output) can be proved by using the condenser of Theorem \ref{thm:perm-condenser} and the high min-entropy extractor in \cite{GW97} with a Ramanujan expander (for example, the expander of \cite{LPS88}). This construction can also give a strong permutation extractor. However, using this extractor would slightly complicate the exposition and does not really influence the final extractor construction presented in Theorem \ref{thm:perm-extractor}.
\end{myremark}

The following lemma basically says that the entropy is conserved by a permutation extractor. It is an adapted version of  \cite[Lemma 26]{RRV99}.
\begin{lemma}
\label{lem:manyextractions}
Let $\{P_y\}_{y \in S}$ be a $(n, k) \to_{\e} m$ strong permutation extractor. Let $X$ be a $k$-source, then $(U_S, P^E_{U_S}(X), P^R_{U_S}(X))$ is $2\e$-close to $(U'_{S}, U'_{\{0,1\}^m}, W)$ where $U'_{S}$ and $U'_{\{0,1\}^m}$ are independent and uniformly distributed over $S$ and $\{0,1\}^m$ respectively, and for all $y \in S, z \in \{0,1\}^m$
\[
\entHmin(W| (U'_{S}, U'_{\{0,1\}^m}) = (y,z)) \geq k - m - 1.
\]
\end{lemma}
\begin{proof}
As $\{P_y^E\}$ is a strong extractor, there exist random variables $U'_{S}$ and $U'_{\{0,1\}^m}$ uniformly distributed on $S$ and $\{0,1\}^m$ such that $\pr{(U_S, P^E_{U_S}(X)) \neq (U'_{S}, U'_{\{0,1\}^m})} \leq \e$.
Define $\Gamma = \{ (y, z) \in S \times \{0,1\}^m : \pr{P^E_y(X) = z} < \frac{1}{2} \cdot 2^{-m} \}$. We have for every $(y,z) \notin \Gamma $ and $x \in \{0,1\}^{n-m}$,
\begin{align*}
\pr{P^R_y(X) = x | P^E_y(X) = z} 	&\leq \frac{\pr{P^R_y(X) = x, P^E_y(X) = z}}{2^{-m-1}} \\
							&\leq 2^{m+1} \pr{X = P^{-1}_y(x,z)} \\
							&\leq 2^{-(k - m - 1)}.
\end{align*}
We then show that $\pr{(U_S, P^E_{U_S}) \in \Gamma} \leq \e$. Using the fact that $\{P^E_y\}$ is a strong extractor, we have
\[
\left|\pr{U'_S, U'_{\{0,1\}^m} \in \Gamma} - \pr{(U_S, P^E_{U_S}) \in \Gamma} \right| \leq \e.
\]
But recall that, by definition of $\Gamma$, $\pr{(U_S, P^E_{U_S}) \in \Gamma} < \frac{1}{2} \pr{U_S, U_{\{0,1\}^m} \in \Gamma}$, so we get
\[
\pr{(U_S, P^E_{U_S}) \in \Gamma} \leq \e.
\]

Finally we define
\[
W = \left\{ \begin{array}{ll}
P^R_{U_S}(X) & \textrm{if $(U_S, P^E_{U_S}(X)) \notin \Gamma $}\\
U^* & \textrm{if $(U_S, P^E_{U_S}(X)) \in \Gamma $}\\
\end{array} \right.
\]
where $U^*$ is uniform on $\{0,1\}^{n-m}$ and independent of all other random variables. We conclude by observing that with probability at least $1-2\e$, we have $(U_S, P^E_{U_S}(X)) = (U'_{S}, U'_{\{0,1\}^m})$ and $P^R_{U_S}(X) = W$.
\end{proof}

 
We then combine these results to obtain the desired extractor. The proof of the following theorem closely follows  \cite[Theorem 5.10]{GUV09} but using the lossy condenser presented in Theorem \ref{thm:perm-condenser} and making small modifications to obtain a permutation extractor.
\begin{theorem}
\label{thm:perm-ext1}
For all integers $n \geq 1$, all $\e \in (0,1/2)$, and all $k \in \left[ 200 \ceil{200 \log(24 n^2/\e)}, n \right]$ there is an explicit $(n, k) \to_{\e} \floor{k/4}$ strong permutation extractor $\{P_y\}_{y \in S}$ with $\log |S| \leq 200 \ceil{200 \log(24 n^2/\e)}$. Moreover, the function $(x,y) \mapsto P_y(x)$ can be computed by circuit of size $O(n \polylog(n/\e))$.
\end{theorem}
\begin{proof}
If $n \leq 2 \cdot 10^6$, we can use the extractor of Lemma \ref{lem:baseext} with $s = 200$ and $\ell \geq 1$ such that $2 \ell t \leq k \leq 2 (\ell+1) t$. This gives an extractor whose seed has size $\frac{k}{200} \leq 10^4 \leq 200 \ceil{200 \log(24 n^2/\e)}$ and that extracts $\ell t \geq \frac{1}{4} \cdot 2 (\ell+1) t \geq \frac{k}{4}$ bits, so the statement still holds true. In the rest of the proof, we assume $n > 2 \cdot 10^6$.

\comment{With this condition on $k$ we have $k \geq 2t = 8 \cdot 200 \cdot \log(24n^2 \cdot 801/\e)$.}

The idea of the construction is to build for an integer $i \geq 0$ an explicit $(n, 2^i \cdot 8d) \to_{\e} 2^{i-1} \cdot 8d$ using $d$ bits of seed by induction on $i$. Fix $t(\e) = \ceil{200\log( 24n^2/\e)}$ and $d(\e) = 200 t(\e)$. The induction hypothesis for an integer $i \geq 0$ is as follows: For all integers $i' \leq i$ and $n$ and $\e > 0$, there is an explicit 
\[
(n, 2^{i'} \cdot 8d(\e)) \to_{\e} 2^{i'-1} \cdot 8d(\e)
\]
strong permutation extractor with seed size $d(\e)$. This extractor is called $\{P^{(i)}_y\}_{y \in S_i}$. 

For both $i=0$ and $i=1$, we can use the extractor of Lemma \ref{lem:baseext} with $s = 20$. 
For $i \in \{0,1\}$, this gives an extractor with seed $\frac{2^i \cdot 8d(\e/81)}{20} + t \leq \frac{16}{20} d(\e) + \frac{16}{20} 200 \ceil{200 \log(81)} \leq d(\e)$.

We now show for $i \geq 2$ how to build the extractor $\{P^{(i)}_y\}$ using the extractors $\{P^{(i')}_y\}$ for $i' < i$. Using the induction hypothesis, we construct the following extractor, which will be applied four times to extract the necessary random bits to prove the induction step. The choice of the form of the min-entropy values will become clear later. Set $\e_0 = \e/20$.
\comment{Probably the most natural thing would be to apply the extractor of $i-1$ just two times. We can do the first step by putting $2d$ entropy to generate the seed and $4d$ entropy to to apply the extractor $i-1$ to it, but the problem is that in the following step, we would only have only $6d$ bits of entropy which is exactly what is need so we cannot do it because there will be some losses. One solution is to use the expander based construction of GW}
\begin{center}
\begin{figure}[t]
\begin{center}
\includegraphics[width=.7\textwidth]{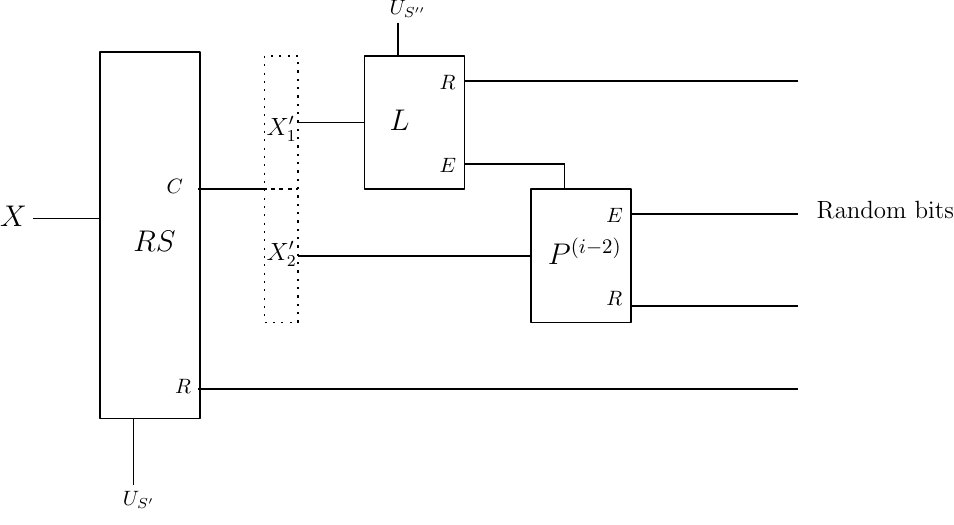}
\end{center}
\caption{The extractor $Q$ is obtained by first applying the condenser of Theorem \ref{thm:perm-condenser} and decomposing the output into two parts. The Leftover Hash Lemma extractor (Lemma \ref{lem:baseext}) is applied to the first half and its output is used as a seed for the extractor $\{P^{(i-2)}_y\}$ coming from the induction hypothesis.}
\label{fig:extractorQ}
\end{figure}
\end{center}

\begin{claim}[]
There exists an
\[
(n, 2^i \cdot 4.5d(\e_0)) \to_{5\e_0} 2^{i} \cdot d(\e_0)
\]
strong permutation extractor $\{Q_y\}_{y \in T}$ with seed size $\log |T| \leq \frac{d(\e_0)}{8}$. 
\end{claim}
To prove the claim, we start by applying the condenser of Theorem \ref{thm:perm-condenser} with $\alpha = 1/200$ and $\e = \e_0$ (so we use a seed of size $t(\e_0)$). The output $X'$ of size at most $2^i \cdot 4.5d(\e_0)$ is then $\e_0$-close to having min-entropy is at least $(1-\alpha)2^i \cdot 4.5 d(\e_0) - t(\e_0)$. The entropy deficiency of this distribution is $\alpha 2^i \cdot 4.5 d(\e_0) + \frac{d(\e_0)}{200} \leq \frac{2^i \cdot 4.5d(\e_0)}{100}$. We then divide $X'$ into two equal blocks $X' = (X'_1, X'_2)$, and we know that it is $2 \e_0$ close to being a $2 \times k'$-source for
\[
k' = \frac{2^{i} \cdot 4.5 d(\e_0)}{2} - \frac{2^{i} \cdot 4.5 d(\e_0)}{100} - \log(1/\e_0) \geq \left(\frac{49}{100} \cdot 2^{i} \cdot 4.5 - \frac{1}{200}\right) d(\e_0)
\]
as $\log(1/\e_0) \leq t(\e_0) = \frac{d(\e_0)}{200}$.
For the extractors we will apply next to this source, we should note that $k' \geq 2d(\e_0)$ and that $2^{i} \cdot 4 d(\e_0) \leq k' < 2^{i} \cdot 8d(\e_0)$.

We now apply the extractor of Lemma \ref{lem:baseext} to $X'_1$ (viewed as a $2d(\e_0)$-source) using a seed of size $\frac{2d(\e_0)}{20}$ and obtaining $X''$ that is $\e_0$ close to uniform on $d(\e_0)$ bits. We then use the extractor $\{P^{(i-2)}_y\}$ obtained by induction for $i-2$ to the $X'_2$ (of size $2^i \cdot 4.5d(\e_0) \leq n$) with seed $X''$ (of size $d(\e_0)$): it is an $(n,2^{i-2} \cdot 8d(\e_0)) \to_{\e_0} 2^{i} \cdot d(\e_0)$ permutation extractor.

The construction is illustrated in Figure \ref{fig:extractorQ}. Note that the number of bits of the seed is $\log |T| \leq t(\e_0) + \frac{2d(\e_0)}{20} \leq \frac{d(\e_0)}{8}$. This concludes the proof of the claim.


\begin{center}
\begin{figure}[t]
\begin{center}
\includegraphics[width=.7\textwidth]{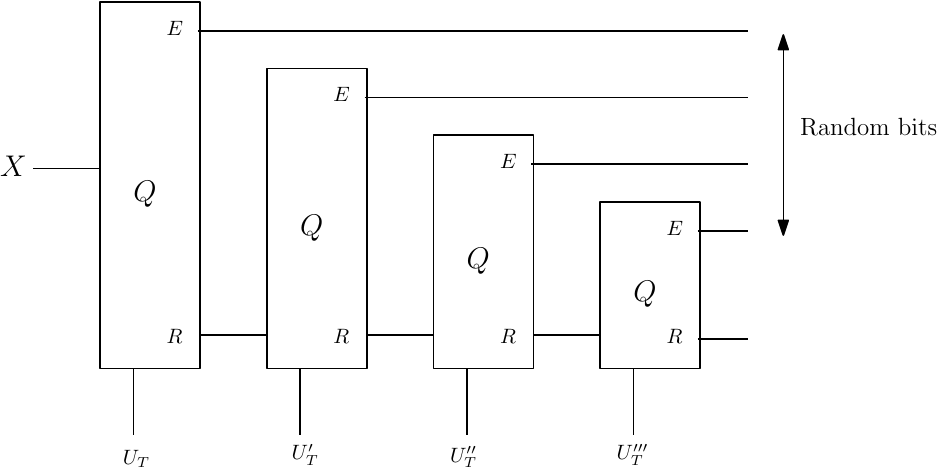}
\end{center}
\caption{The permutation extractor $\{Q_y\}$ described in the claim is applied four times with independent seeds in order to extract $2^{i-1} \cdot 8 d(\e)$ random bits.}
\label{fig:repeatextractor}
\end{figure}
\end{center}

The source $X$ we begin with is a $2^i \cdot 8 d(\e)$-source. But we have $2^i \cdot 8d(\e) \geq 2^{i} \cdot 8d(\e_0) - 2^i \cdot 8 \cdot 200^2 \log 20 \geq 2^{i} \cdot 4.5 d(\e_0)$ so that we can apply the permutation extractor $(Q_y)_{y \in T}$ of the claim. We obtain $Q^{E}_{U_T}(X)$ which is $\e_0$-close to $2^{i} \cdot d(\e_0)$ random bits. As $Q^E$ is part of a permutation extractor, the remaining entropy is not lost: it is in $Q^R_{U_T}(X)$. More precisely, applying Lemma \ref{lem:manyextractions}, we get $Q^R_{U_T}(X)$ is $\e_0$-close to a source of min-entropy at least $2^{i} \cdot 8d(\e) - 2^{i} \cdot d(\e_0) - 1$. As $2^{i} \cdot 8d(\e) - 2^{i} \cdot d(\e_0) - 1 \geq 2^i \cdot 4.5d(\e_0)$, we can apply the extractor $(Q)_{y \in T}$ of the claim to this source. Note that the input size has decreased but as mentioned earlier this only makes it easier to extract random bits as one can always encode in part of the input space. To apply $Q$, we use a fresh new seed that outputs a bit string that is close to uniform on $2^{i-3} \cdot 8d(\e_0)$ bits and the remaining entropy can be found in the $R$ register. We apply this procedure four times in total as shown in Figure \ref{fig:repeatextractor}. Note that the reason we can apply it four times is that at the last application $2^i \cdot 8d(\e) - 3 \cdot 2^{i-3} \cdot 8 d(\e_0) - 3 \geq 2^i \cdot 4.5d(\e_0)$. As the extractor $(Q_y)_{y \in T}$ has error at most $5\e_0$, the total error is bounded by $20\e_0 = \e$.

\comment{The reason all these inequalities hold is that $2^i \cdot 8d(\e) - 3 \cdot 2^{i-3} \cdot 8 d(\e_0) - 3 \geq 2^i \cdot 5d(\e_0) - 2^i \cdot 8 \cdot 200^2 \log 20$. In fact we want $\frac{24 n^2}{\e_0} \geq 20^{16}$, which is garenteed if $n \geq \sqrt{\frac{1}{20 \cdot 24} 20^{16}}$. So we basically want $n \geq 2 \cdot 10^6$.}

We thus obtain an
\[
(n, 2^{i} \cdot 8d(\e)) \to_{\e} 4 \cdot 2^{i-3} \cdot 8d(\e_0)
\]
strong permutation extractor with seed set $S = T^4$ so that $\log |S| \leq 4 \cdot \frac{d(\e_0)}{8} \leq d(\e)$.
\end{proof}

By a repeated application of the previous theorem, we can extract a larger fraction of the min-entropy.
\newtheorem{thm-perm-extractor}{Theorem \ref{thm:perm-extractor}}
\begin{thm-perm-extractor}[]
For all (constant) $\delta \in (0,1)$, there exists $c > 0$, such that for all positive integers $n$, all $k \in [c \log(n/\e), n]$, and all $\e \in (0,1/2)$, there is an explicit $(n, k) \to_{\e} (1-\delta) k$ strong permutation extractor $\{P_y\}_{y \in S}$ with $\log |S| = O(\log(n/\e))$. Moreover, the functions $(x,y) \mapsto P_y(x)$ and $(x,y) \mapsto P^{-1}_y(x)$ can be computed by circuits of size $O(n \polylog(n/\e))$.
\end{thm-perm-extractor}
\begin{proof}
We start by applying the extractor of Theorem \ref{thm:perm-ext1}. We extract part of the min-entropy of the source and the remaining min-entropy is in the $R$ system (Lemma \ref{lem:manyextractions}). This min-enrtopy can be extracted using once again the extractor of Theorem \ref{thm:perm-ext1}. After $O(\log(1/\delta))$ applications of the extractor, we obtain the desired result.
\end{proof}

%
%
%
%
\section{Impossibility of locking using Pauli operators}
\label{sec:app-pauli}
The objective of this section is to give an example of a construction that is not a locking scheme to illustrate what is needed to obtain a locking scheme. The $2 \times 2$ Pauli matrices are the four matrices $\{\1, \sigma_x, \sigma_z, i\sigma_x\sigma_z\}$ where
\[
\sigma_x =
\left( \begin{array}{cc}
0 & 1  \\
1 & 0  \\
\end{array} \right) 
\qquad \textrm{ and } \qquad 
\sigma_z =
\left( \begin{array}{cc}
1 & 0  \\
0 & -1  \\
\end{array} \right).
\]
For bit strings $u,v \in \{0,1\}^n$, we define the unitary operation $\sigma_x^u \sigma_z^v$ on $\left(\CC^{2}\right)^{\otimes n}$ by
\[
\sigma_x^u \sigma_z^v = \sigma_x^{u_1} \sigma_z^{v_1} \otimes \dots \otimes \sigma_x^{u_n} \sigma_z^{v_n}.
\]
It was shown in \cite{AMTW00} that one can encrypt an $n$-qubits state $\ket{\psi}$ perfectly using a key $(U,V)$ of $2n$ bits. To encrypt $\ket{\psi}$, one simply applies $\sigma_x^{U} \sigma_z^{V}$ to $\ket{\psi}$. This can be thought of as a quantum version of one-time pad encryption. Of course, this encryption scheme also defines a $(0,0)$-locking scheme, but the size of the key is $2n$ bits. Recall that we want to use the assumption that the message is random to reduce the key size to $O(\polylog(n))$ bits.

\citeN{AS04} showed that to achieve approximate encryption, it is sufficient to choose the key uniformly at random from a subset $S \subseteq \{0,1\}^{2n}$ of size only $O(n^2 2^n)$. Such pseudorandom subsets are called $\delta$-biased sets and have also been used to construct entropically secure encryption schemes \cite{DS05,DD10}. For example, \citeN{DD10} showed that it is possible to encrypt a uniformly random state by applying $\sigma_x^{U}\sigma_z^{V}$ where $(U,V)$ is chosen uniformly from a set $S \subset \{0,1\}^n$ of size $O(n^2)$ (see \cite{DS05,DD10} for a precise definition of entropic security). Such a scheme can seem like a good candidate for a locking scheme. The following proposition shows that this encyption scheme is far from being $\e$-locking. Note that this also shows that the notion of entropic security defined in \cite{Des09,DD10} is weaker than the definition of locking.

\begin{proposition}
\label{prop:subset-pauli}
Consider an $\e$-locking scheme $\cE$ of the form $\cE(x,k = (u,v)) = \sigma_x^u \sigma_z^v \ket{x}$ where the message $x \in \{0,1\}^n$ and the key $u,v \in \{0,1\}^{n}$ (see Definition \ref{def:locking}). Suppose the secret key $K$ is chosen uniformly from a set $S \subseteq \{0,1\}^{2n}$. Then $|S| \geq (1-\e) 2^n$.
\end{proposition}
\begin{proof}
Let $X$ be the message ($X$ is uniformly distributed over $\{0,1\}^n$) and $(U,V)$ be the key. The key is uniformly distributed on $S$.
We show that a measurement in the computational basis gives a lot of information about $X$ . Let $I$ be the outcome of measuring $\cE(X,K)$ in the computational basis. We have for $x, i \in \{0,1\}^n$,
\begin{align*}
\pr{X = x| I = i} 	&= \pr{I = i | X=x} \\
			&= \frac{1}{|S|} \sum_{(u,v) \in S} \left| \bra{i} \sigma_x^u \sigma_z^v \ket{x} \right|^2.
\end{align*}
Observing that the term $\left| \bra{i} \sigma_x^u \sigma_z^v \ket{x} \right|^2 \in \{0,1\}$, we have that for any fixed $i$, there are at most $|S|$ different values of $x$ for which $\pr{X=x|I=i} > 0$. Thus, defining $T = \{x \in \{0,1\}^n : \pr{X=x|I=i} = 0\}$, we have
\[
\tracedist{p_{X|\event{I=i}}, p_X} \geq \pr{X \in T} - \pr{X \in T | I = i} = \frac{|T|}{2^n} = 1 - \frac{|S|}{2^n}.
\]
By the definition of a locking scheme, we should have 
\[
\tracedist{p_{X|\event{I=i}}, p_X} \leq \e
\]
which concludes the proof.
\end{proof}

\section*{Acknowledgments}

OF would like to thank Luc Devroye for many discussions on concentration inequalities and in particular about Lemma \ref{lem:avgconc}. We would also like to thank Tsuyoshi Ito, Marius Junge, Ivan Savov,
Gideon Schechtman, Stanislaw Szarek and Andreas Winter for helpful conversations, as well as the Mittag-Leffler Institute for its  
hospitality. We would also like to thank the anonymous referees for many suggestions that  improved the presentation. In particular, we thank a referee for suggesting the use of a concentration result on the product of spheres (Lemma \ref{lem:conc-product-sphere}).


\bibliographystyle{acmsmall}
\bibliography{norm_emb}

\begin{thebibliography}{}

\bibitem[\protect\citeauthoryear{Ahlswede and Dueck}{Ahlswede and
  Dueck}{1989}]{AD89}
{\sc Ahlswede, R.} {\sc and} {\sc Dueck, G.} 1989.
\newblock Identification via channels.
\newblock {\em IEEE Trans. Inform. Theory\/}~{\em 35,\/}~1, 15 --29.

\bibitem[\protect\citeauthoryear{Ambainis}{Ambainis}{2010}]{Amb09}
{\sc Ambainis, A.} 2010.
\newblock {Limits on entropic uncertainty relations}.
\newblock {\em Quantum Inf. Comput.\/}~{\em 10,\/}~9 \& 10, 848--858.

\bibitem[\protect\citeauthoryear{Ambainis, Mosca, Tapp, and de~Wolf}{Ambainis
  et~al\mbox{.}}{2000}]{AMTW00}
{\sc Ambainis, A.}, {\sc Mosca, M.}, {\sc Tapp, A.}, {\sc and} {\sc de~Wolf,
  R.} 2000.
\newblock {Private quantum channels}.
\newblock In {\em Proc. ACM STOC}. 547--553.

\bibitem[\protect\citeauthoryear{Ambainis and Smith}{Ambainis and
  Smith}{2004}]{AS04}
{\sc Ambainis, A.} {\sc and} {\sc Smith, A.} 2004.
\newblock {Small Pseudo-random Families of Matrices: Derandomizing Approximate
  Quantum Encryption}.
\newblock In {\em APPROX-RANDOM}. LNCS Series, vol. 3122. 249--260.

\bibitem[\protect\citeauthoryear{Aubrun, Szarek, and Werner}{Aubrun
  et~al\mbox{.}}{2010}]{ASW10}
{\sc Aubrun, G.}, {\sc Szarek, S.}, {\sc and} {\sc Werner, E.} 2010.
\newblock {Nonadditivity of R\'enyi entropy and Dvoretzky's theorem}.
\newblock {\em J. Math. Phys.\/}~{\em 51,\/}~2, 022102.

\bibitem[\protect\citeauthoryear{Aubrun, Szarek, and Werner}{Aubrun
  et~al\mbox{.}}{2011}]{ASW10b}
{\sc Aubrun, G.}, {\sc Szarek, S.}, {\sc and} {\sc Werner, E.} 2011.
\newblock {Hastings’s Additivity Counterexample via Dvoretzky’s Theorem}.
\newblock {\em Comm. Math. Phys.\/}~{\em 305}, 85--97.

\bibitem[\protect\citeauthoryear{Audenaert}{Audenaert}{2007}]{Aud07}
{\sc Audenaert, K.} 2007.
\newblock A sharp continuity estimate for the von {Neumann} entropy.
\newblock {\em J. Phys. A - Math. Theor.\/}~{\em 40,\/}~28, 8127.

\bibitem[\protect\citeauthoryear{Ball}{Ball}{1997}]{Bal97}
{\sc Ball, K.} 1997.
\newblock {An elementary introduction to modern convex geometry}.
\newblock {\em Flavors of geometry\/}~{\em 31}, 1--58.

\bibitem[\protect\citeauthoryear{Ballester and Wehner}{Ballester and
  Wehner}{2007}]{BW07}
{\sc Ballester, M.~A.} {\sc and} {\sc Wehner, S.} 2007.
\newblock Entropic uncertainty relations and locking: Tight bounds for mutually
  unbiased bases.
\newblock {\em Phys. Rev. A\/}~{\em 75,\/}~2, 022319.

\bibitem[\protect\citeauthoryear{Bennett and Brassard}{Bennett and
  Brassard}{1984}]{BB84}
{\sc Bennett, C.~H.} {\sc and} {\sc Brassard, G.} 1984.
\newblock {Quantum cryptography: Public key distribution and coin tossing}.
\newblock In {\em Proc. IEEE International Conference on Computers, Systems and
  Signal Processing}.

\bibitem[\protect\citeauthoryear{Bennett, DiVincenzo, Smolin, and
  Wootters}{Bennett et~al\mbox{.}}{1996}]{BDSW96}
{\sc Bennett, C.~H.}, {\sc DiVincenzo, D.}, {\sc Smolin, J.~A.}, {\sc and} {\sc
  Wootters, W.~K.} 1996.
\newblock Mixed-state entanglement and quantum error correction.
\newblock {\em Phys. Rev. A\/}~{\em 54,\/}~5, 3824--3851.

\bibitem[\protect\citeauthoryear{Berta, Fawzi, and Wehner}{Berta
  et~al\mbox{.}}{2012}]{BFW12}
{\sc Berta, M.}, {\sc Fawzi, O.}, {\sc and} {\sc Wehner, S.} 2012.
\newblock Quantum to classical randomness extractors.
\newblock In {\em Proc. CRYPTO}. LNCS Series, vol. 7417. Springer Verlag,
  776--793.

\bibitem[\protect\citeauthoryear{Bialynicki-Birula and
  Mycielski}{Bialynicki-Birula and Mycielski}{1975}]{BM75}
{\sc Bialynicki-Birula, I.} {\sc and} {\sc Mycielski, J.} 1975.
\newblock {Uncertainty relations for information entropy in wave mechanics}.
\newblock {\em Comm. Math. Phys.\/}~{\em 44,\/}~2, 129--132.

\bibitem[\protect\citeauthoryear{Buhrman, Christandl, Hayden, Lo, and
  Wehner}{Buhrman et~al\mbox{.}}{2006}]{BCHLW06}
{\sc Buhrman, H.}, {\sc Christandl, M.}, {\sc Hayden, P.}, {\sc Lo, H.~K.},
  {\sc and} {\sc Wehner, S.} 2006.
\newblock {Security of quantum bit string commitment depends on the information
  measure}.
\newblock {\em Phys. Rev. Lett.\/}~{\em 97,\/}~25, 250501.

\bibitem[\protect\citeauthoryear{Buhrman, Christandl, Hayden, Lo, and
  Wehner}{Buhrman et~al\mbox{.}}{2008}]{BCHLW08}
{\sc Buhrman, H.}, {\sc Christandl, M.}, {\sc Hayden, P.}, {\sc Lo, H.~K.},
  {\sc and} {\sc Wehner, S.} 2008.
\newblock {Possibility, impossibility, and cheat sensitivity of quantum-bit
  string commitment}.
\newblock {\em Phys. Rev. A\/}~{\em 78,\/}~2, 22316.

\bibitem[\protect\citeauthoryear{Buhrman, Cleve, Watrous, and de~Wolf}{Buhrman
  et~al\mbox{.}}{2001}]{BCWW01}
{\sc Buhrman, H.}, {\sc Cleve, R.}, {\sc Watrous, J.}, {\sc and} {\sc de~Wolf,
  R.} 2001.
\newblock {Quantum fingerprinting}.
\newblock {\em Phys. Rev. Lett.\/}~{\em 87,\/}~16, 167902.

\bibitem[\protect\citeauthoryear{Damg{\aa}rd, Fehr, Renner, Salvail, and
  Schaffner}{Damg{\aa}rd et~al\mbox{.}}{2007}]{DFRSS07}
{\sc Damg{\aa}rd, I.}, {\sc Fehr, S.}, {\sc Renner, R.}, {\sc Salvail, L.},
  {\sc and} {\sc Schaffner, C.} 2007.
\newblock {A tight high-order entropic quantum uncertainty relation with
  applications}.
\newblock In {\em Proc. CRYPTO}. LNCS Series, vol. 4622. 360--378.

\bibitem[\protect\citeauthoryear{Damg{\aa}rd, Fehr, Salvail, and
  Schaffner}{Damg{\aa}rd et~al\mbox{.}}{2005}]{DFSS05}
{\sc Damg{\aa}rd, I.}, {\sc Fehr, S.}, {\sc Salvail, L.}, {\sc and} {\sc
  Schaffner, C.} 2005.
\newblock Cryptography in the bounded quantum-storage model.
\newblock In {\em Proc. IEEE FOCS}. 449--458.

\bibitem[\protect\citeauthoryear{Damg\r{a}rd, Pedersen, and
  Salvail}{Damg\r{a}rd et~al\mbox{.}}{2004}]{DPS04}
{\sc Damg\r{a}rd, I.}, {\sc Pedersen, T.~B.}, {\sc and} {\sc Salvail, L.} 2004.
\newblock {On the key-uncertainty of quantum ciphers and the computational
  security of one-way quantum transmission}.
\newblock In {\em Proc. EUROCRYPT}. LNCS Series, vol. 3027. 91--108.

\bibitem[\protect\citeauthoryear{Damg\r{a}rd, Pedersen, and
  Salvail}{Damg\r{a}rd et~al\mbox{.}}{2005}]{DPS05}
{\sc Damg\r{a}rd, I.}, {\sc Pedersen, T.~B.}, {\sc and} {\sc Salvail, L.} 2005.
\newblock A quantum cipher with near optimal key-recycling.
\newblock In {\em Proc. CRYPTO}. LNCS Series, vol. 3621. 494--510.

\bibitem[\protect\citeauthoryear{Dankert, Cleve, Emerson, and Livine}{Dankert
  et~al\mbox{.}}{2009}]{DCEL09}
{\sc Dankert, C.}, {\sc Cleve, R.}, {\sc Emerson, J.}, {\sc and} {\sc Livine,
  E.} 2009.
\newblock {Exact and approximate unitary 2-designs and their application to
  fidelity estimation}.
\newblock {\em Phys. Rev. A\/}~{\em 80,\/}~1, 12304.

\bibitem[\protect\citeauthoryear{Desrosiers}{Desrosiers}{2009}]{Des09}
{\sc Desrosiers, S.~P.} 2009.
\newblock Entropic security in quantum cryptography.
\newblock {\em Quantum Inf. Process.\/}~{\em 8}, 331--345.

\bibitem[\protect\citeauthoryear{Desrosiers and Dupuis}{Desrosiers and
  Dupuis}{2010}]{DD10}
{\sc Desrosiers, S.~P.} {\sc and} {\sc Dupuis, F.} 2010.
\newblock Quantum entropic security and approximate quantum encryption.
\newblock {\em IEEE Trans. Inform. Theory\/}~{\em 56,\/}~7, 3455 --3464.

\bibitem[\protect\citeauthoryear{Deutsch}{Deutsch}{1983}]{Deu83}
{\sc Deutsch, D.} 1983.
\newblock Uncertainty in quantum measurements.
\newblock {\em Phys. Rev. Lett.\/}~{\em 50,\/}~9, 631--633.

\bibitem[\protect\citeauthoryear{DiVincenzo, Horodecki, Leung, Smolin, and
  Terhal}{DiVincenzo et~al\mbox{.}}{2004}]{DHLST04}
{\sc DiVincenzo, D.~P.}, {\sc Horodecki, M.}, {\sc Leung, D.~W.}, {\sc Smolin,
  J.~A.}, {\sc and} {\sc Terhal, B.~M.} 2004.
\newblock {Locking classical correlations in quantum states}.
\newblock {\em Phys. Rev. Lett.\/}~{\em 92,\/}~6, 67902.

\bibitem[\protect\citeauthoryear{Dodis and Smith}{Dodis and Smith}{2005}]{DS05}
{\sc Dodis, Y.} {\sc and} {\sc Smith, A.} 2005.
\newblock {Entropic security and the encryption of high entropy messages}.
\newblock {\em Theory of Cryptography\/}, 556--577.

\bibitem[\protect\citeauthoryear{Doeblin}{Doeblin}{1938}]{Doe38}
{\sc Doeblin, W.} 1938.
\newblock Expos{\'e} de la th{\'e}orie des cha{\i}nes simples constantes de
  markov {\'a} un nombre fini d’{\'e}tats.
\newblock {\em Math{\'e}matique de l'Union Interbalkanique\/}~{\em
  2,\/}~77-105, 78--80.

\bibitem[\protect\citeauthoryear{Dupuis}{Dupuis}{2010}]{Dup09}
{\sc Dupuis, F.} 2010.
\newblock {A decoupling approach to quantum information theory}.
\newblock Ph.D. thesis, {Universit\'e de Montreal}.

\bibitem[\protect\citeauthoryear{Dupuis, Florjanczyk, Hayden, and Leung}{Dupuis
  et~al\mbox{.}}{2010}]{FDHL10}
{\sc Dupuis, F.}, {\sc Florjanczyk, J.}, {\sc Hayden, P.}, {\sc and} {\sc
  Leung, D.} 2010.
\newblock Locking classical information.

\bibitem[\protect\citeauthoryear{Dvijotham and Fazel}{Dvijotham and
  Fazel}{2010}]{DF10}
{\sc Dvijotham, K.} {\sc and} {\sc Fazel, M.} 2010.
\newblock A nullspace analysis of the nuclear norm heuristic for rank
  minimization.
\newblock In {\em ICASSP}. IEEE, 3586--3589.

\bibitem[\protect\citeauthoryear{Dvoretzky}{Dvoretzky}{1961}]{Dvo61}
{\sc Dvoretzky, A.} 1961.
\newblock {Some results on convex bodies and Banach spaces}.
\newblock In {\em Proc. Internat. Sympos. Linear Spaces}. Jerusalem Academic
  Press, 123--160.

\bibitem[\protect\citeauthoryear{Fannes}{Fannes}{1973}]{Fan73}
{\sc Fannes, M.} 1973.
\newblock {A continuity property of the entropy density for spin lattice
  systems}.
\newblock {\em Comm. Math. Phys.\/}~{\em 31,\/}~4, 291--294.

\bibitem[\protect\citeauthoryear{Figiel, Lindenstrauss, and Milman}{Figiel
  et~al\mbox{.}}{1977}]{FLM77}
{\sc Figiel, T.}, {\sc Lindenstrauss, J.}, {\sc and} {\sc Milman, V.~D.} 1977.
\newblock {The dimension of almost spherical sections of convex bodies}.
\newblock {\em Acta Math.\/}~{\em 139,\/}~1, 53--94.

\bibitem[\protect\citeauthoryear{Gavinsky and Ito}{Gavinsky and
  Ito}{2010}]{GI10}
{\sc Gavinsky, D.} {\sc and} {\sc Ito, T.} 2010.
\newblock {Quantum Fingerprints that Keep Secrets}.

\bibitem[\protect\citeauthoryear{Goldreich}{Goldreich}{2008}]{Gol08}
{\sc Goldreich, O.} 2008.
\newblock {\em {Computational complexity: a conceptual perspective}}.
\newblock Cambridge University Press.

\bibitem[\protect\citeauthoryear{Goldreich and Wigderson}{Goldreich and
  Wigderson}{1997}]{GW97}
{\sc Goldreich, O.} {\sc and} {\sc Wigderson, A.} 1997.
\newblock {Tiny families of functions with random properties: A quality-size
  trade-off for hashing}.
\newblock {\em Random Struct. Algor.\/}~{\em 11,\/}~4, 315--343.

\bibitem[\protect\citeauthoryear{Guruswami, Lee, and Razborov}{Guruswami
  et~al\mbox{.}}{2008}]{GLR08}
{\sc Guruswami, V.}, {\sc Lee, J.}, {\sc and} {\sc Razborov, A.} 2008.
\newblock {Almost Euclidean subspaces of \&ell; N 1 via expander codes}.
\newblock In {\em Proc. ACM-SIAM SODA}. Society for Industrial and Applied
  Mathematics, 353--362.

\bibitem[\protect\citeauthoryear{Guruswami, Umans, and Vadhan}{Guruswami
  et~al\mbox{.}}{2009}]{GUV09}
{\sc Guruswami, V.}, {\sc Umans, C.}, {\sc and} {\sc Vadhan, S.} 2009.
\newblock {Unbalanced expanders and randomness extractors from Parvaresh--Vardy
  codes}.
\newblock {\em J. ACM\/}~{\em 56,\/}~4.

\bibitem[\protect\citeauthoryear{Hallgren, Moore, R{\"o}tteler, Russell, and
  Sen}{Hallgren et~al\mbox{.}}{2010}]{HMRRS10}
{\sc Hallgren, S.}, {\sc Moore, C.}, {\sc R{\"o}tteler, M.}, {\sc Russell, A.},
  {\sc and} {\sc Sen, P.} 2010.
\newblock {Limitations of quantum coset states for graph isomorphism}.
\newblock {\em J. ACM\/}~{\em 57,\/}~6.

\bibitem[\protect\citeauthoryear{Harrow, Hayden, and Leung}{Harrow
  et~al\mbox{.}}{2004}]{HHL04}
{\sc Harrow, A.}, {\sc Hayden, P.}, {\sc and} {\sc Leung, D.} 2004.
\newblock Superdense coding of quantum states.
\newblock {\em Phys. Rev. Lett.\/}~{\em 92,\/}~18, 187901.

\bibitem[\protect\citeauthoryear{Hastings}{Hastings}{2009}]{Has09}
{\sc Hastings, M.~B.} 2009.
\newblock {Superadditivity of communication capacity using entangled inputs}.
\newblock {\em Nature Physics\/}~{\em 5,\/}~4, 255--257.

\bibitem[\protect\citeauthoryear{Hayden, Leung, Shor, and Winter}{Hayden
  et~al\mbox{.}}{2004}]{HLSW04}
{\sc Hayden, P.}, {\sc Leung, D.}, {\sc Shor, P.~W.}, {\sc and} {\sc Winter,
  A.} 2004.
\newblock {Randomizing quantum states: Constructions and applications}.
\newblock {\em Comm. Math. Phys.\/}~{\em 250,\/}~2, 371--391.

\bibitem[\protect\citeauthoryear{Hayden, Leung, and Winter}{Hayden
  et~al\mbox{.}}{2006}]{HLW04}
{\sc Hayden, P.}, {\sc Leung, D.}, {\sc and} {\sc Winter, A.} 2006.
\newblock {Aspects of generic entanglement}.
\newblock {\em Comm. Math. Phys.\/}~{\em 265,\/}~1, 95--117.

\bibitem[\protect\citeauthoryear{Hayden and Winter}{Hayden and
  Winter}{2008}]{HW08}
{\sc Hayden, P.} {\sc and} {\sc Winter, A.} 2008.
\newblock {Counterexamples to the maximal $p$-norm multiplicativity conjecture
  for all $p> 1$}.
\newblock {\em Comm. Math. Phys.\/}~{\em 284,\/}~1, 263--280.

\bibitem[\protect\citeauthoryear{Hayden and Winter}{Hayden and
  Winter}{2012}]{HW10}
{\sc Hayden, P.} {\sc and} {\sc Winter, A.} 2012.
\newblock Weak decoupling duality and quantum identification.
\newblock {\em IEEE Trans. Inform. Theory\/}~{\em 58,\/}~7, 4914--4929.

\bibitem[\protect\citeauthoryear{Heath, Strohmer, and Paulraj}{Heath
  et~al\mbox{.}}{2006}]{HSP06}
{\sc Heath, R.~W.}, {\sc Strohmer, T.}, {\sc and} {\sc Paulraj, A.~J.} 2006.
\newblock {On quasi-orthogonal signatures for CDMA systems}.
\newblock {\em IEEE Trans. Inform. Theory\/}~{\em 52,\/}~3, 1217--1226.

\bibitem[\protect\citeauthoryear{Heisenberg}{Heisenberg}{1927}]{Hei27}
{\sc Heisenberg, W.} 1927.
\newblock {{\"U}ber den anschaulichen Inhalt der quantentheoretischen Kinematik
  und Mechanik}.
\newblock {\em Zeitschrift f{\"u}r Physik A Hadrons and Nuclei\/}~{\em
  43,\/}~3, 172--198.

\bibitem[\protect\citeauthoryear{Horodecki, Horodecki, Horodecki, and
  Oppenheim}{Horodecki et~al\mbox{.}}{2005}]{HHHO05}
{\sc Horodecki, K.}, {\sc Horodecki, M.}, {\sc Horodecki, P.}, {\sc and} {\sc
  Oppenheim, J.} 2005.
\newblock Locking entanglement with a single qubit.
\newblock {\em Phys. Rev. Lett.\/}~{\em 94,\/}~20, 200501.

\bibitem[\protect\citeauthoryear{Impagliazzo, Levin, and Luby}{Impagliazzo
  et~al\mbox{.}}{1989}]{ILL89}
{\sc Impagliazzo, R.}, {\sc Levin, L.}, {\sc and} {\sc Luby, M.} 1989.
\newblock {Pseudo-random generation from one-way functions}.
\newblock In {\em Proc. ACM STOC}. 12--24.

\bibitem[\protect\citeauthoryear{Indyk}{Indyk}{2006}]{Ind06}
{\sc Indyk, P.} 2006.
\newblock Stable distributions, pseudorandom generators, embeddings, and data
  stream computation.
\newblock {\em J. ACM\/}~{\em 53,\/}~3, 307--323.

\bibitem[\protect\citeauthoryear{Indyk}{Indyk}{2007}]{Ind07}
{\sc Indyk, P.} 2007.
\newblock {Uncertainty principles, extractors, and explicit embeddings of L2
  into L1}.
\newblock In {\em Proc. ACM STOC}. 615--620.

\bibitem[\protect\citeauthoryear{Indyk and Szarek}{Indyk and
  Szarek}{2010}]{IS10}
{\sc Indyk, P.} {\sc and} {\sc Szarek, S.} 2010.
\newblock Almost-euclidean subspaces of $\ell_1^n$ via tensor products: A
  simple approach to randomness reduction.
\newblock In {\em APPROX-RANDOM}. LNCS Series, vol. 6302. 632--641.

\bibitem[\protect\citeauthoryear{Kashin}{Kashin}{1977}]{Kas77}
{\sc Kashin, B.} 1977.
\newblock {Sections of some finite dimensional sets and classes of smooth
  functions}.
\newblock {\em Izv. Acad. Nauk SSSR\/}~{\em 41}, 334--351.

\bibitem[\protect\citeauthoryear{Koashi and Winter}{Koashi and
  Winter}{2004}]{KW04}
{\sc Koashi, M.} {\sc and} {\sc Winter, A.} 2004.
\newblock Monogamy of quantum entanglement and other correlations.
\newblock {\em Phys. Rev. A\/}~{\em 69,\/}~2, 022309.

\bibitem[\protect\citeauthoryear{K\"onig, Renner, Bariska, and Maurer}{K\"onig
  et~al\mbox{.}}{2007}]{KRBM07}
{\sc K\"onig, R.}, {\sc Renner, R.}, {\sc Bariska, A.}, {\sc and} {\sc Maurer,
  U.} 2007.
\newblock Small accessible quantum information does not imply security.
\newblock {\em Phys. Rev. Lett.\/}~{\em 98,\/}~14, 140502.

\bibitem[\protect\citeauthoryear{K\"onig, Wehner, and Wullschleger}{K\"onig
  et~al\mbox{.}}{2012}]{KWW09}
{\sc K\"onig, R.}, {\sc Wehner, S.}, {\sc and} {\sc Wullschleger, J.} 2012.
\newblock Unconditional security from noisy quantum storage.
\newblock {\em IEEE Trans. Inform. Theory\/}~{\em 58,\/}~3, 1962 --1984.

\bibitem[\protect\citeauthoryear{Kushilevitz and Nisan}{Kushilevitz and
  Nisan}{1997}]{KN97}
{\sc Kushilevitz, E.} {\sc and} {\sc Nisan, N.} 1997.
\newblock {\em {Communication Complexity}}.
\newblock Cambridge University Press.

\bibitem[\protect\citeauthoryear{Ledoux}{Ledoux}{2001}]{Led01}
{\sc Ledoux, M.} 2001.
\newblock {\em {The concentration of measure phenomenon}}.
\newblock American Mathematical Society.

\bibitem[\protect\citeauthoryear{Leung}{Leung}{2009}]{Leu09}
{\sc Leung, D.} 2009.
\newblock {A survey on locking of bipartite correlations}.
\newblock In {\em J. Phys. Conf. Ser.} Vol. 143. 012008.

\bibitem[\protect\citeauthoryear{Lo and Chau}{Lo and Chau}{1997}]{LC97}
{\sc Lo, H.~K.} {\sc and} {\sc Chau, H.~F.} 1997.
\newblock {Is quantum bit commitment really possible?}
\newblock {\em Phys. Rev. Lett.\/}~{\em 78,\/}~17, 3410--3413.

\bibitem[\protect\citeauthoryear{Lubotzky, Phillips, and Sarnak}{Lubotzky
  et~al\mbox{.}}{1988}]{LPS88}
{\sc Lubotzky, A.}, {\sc Phillips, R.}, {\sc and} {\sc Sarnak, P.} 1988.
\newblock {Ramanujan graphs}.
\newblock {\em Combinatorica\/}~{\em 8,\/}~3, 261--277.

\bibitem[\protect\citeauthoryear{Maassen and Uffink}{Maassen and
  Uffink}{1988}]{MU88}
{\sc Maassen, H.} {\sc and} {\sc Uffink, J.~B.~M.} 1988.
\newblock Generalized entropic uncertainty relations.
\newblock {\em Phys. Rev. Lett.\/}~{\em 60,\/}~12, 1103--1106.

\bibitem[\protect\citeauthoryear{Matou{\v{s}}ek}{Matou{\v{s}}ek}{2002}]{Mat02}
{\sc Matou{\v{s}}ek, J.} 2002.
\newblock {\em {Lectures on discrete geometry}}.
\newblock Springer-Verlag.

\bibitem[\protect\citeauthoryear{Mayers}{Mayers}{1997}]{May97}
{\sc Mayers, D.} 1997.
\newblock {Unconditionally secure quantum bit commitment is impossible}.
\newblock {\em Phys. Rev. Lett.\/}~{\em 78,\/}~17, 3414--3417.

\bibitem[\protect\citeauthoryear{Milman}{Milman}{1971}]{Mil71}
{\sc Milman, V.~D.} 1971.
\newblock New proof of the theorem of {A. Dvoretzky} on intersections of convex
  bodies.
\newblock {\em Funct. Anal. Appl.\/}~{\em 5}, 288--295.

\bibitem[\protect\citeauthoryear{Milman and Schechtman}{Milman and
  Schechtman}{1986}]{MS86}
{\sc Milman, V.~D.} {\sc and} {\sc Schechtman, G.} 1986.
\newblock {\em Asymptotic theory of finite dimensional normed spaces}. Lecture
  Notes in Mathematics Series, vol. 1200.
\newblock Springer-Verlag.

\bibitem[\protect\citeauthoryear{Oppenheim and Horodecki}{Oppenheim and
  Horodecki}{2005}]{OH05}
{\sc Oppenheim, J.} {\sc and} {\sc Horodecki, M.} 2005.
\newblock How to reuse a one-time pad and other notes on authentication,
  encryption, and protection of quantum information.
\newblock {\em Phys. Rev. A\/}~{\em 72,\/}~4, 042309.

\bibitem[\protect\citeauthoryear{Pisier}{Pisier}{1989}]{Pis89}
{\sc Pisier, G.} 1989.
\newblock {\em {The volume of convex bodies and Banach space geometry}}.
\newblock Cambridge University Press.

\bibitem[\protect\citeauthoryear{Radhakrishnan, R{\"o}tteler, and
  Sen}{Radhakrishnan et~al\mbox{.}}{2009}]{RRS09}
{\sc Radhakrishnan, J.}, {\sc R{\"o}tteler, M.}, {\sc and} {\sc Sen, P.} 2009.
\newblock {Random Measurement Bases, Quantum State Distinction and Applications
  to the Hidden Subgroup Problem}.
\newblock {\em Algorithmica\/}~{\em 55,\/}~3, 490--516.

\bibitem[\protect\citeauthoryear{Raz, Reingold, and Vadhan}{Raz
  et~al\mbox{.}}{1999}]{RRV99}
{\sc Raz, R.}, {\sc Reingold, O.}, {\sc and} {\sc Vadhan, S.} 1999.
\newblock {Extracting all the randomness and reducing the error in Trevisan's
  extractors}.
\newblock In {\em Proc. ACM STOC}. ACM, 149--158.

\bibitem[\protect\citeauthoryear{Reingold, Vadhan, and Wigderson}{Reingold
  et~al\mbox{.}}{2000}]{RVW00}
{\sc Reingold, O.}, {\sc Vadhan, S.}, {\sc and} {\sc Wigderson, A.} 2000.
\newblock Entropy waves, the zig-zag graph product, and new constant-degree
  expanders and extractors.
\newblock In {\em Proc. IEEE FOCS}. 3 --13.

\bibitem[\protect\citeauthoryear{Robertson}{Robertson}{1929}]{Rob29}
{\sc Robertson, H.~P.} 1929.
\newblock {The uncertainty principle}.
\newblock {\em Physical Review\/}~{\em 34,\/}~1, 163--164.

\bibitem[\protect\citeauthoryear{Russell and Wang}{Russell and
  Wang}{2002}]{RW02}
{\sc Russell, A.} {\sc and} {\sc Wang, H.} 2002.
\newblock {How to fool an unbounded adversary with a short key}.
\newblock In {\em Proc. EUROCRYPT}. LNCS Series, vol. 2332. 133--148.

\bibitem[\protect\citeauthoryear{Shaltiel}{Shaltiel}{2002}]{Sha02}
{\sc Shaltiel, R.} 2002.
\newblock Recent developments in explicit constructions of extractors.
\newblock {\em Bull. EATCS\/}~{\em 77}, 67--95.

\bibitem[\protect\citeauthoryear{Shoup}{Shoup}{1990}]{Sho90}
{\sc Shoup, V.} 1990.
\newblock {New algorithms for finding irreducible polynomials over finite
  fields}.
\newblock {\em Math. Comp.\/}~{\em 54,\/}~189, 435--447.

\bibitem[\protect\citeauthoryear{Shoup}{Shoup}{1992}]{Sho92}
{\sc Shoup, V.} 1992.
\newblock Searching for primitive roots in finite fields.
\newblock {\em Math. Comp.\/}~{\em 58,\/}~197, pp. 369--380.

\bibitem[\protect\citeauthoryear{Spekkens and Rudolph}{Spekkens and
  Rudolph}{2001}]{SR01}
{\sc Spekkens, R.~W.} {\sc and} {\sc Rudolph, T.} 2001.
\newblock {Degrees of concealment and bindingness in quantum bit commitment
  protocols}.
\newblock {\em Phys. Rev. A\/}~{\em 65,\/}~1, 12310.

\bibitem[\protect\citeauthoryear{Szarek}{Szarek}{2006}]{Sza06}
{\sc Szarek, S.} 2006.
\newblock {Convexity, complexity, and high dimensions}.
\newblock In {\em International Congress of Mathematicians}. Vol.~2.
  1599--1621.

\bibitem[\protect\citeauthoryear{Tomamichel, Lim, Gisin, and Renner}{Tomamichel
  et~al\mbox{.}}{2012}]{TLGR12}
{\sc Tomamichel, M.}, {\sc Lim, C.}, {\sc Gisin, N.}, {\sc and} {\sc Renner,
  R.} 2012.
\newblock Tight finite-key analysis for quantum cryptography.
\newblock {\em Nat. Comm.\/}~{\em 3}, 634.

\bibitem[\protect\citeauthoryear{Tomamichel and Renner}{Tomamichel and
  Renner}{2011}]{TR11}
{\sc Tomamichel, M.} {\sc and} {\sc Renner, R.} 2011.
\newblock Uncertainty relation for smooth entropies.
\newblock {\em Phys. Rev. Lett.\/}~{\em 106,\/}~11, 110506.

\bibitem[\protect\citeauthoryear{Tropp}{Tropp}{2004}]{Tro04}
{\sc Tropp, J.} 2004.
\newblock Topics in sparse approximation.
\newblock Ph.D. thesis, University of Texas at Austin.

\bibitem[\protect\citeauthoryear{Vadhan}{Vadhan}{2007}]{Vad07}
{\sc Vadhan, S.} 2007.
\newblock The unified theory of pseudorandomness: guest column.
\newblock {\em ACM SIGACT News\/}~{\em 38,\/}~3, 39--54.

\bibitem[\protect\citeauthoryear{von~zur Gathen and Gerhard}{von~zur Gathen and
  Gerhard}{1999}]{GG99}
{\sc von~zur Gathen, J.} {\sc and} {\sc Gerhard, J.} 1999.
\newblock {\em {Modern computer algebra}}.
\newblock Cambridge University Press.

\bibitem[\protect\citeauthoryear{Wehner and Winter}{Wehner and
  Winter}{2010}]{WW10}
{\sc Wehner, S.} {\sc and} {\sc Winter, A.} 2010.
\newblock {Entropic uncertainty relations---a survey}.
\newblock {\em New J. Phys.\/}~{\em 12}, 025009.

\bibitem[\protect\citeauthoryear{Winter}{Winter}{2004}]{Win04}
{\sc Winter, A.} 2004.
\newblock {Quantum and classical message identification via quantum channels}.
\newblock {\em Quantum Inf. Comput.\/}~{\em 4,\/}~6\&{}7, 563--578.

\bibitem[\protect\citeauthoryear{Wootters and Fields}{Wootters and
  Fields}{1989}]{WF89}
{\sc Wootters, W.~K.} {\sc and} {\sc Fields, B.~D.} 1989.
\newblock Optimal state-determination by mutually unbiased measurements.
\newblock {\em Ann. Physics\/}~{\em 191,\/}~2, 363 -- 381.

\bibitem[\protect\citeauthoryear{Zuckerman}{Zuckerman}{1997}]{Zuc97}
{\sc Zuckerman, D.} 1997.
\newblock Randomness-optimal oblivious sampling.
\newblock {\em Random Struct. Algor.\/}~{\em 11,\/}~4, 345--367.

\end{thebibliography}

\end{document}